\newif\ifsubmission
\submissionfalse

\newif\ifanon
\anonfalse

\newif\ifnotes
 
\notesfalse

\ifsubmission
\documentclass[11pt]{article} 
\usepackage{fullpage} 
\else
\documentclass[11pt]{article}
\usepackage{fullpage}
\fi

\usepackage[utf8]{inputenc}
\usepackage[T1]{fontenc}

\usepackage{amsmath}
\usepackage{amssymb}

\ifsubmission

\usepackage{amsthm}
\else
\usepackage{amsthm}
\fi

\usepackage{comment}
\usepackage[dvipsnames]{xcolor}
\usepackage{mathtools}
\usepackage[utf8]{inputenc}
\usepackage{amsfonts}
\usepackage{graphicx}
\usepackage{mathrsfs}
\usepackage{physics}
\usepackage{soul}
\usepackage{bm} 
\usepackage[ruled,linesnumbered]{algorithm2e}

\usepackage{authblk}
\definecolor{DarkBlue}{RGB}{0,0,150}
\definecolor{DarkRed}{RGB}{150,0,0}
\definecolor{DarkGreen}{RGB}{0,150,0}
\usepackage[colorlinks,linkcolor=Red,citecolor=DarkBlue]{hyperref}
\usepackage[capitalize]{cleveref}
  \crefname{step}{Step}{Steps}
\usepackage[hyperpageref]{backref}

\newcommand{\authnote}[3]{\textcolor{#3}{[{\footnotesize {\bf #1:} { {#2}}}]}}

\newcommand{\yael}[1]{\ifnotes \authnote{Yael}{#1}{Plum} \fi}

\ifsubmission
\else
\newtheorem{theorem}{Theorem}[section]
\newtheorem{proposition}[theorem]{Proposition}

\newtheorem{lemma}[theorem]{Lemma}

\newtheorem{corollary}[theorem]{Corollary}

\newtheorem{definition}[theorem]{Definition}

\newtheorem{remark}[theorem]{Remark}

\fi

\newtheorem{fact}[theorem]{Fact}
\newtheorem{construction}[theorem]{Construction}
\Crefname{importedtheorem}{Imported Theorem}{Imported Theorems}
\Crefname{theorem}{Theorem}{Theorems}
\Crefname{proposition}{Proposition}{Propositions}
\Crefname{claim}{Claim}{Claims}
\Crefname{lemma}{Lemma}{Lemmas}
\Crefname{conjecture}{Conjecture}{Conjectures}
\Crefname{corollary}{Corollary}{Corollaries}
\Crefname{construction}{Construction}{Constructions}
\Crefname{property}{Property}{Properties}

\theoremstyle{definition}

\Crefname{definition}{Definition}{Definitions}
\Crefname{assumption}{Assumption}{Assumptions}
\Crefname{notation}{Notation}{Notations}

\theoremstyle{remark}

\Crefname{question}{Question}{Questions}
\Crefname{remark}{Remark}{Remarks}
\Crefname{comment}{Comment}{Comments}
\Crefname{fact}{Fact}{Facts}

\newcommand{\probcond}[2]{\Pr \left[
\begin{array}{l}#1\end{array} \; : \;
\begin{array}{l}#2\end{array} \right]
}
\newcommand{\Gen}{\mathsf{Gen}}
\newcommand{\GenM}{\mathsf{Gen_{\text{\em W}}}}

\newcommand{\RealW}{\mathsf{RealW}}

\newcommand{\CommitM}{\mathsf{Commit_{\text{\em W}}}}
\newcommand{\Commit}{\mathsf{Commit}}

 \newcommand{\OpenM}{\mathsf{Open_{\text{\em W}}}}
\newcommand{\Open}{\mathsf{Open}}

\newcommand{\TestM}{\mathsf{Test_{\text{\em W}}}}
\newcommand{\Test}{\mathsf{Test}}

 \newcommand{\OutM}{\mathsf{Out_{\text{\em W}}}}
\newcommand{\Out}{\mathsf{Out}}

\def\Eval{\mathsf{Eval}}
\def\Supp{\mathsf{Supp}}

\newcommand{\Invert}{\mathsf{Invert}}

\newcommand{\tcf}{\mathsf{TCF}}
\newcommand{\TCF}{\mathsf{NTCF}}

\def\CNOT{\mathsf{CNOT}}
\def\Check{\mathsf{Check}}

\def\vecx{\mathbf{x}}
\def\vecX{\mathbf{X}}
\def\vecy{\mathbf{y}}
\def\vecd{\mathbf{d}}
\def\vecz{\mathbf{z}}
\def\vecZ{\mathbf{Z}}
\def\vecm{\mathbf{m}}
\def\vech{\mathbf{h}}
\def\vecs{\mathbf{s}}
\def\veca{\mathbf{a}}
\def\vecb{\mathbf{b}}
\def\vecw{\mathbf{w}}
\def\veco{\mathbf{o}}
\def\vecr{\mathbf{r}}
\def\vecA{\mathbf{A}}
\def\vecB{\mathbf{B}}
\def\vecC{\mathbf{C}}
\def\vecN{\mathbf{N}}
\def\vect{\mathbf{t}}
\def\vece{\mathbf{e}}
\def\vecv{\mathbf{v}}
\def\vecu{\mathbf{u}}

\def\ss{\mathbf{ss}}

\newcommand{\C}{C} 

\newcommand{\aux}{\mathsf{aux}}

\newcommand{\out}{\mathsf{out}}

\newcommand{\Good}{\mathsf{Good}}

\newcommand{\BPP}{\mathsf{BPP}}
\newcommand{\BQP}{\mathsf{QPT}}

\newcommand{\ppt}{\mathsf{PPT}}
\newcommand{\PPT}{\mathsf{PPT}}
\newcommand{\QPT}{\mathsf{QPT}}

\newcommand{\ct}{\mathsf{ct}}
\newcommand{\LWE}{\mathsf{LWE}}
\newcommand{\PCP}{\mathsf{PCP}}

\newcommand{\Ext}{\mathsf{Ext}}
\newcommand{\MExt}{\mathsf{WExt}}
\newcommand{\Sim}{\mathsf{Sim}}
\newcommand{\Real}{\mathsf{Real}}

\newcommand{\Ideal}{\mathsf{Ideal}}

\def\cA{{\cal A}}
\def\cB{{\cal B}}
\def\C{{\mathbb{C}}}
\def\cC{{\mathsf{C}}}

\def\cE{{\cal E}}

\def\cH{{\cal H}}

\def\cL{{\cal L}}

\def\cO{{\cal O}}

\def\cR{\mathcal{R}}
\def\cS{{\cal S}}

\def\cV{{\mathsf{V}}}

\def\cX{{\cal X}}

\def\cZ{{\cal Z}}
\def\Y{{\mathcal Y}}
\def\Z{{\mathcal Z}}

\def\W{{\mathcal W}}
\def\P{{\mathcal P}}

\def\D{\mathsf D}

\def\coin{\mathsf{Coin}}

\newcommand{\OPEN}{\mathsf{OPEN}}
\newcommand{\COMMIT}{\mathsf{COMMIT}}

\newcommand{\RegH}{\mathcal{H}}

\newcommand{\QMA}{\mathsf{QMA}}
\newcommand{\NP}{\mathsf{NP}}

\newcommand{\secp}{\lambda}
\newcommand{\poly}{\mathsf{poly}}
\newcommand{\negl}{\mathsf{negl}}

\newcommand{\hk}{\mathsf{hk}}
\newcommand{\pk}{\mathsf{pk}}
\newcommand{\sk}{\mathsf{sk}}

\newcommand{\A}{\mathsf{A}}
\newcommand{\B}{\mathsf{B}}

\newcommand{\rt}{\mathsf{rt}}

\newcommand{\MP}{\mathsf{MP}} 

\newcommand{\brho}{\bm{\rho}}
\newcommand{\bpsi}{\bm{\psi}}

\newcommand{\bsigma}{\bm{\sigma}}
\newcommand{\btau}{\bm{\tau}}
\newcommand{\bpi}{\bm{\pi}}

\newcommand{\tensor}{\otimes}

\DeclareMathOperator*{\E}{\mathbb{E}}

\newcommand{\SimGen}{\mathsf{SimGen}}

\newcommand{\acc}{\mathsf{acc}}
\newcommand{\rej}{\mathsf{rej}}

\newcommand{\Ver}{\mathsf{Ver}}

\def\QMA{\textbf{QMA}}

\title{Classical Commitments to Quantum States}

\ifsubmission
\author{}
\author{
Sam Gunn \inst{1}\and 
Yael Tauman Kalai\inst{2}   \and 
Anand Natarajan \inst{2}   \and
\'{A}gi Vill\'{a}nyi \inst{2} 
}
\institute{UC Berkeley \and MIT}
\author[1]{Sam Gunn\thanks{\texttt{gunn@berkeley.edu}}}
\author[2]{Yael Tauman Kalai\thanks{\texttt{tauman@mit.edu} }}
\author[2]{Anand Natarajan\thanks{\texttt{anandn@mit.edu}}}
\author[2]{\'{A}gi Vill\'{a}nyi\thanks{\texttt{agivilla@mit.edu}}}
\affil[1]{UC Berkeley}
\affil[2]{MIT}
\date{\today}
\else
\ifanon
\author{Anonymous}
\else
\author[1]{Sam Gunn\thanks{\texttt{gunn@berkeley.edu}}}
\author[2]{Yael Tauman Kalai\thanks{\texttt{tauman@mit.edu} }}
\author[2]{Anand Natarajan\thanks{\texttt{anandn@mit.edu}}}
\author[2]{\'{A}gi Vill\'{a}nyi\thanks{\texttt{agivilla@mit.edu}}}
\affil[1]{UC Berkeley}
\affil[2]{MIT}
\date{\today}
\fi
\fi

\begin{document}

\maketitle
    \begin{abstract}
    We define the notion of a classical commitment scheme to quantum states, which allows a quantum prover to compute a classical commitment to a quantum state, and later open each qubit of the state in either the standard or the Hadamard basis.  Our notion is a strengthening of the measurement protocol from  Mahadev (STOC 2018).  We construct such a commitment scheme from the post-quantum Learning With Errors ($\LWE$) assumption, and more generally from any noisy trapdoor claw-free function family that has the distributional strong adaptive hardcore bit property (a property that we define in this work). 
    
    Our scheme is {\em succinct} in the sense that the running time of the verifier in the commitment phase depends only on the security parameter (independent of the size of the committed state), and its running time in the opening phase grows only with the number of qubits that are being opened (and the security parameter).  As a corollary we obtain a classical succinct argument system for $\QMA$ under the post-quantum $\LWE$ assumption.  Previously, this was only known assuming post-quantum secure indistinguishability obfuscation.  As an additional corollary we obtain a generic way of converting any $X$/$Z$ quantum $\PCP$ into a succinct argument system under the quantum hardness of $\LWE$.
   \end{abstract} 

\newpage
\setcounter{tocdepth}{2}
\tableofcontents
\clearpage

\section{Introduction}
A commitment scheme is one of the most basic primitives in classical cryptography, with far reaching applications ranging from zero-knowledge proofs \cite{FOCS:GolMicWig86,BrassardCC88}, identification schemes and signature schemes \cite{C:FiaSha86}, secure multi-party computation protocols \cite{STOC:GolMicWig87,C:ChaDamVan87}, and succinct arguments \cite{FOCS:Micali94}. There is a long history of studying commitments to \emph{classical} information, both in the classical and post-quantum worlds, and recently, Gunn et al.~\cite{GJMZ22} systematically explored commitments to \emph{quantum} states, using quantum messages. In this work, we initiate the formal study of commitments to quantum states using \emph{classical} messages.
Specifically, we study the existence of \emph{classical} commitments to quantum states, where all messages (the commitment and the opening) are classical, and the receiver is a classical machine. 
\yael{I commented out the previous sentenced and replaced it with the following:} This setting models a likely future  where classical devices will have access to powerful (possibly untrusted) quantum devices.  The goal of this work is to provide the foundations needed for these classical devices to use the (untrusted) quantum devices effectively. 

Our major contributions are a definition of a classical commitment to quantum states, including a sensible notion of a classical opening of a committed quantum state; a construction based on the post-quantum Learning With Errors ($\LWE$) assumption; and a construction of a \emph{succinct} commitment to quantum states (analogous to Merkle hashing in the classical setting~\cite{merkle1987digital}), also under post-quantum $\LWE$.\footnote{More generally, our constructions are based on the existence of a (noisy) claw-free trapdoor function
family with a distributional strong adaptive hard-core bit property, which in particular can be instantiated under the $\LWE$ assumption.} As an immediate application, we obtain a succinct classical argument system for $\QMA$ based only on post-quantum hardness of $\LWE$, improving on previous work which required indistinguishability obfuscation~\cite{Bartusek22}. To our knowledge, our work constitutes the first work to define a notion of a binding classical commitment to quantum states, and to give a construction that achieves it.

Our construction builds directly on the seminal \emph{measurement protocol} of Mahadev~\cite{Mah18a}, which was used by her to construct the first classical argument system for $\mathsf{QMA}$. Loosely speaking, a measurement protocol is a way for a classical verifier to request a quantum prover to measure each qubit of a quantum state (of the prover's choice) in the $X$ or $Z$ basis, with the guarantee that the prover's opening must be ``consistent with a quantum state.'' This motivates our definition of a classical \emph{opening} of a quantum state: the receiver should be able to request the sender to open each qubit of the committed state in either the $X$ or $Z$ basis. (One could imagine asking for openings in more general bases, but these two seem to be a desirable minimum.) However, a measurement protocol does not automatically give rise to a commitment, for several reasons.  First, there is a major structural difference: in a measurement protocol, all phases of the protocol---even the keys chosen in the initial setup---may depend on the choice of opening basis! (Indeed, in Mahadev's protocol, the keys consist of either  ``2-to-1'' or ``injective'' claw-free functions depending on the basis to be measured.) This is far from what we would like in a commitment: the initial ``commitment'' phase should be \emph{completely} independent of the basis in which the receiver ultimately chooses to request an opening.

Thus, the first step to building our construction is to convert Mahadev's measurement protocol into something having the syntax of a commitment\footnote{Technically, we do this by always using the ``2-to-1'' mode of the claw-free function. Moreover, we do not even rely on the existence of a dual-mode (as was done by Mahadev~\cite{Mah18a}), and simply use a ``2-to-1'' claw-free family.}, and henceforth we refer to this modified protocol as Mahadev's ``weak'' commitment\footnote{We refer to it as a weak commitment since (as we elaborate on below) it does not have the desired binding property.  }. In the most basic version of this protocol, a quantum sender holding a qubit in state $\ket{\psi}$ interacts with a classical receiver, sending a classical message that commits to $\ket{\psi}$. Later, the sender is requested by the receiver to ``open'' the committed qubit in either the standard or the Hadamard basis. To open, the sender performs an appropriate measurement and returns the outcome, which can be \emph{decoded} by the receiver (using a cryptographic trapdoor), to obtain an outcome from measuring $\ket{\psi}$ in the appropriate basis. 

A commitment scheme must be \emph{binding}, meaning that the sender cannot change their mind about the committed state once the commitment has been sent. It turns out that the modified Mahadev scheme is a ``weak'' commitment because it partially satisfies the binding property: it is binding in the standard basis, but \emph{not at all} binding in the Hadamard basis. In fact, the sender, after committing to $\ket{+}$, can always freely change the committed state to $\ket{-}$ without ever being detected! Relatedly, in the modified Mahadev scheme, the receiver performs a test on the opening in the standard basis case, and only accepts the opening if it is valid, but performs no test in the Hadamard case.

Motivated by this observation, we show that a simple twist on Mahadev's weak commitment is truly binding (in a rigorous sense which we define) in both bases. We elaborate on our binding definition in \Cref{sec:intro:def,sec:overview}, and on our construction in \Cref{sec:intro:construction,sec:overview}, and below only give a teaser.  In our construction, the sender first commits to $\ket{\psi}$ under Mahadev's weak commitment, generating a commitment string $y_0$ and a (multi-qubit) post-commitment state $\ket{\psi_1}$. It then coherently \emph{opens} this state in the Hadamard basis---that is, it executes a unitary version of the opening algorithm, but does not perform the final measurement, instead producing a quantum state $\ket{\psi_1}$. Finally, the sender applies Mahadev's weak commitment \emph{again} to the state $\ket{\psi_1}$, qubit-by-qubit, obtaining a vector of commitment strings $\vec{y}$ and a post-commitment state $\ket{\psi_2}$. The strings $(y_0,\vec{y})$  now constitute a classical commitment to the state $\ket{\psi}$. To open this commitment in the Hadamard basis, the sender simply applies the standard basis opening procedure for the second Mahadev commitment, yielding a string $z$ which the receiver will test and decode using the commitment vector $\vec{y}$. By the standard-basis binding of Mahadev's commitment, we are guaranteed that the decoded outcome from $z$---assuming the test passes---yields the same result as measuring $\ket{\psi_1}$ in the standard basis, and by construction, this gives a Hadamard-basis opening of $\ket{\psi}$, which it can then decode using the commitment string~$y_0$. But how do we open the commitment in the standard basis? It is far from obvious that this is even possible!  For this we exploit specific features of the Mahadev scheme---in particular, the fact that the opening procedure is ``native'':  Opening  in the standard basis constitutes measuring the registers in the standard basis, and opening in the Hadamard basis constitutes measuring the registers in the Hadamard basis.  This fact is useful both to argue that the opening is correct and to prove that the binding property is achieved. 
We note that in our new scheme the verifier tests the validity of both the standard basis opening and the Hadamard basis opening, and decodes both openings using the cryptographic trapdoor.\footnote{We mention that in Mahdadev's scheme, the verifier only tests the validity of the standard basis opening, and this test, as well as the decoding, is done publicly (without the trapdoor).  The verifier uses the trapdoor only to decode the Hadamard basis opening, which it did not test.}

Our basic  construction for a single qubit can be extended to states with any number of qubits to get a \emph{non-succinct} commitment to a quantum state. We next ask whether our commitment scheme can be made \emph{succinct}: can the sender commit to an $\ell$ qubit state, and open to a small number of these qubits, by exchanging much fewer than $\ell$ bits with the receiver? Here, already in the case of ``weak'' commitments, there is a significant technical obstacle with just the \emph{very first message} from the receiver to the sender: openings in Mahadev's scheme can leak information about the secret key, so each committed qubit must use a fresh secret key to maintain any security at all. This means that, already in the initial key-exchange phase, the receiver must send the sender $\geq \ell$ bits. We show that, surprisingly, the ``strong'' binding property of our commitment, together with specific properties of the underlying (noisy) trapdoor claw-free family, allows us to overcome this barrier. Namely, we show that strong binding, together with specific properties of the underlying (noisy) trapdoor claw-free family, implies that the openings do not leak information about the key in our scheme, allowing us to use the same key for all committed qubits. We emphasize that, even to obtain a succinct ``weak'' commitment, or a succinct measurement protocol, the only route we know of using standard (post-quantum) cryptographic assumptions is through our strongly binding commitments! We view this as an interesting indication of the possible usefulness of our strong binding property in further applications.

As a teaser for how exactly the leakage occurs, and how we avoid it, for now we remark that in the Mahadev weak commitment, the adversary can cause the receiver to generate outputs of the form $d' \cdot s$, for known vectors $d'$ of its choice, where $s$ is the secret. This means that the output for sufficiently many qubits may leak the secret $s$. For an honest sender, this would not be an issue because the vectors $d'$ would be obtained by a quantum measurement with unpredictable answers, and thus have high min-entropy. We show that in our scheme, even \emph{dishonest} senders are forced to produce $d'$ with (sufficient) min-entropy, because of the additional tests done in our opening procedure. This is what prevents the outcomes from leaking information about $s$.

Reusing the key directly only gives us a short first message, which yields a ``semi-succinct'' commitment, in which messages from the receiver are short, but messages from the sender are long. 
In fact, this already yields an application of our results: a \emph{fully-succinct} classical argument system for $\QMA$ which is secure assuming post quantum security of $\LWE$. We obtain this by following the template of Bartusek et al.~\cite{Bartusek22}, but replacing their use of Mahadev's measurement protocol with our succinct commitment.
\begin{theorem}[Informal]\label{thm:informal:app1}
 There exists a (classical) succinct interactive argument for $\QMA$ under the post-quantum Learning With Errors ($\LWE$) assumption.\footnote{More generally, assuming the existence of a (noisy) trapdoor claw free function family with a distributional strong adaptive hard-core bit property, which we elaborate on later on.} 
 \end{theorem}
 
This improves on the result of~\cite{Bartusek22} in terms of cryptographic assumptions: they required the assumption of post-quantum indistinguishability obfuscation (iO) to succinctly generate $\ell$ keys for Mahadev's protocol, whereas our protocol only requires the  post-quantum security of $\LWE$. It is currently not known how to deduce post-quantum iO from \emph{any} standard cryptographic assumptions, whereas $\LWE$ is the ``paradigmatic'' post-quantum cryptographic assumption. 

To construct a succinct argument system for $\QMA$, the approach we and \cite{Bartusek22} both follow is to construct a semi-succinct argument system, and then make it fully succinct by composing with (state-preserving) post-quantum interactive arguments of knowledge~\cite{CMSZ,LombardiMS22}. It turns out that the same tools let us construct outright a fully succinct commitment scheme: for this to be meaningful, we imagine that the sender only opens to a small number of qubits chosen by the receiver, rather than to all of the qubits. In classical cryptography, succinct commitments are natural partners of PCPs, as they enable a verifier to delegate the task of checking a PCP to the prover. While quantum PCPs do not currently exist, we hope that our succinct commitment can be paired with a suitable future PCP to design interesting protocols.

\subsection{The Definition}\label{sec:intro:def}

\paragraph{Defining a non-succinct commitment scheme}  Our definition of a (non-succinct) commitment scheme is a natural extension of the classical counterpart. It consists of a key generation algorithm $\Gen$ that takes as input the security parameter $1^\secp$ and a length parameter $1^\ell$ and  outputs a pair of public and secret keys $(\pk,\sk)$; a commit algorithm $\Commit$ that takes as input a public key $\pk$ and an $\ell$-qubit quantum state $\bsigma$ and outputs a classical string $\vecy$ and a post-commitment state $\brho$, where $\vecy$ is the commitment to the quantum state $\bsigma$;\footnote{We note that both the length of $\pk$ and the length of the commitment string $\vecy$ may grow polynomially with the length $\ell$ of the committed state $\bsigma$.} an open algorithm $\Open$ that takes as input the post-commitment state~$\brho$ and a basis choice 
$\vecb=(b_1,\ldots,b_\ell)\in\{0,1\}^\ell$, where $b_i=0$ corresponds to opening the $i$'th qubit in the standard basis and $b_i=1$ corresponds to opening the $i$'th qubit in the Hadamard basis, and outputs an opening $\vecz\in\{0,1\}^{\ell\cdot\poly(\secp)}$; and the final algorithm $\Out$ that takes as input a secret key $\sk$, a commitment string $\vecy$, a basis choice $\vecb\in\{0,1\}^\ell$ and an opening $\vecz$, and 
outputs the measurement result $\vecm\in\{0,1\}^\ell$ or $\bot$ if the opening is rejected.\footnote{We note that in the actual definition we partition this algorithm into two parts:  $\Ver$ and $\Out$ where the former only outputs a bit indicating if the opening is valid or not and the latter outputs the actual opening if valid.  This partition is only for convenience.} 

We mention that the above syntax yields a commitment scheme that is {\em privately verifiable} in the sense that $\sk$ is needed to decode the measurement value $\vecm$ from the opening value $\vecz$. While it would be desirable to construct a commitment scheme that is publicly verifiable, where $\Gen$ only generates a public key $\pk$, and this public key is used by the opening algorithm to generate the output $\vecm$ along with an opening $\vecz$ which can be verified given $\pk$, we believe that this public key variant is impossible to achieve. This impossibility was formalized on the quantum setting (i.e., where the commitment is a quantum state) by \cite{GJMZ22}, and we leave it as an open problem to prove the impossibility in the classical setting.

We require two properties from our commitment scheme:  completeness and binding. We note that for commitments to classical strings it is common to require a {\em hiding} property.  We do not require it since one can easily obtain hiding by committing to the commitment string $\vecy$ using a classical commitment scheme (that is binding and hiding).
\begin{itemize}
    \item {\bf Correctness.}  The correctness property asserts that if an honest committer commits to an $\ell$-qubit state $\bsigma$ then for any basis choice $\vecb\in\{0,1\}^\ell$, the algorithm $\Out$, applied to the opening string $\vecz$ generated by $\Open$, yields an output $\vecm$ whose distribution is statistically close to the distribution obtained by simply measuring $\bsigma$ in the basis $\vecb$. 
    \item {\bf Binding.}  Loosely speaking, the binding property asserts that
   for any (possibly malicious) $\BQP$ algorithm $\Commit^*$ that commits to an $\ell$-qubit quantum state, there is a {\em single} extracted quantum state $\btau$ such that for {\em any} $\BQP$ algorithm $\Open^*$ and {\em any} basis $(b_1,\ldots,b_\ell)$, where $b_i=0$ corresponds to measuring the $i$'th qubit in the standard basis and $b_i=1$ corresponds to measuring it in the Hadamard basis, the output obtained by $\Open^*(b_1,\ldots,b_\ell)$ 
is computationally indistinguishable from measuring $\btau$ in basis $(b_1,\ldots,b_\ell)$, assuming $\Open^*$ is always accepted. We relax the requirement that $\Open^*$ is always accepted, and allow $\Open^*$ to be rejected with probability~$\delta$ at the price of the two distributions being $O(\sqrt{\delta})$-computationally indistinguishable. We elaborate on the binding property in \Cref{sec:overview}. We note that our definition of binding is nontrivial only for senders that are accepted with a high success probability. By repeating the protocol sequentially $O(1/\delta\cdot \log(1/\delta))$ times we can ensure that if all the openings are accepted with probability $\geq \delta$ then a random one of these openings is accepted with probability $1-\delta$.
While this is weaker than classical notions of binding commitments (which apply to any sender that is accepted with non-negligible probability), it is sufficient for constructing a succinct argument system for $\QMA$.
    
\end{itemize}

\paragraph{Comparison with Mahadev's measurement protocol.}
Our commitment scheme is stronger than that of a {\em measurement protocol}, originally considered in \cite{Mah18a} and formally defined in \cite{Bartusek22}. Beyond the syntactic difference, where in a measurement protocol the opening basis must be determined during the key generation phase (and the key generation algorithm takes as input the basis $\vecb\in\{0,1\}^\ell$), our binding property is significantly stronger. 
A measurement protocol guarantees that any (possibly malicious) $\BQP$ algorithm $\Open^*$ must be consistent with an $\ell$-qubit state, but different opening algorithms can be consistent with different quantum states.

\paragraph{Defining a succinct commitment scheme.}  
The syntax for a succinct commitment differs quite substantially from the syntax of a non-succinct commitment described above.
First, $\Gen$ only takes as input the security parameter $1^\secp$ (and does not take as input the length parameter~$1^\ell$); in addition, $\Commit$ is required to output a succinct commitment of size $\poly(\secp)$.  However, there is a more substantial difference which stems from the fact that,  similarly to the non-succinct variant, we require a succinct commitment to have a binding property that asserts that one can extract an $\ell$-qubit quantum state $\btau$ such that the output distribution of any successful opening is indistinguishable from measuring $\btau$. Since in this setting we consider opening algorithms that only open a few of the qubits,  there is no way we can extract an $\ell$-qubit state from such algorithms.  As a remedy, we add an {\em interactive test phase}. This test phase is executed with probability $1/2$, and if executed then at the end of it the verifier outputs $0$ or $1$, indicating accept or reject, and the protocol terminates without further executing the opening phase, since the test protocol destroys the state. We note that Mahadev's measurement protocol has a non-interactive test phase which is executed with probability $1/2$.  In our setting this test phase is {\em interactive}. It is this interactive nature that allows us to extract a large state from a succinct protocol.

\subsection{The Construction}\label{sec:intro:construction}
\paragraph{Our construction: the single qubit case.}
We construct the commitment scheme in stages.  We first construct a {\em single-qubit} commitment scheme; this scheme is inspired by the construction from Mahadev~\cite{Mah18a}. We elaborate on it in \Cref{sec:overview}, but give a very high-level description here. First, let us recall Mahadev's weak commitment for a single qubit. In this scheme, the sender receives a public key that enable it to evaluate a \emph{two-to-one trapdoor claw-free ($\tcf$) function}  $f: \{0,1\} \times \mathcal{X} \to \mathcal{Y}$.\footnote{We mention that under the $\LWE$ assumption we only have a ``noisy'' $\tcf$ function family, which was constructed in \cite{BCMVV18}.  We do not go into this technicality in the introduction and overview sections.} For every image $y \in \mathcal{Y}$, there are exactly two preimages, which have the form $(0, x_0)$ and $(1, x_1)$, where $x_0,x_1\in\{0,1\}^n$, but any such pair (called a ``claw'') is cryptographically hard to find. In Mahadev's scheme, to commit to a qubit in state $\ket{\psi} = \sum_{b \in \{0,1\}} \alpha_b \ket{b}$, the sender first prepares
\[ \sum_{b\in\{0,1\}}  \sum_{x \in \mathcal{X}} \alpha_b\ket{b} \ket{x} \ket{f(b,x)}, \]
and then measures the last register to obtain a random outcome $y$. The resulting state is the $(n+1)$-qubit state
\[ \sum_{b\in\{0,1\}} \alpha_b\ket{b} \ket{x_b}. \]
To open this in the standard basis, the honest sender measures in the standard basis and returns $(b, x_b)$; the receiver checks that $f(b,x_b) = y$, and if so, records a measurement outcome of $b$.
Intuitively, this constitutes a ``binding'' commitment in the standard basis because it is impossible for the sender to know both $x_0$ and $x_1$, and thus impossible to flip between them. To open in the Hadamard basis, the honest sender measures in the \emph{Hadamard} basis; a short calculation shows that the outcome is a random string $d\in\{0,1\}^{n+1}$, where the probability that $d \cdot (1,x_0 \oplus x_1) \equiv 0 \pmod{2}$ is exactly equal to $|\alpha_0 + \alpha_1|^2/2$, the probability that a Hadamard basis measurement on the \emph{original} state $\ket{\psi}$ would have yielded $+$.  The receiver uses the cryptographic trapdoor to compute $d \cdot (1, x_0 \oplus x_1) \pmod{2}$ as the measurement outcome of the opening, and performs \emph{no} test. This is not at all a binding commitment: indeed, the ``commitments'' to a Hadamard basis states $\ket{\pm}$ look like
\[ \ket{\pm} \mapsto \frac{1}{\sqrt{2}} ( \ket{0} \ket{x_0} \pm \ket{1} \ket{x_1}), \]
and one can easily map from one state to the other by applying a Pauli $Z$ operator to the first qubit.

We now describe our modification to convert this weak commitment (denoted $\mathrm{commit}_W$) into a binding commitment: simply apply a Hadamard transform to the post-commitment state, and then weakly commit again to the resulting $n$-qubit state, applying the Mahadev scheme qubit by qubit, with a new $\tcf$ function $f_i$ for each qubit.
\begin{align*}
  \sum_b \alpha_b \ket{b} &\mapsto^{\mathrm{commit}_W \to y_0}   \sum_b \alpha_b \ket{b, x_b}  \\
                          &\mapsto^{H^{\otimes (n+1)}} \sum_{d \in \{0,1\}^{n+1} } \beta_d \ket{d} \\
                          &\mapsto^{\mathrm{commit}_W \to y_1, \dots, y_{n+1}}  \sum_{d} \beta_d \ket{d_1, x'_{1,d_1}}  \dots \ket{d_{n+1}, x'_{n+1, d_{n+1}}}.
\end{align*}
Here, $d_j$ denotes the $j$th bit of $d$, and $x'_{j, b}$ denotes the corresponding preimage of $y_j$ under the $\tcf$ function $f_j$ (so $f_j(b, x'_{j,b}) = y_j$). 

Let us see how to open this commitment. It will be easier to start with the Hadamard basis: to open in this basis, the sender measures their state in the \emph{standard} basis, and returns the string $(d_1, x_1, \dots, d_{n+1}, x_{n+1})$. The receiver checks that each $(d_i, x_i)$ is a preimage of the corresponding $y_i$, and then records the measurement outcome as $(d_1, \dots, d_{n+1}) \cdot (1, x_0 \oplus x_1)$. To open in the standard basis, the sender measures their state in the \emph{Hadamard} basis, obtaining a (long) string $z$, and the receiver converts this into a measurement outcome by applying the Mahadev procedure for the \emph{Hadamard} basis. Specifically, it first splits $z$ into equal blocks of size $n+1$, and applies the Mahadev Hadamard procedure on each block, to get $n+1$ bits $m_1, \dots, m_{n+1}$.  
\begin{align*}
     z&= (z_1, \dots, z_{n+1}) \\
     &\mapsto \;  (m_1 = z_1 \cdot (1, x'_{1,0} \oplus x'_{1,1}), \dots, m_{n+1} = z_{n+1}\cdot (1, x'_{n+1,0} \oplus x'_{n+1,1})) \\ 
\end{align*}
Now, this corresponds to the outcome of opening the weak commitment of $\sum_{d} \beta_d \ket{d}$ in the Hadamard basis. But this state in turn was equal to the Hadamard transform of $\sum_b \alpha_b \ket{b, x_b}$. Thus, the outcomes $m_1, \dots, m_{n+1}$ should look like the outcome of measuring $\sum_b \alpha_b \ket{b, x_b}$ in the standard basis: that is, like a preimage of $y_0$ under the $\tcf$ function $f$! Thus, the receiver tests the outcomes by checking that
\[ f(m_1, \dots, m_{n+1}) = y_0,\]
and if this passes, it records $m_1$ as the measurement outcome. 

At an intuitive level, what makes this commitment scheme binding is that the receiver performs a test in \emph{both} bases. More formally, we show binding in two parts: (1) there exists a qubit state consistent with the openings reported by the sender, and (2) for any two opening algorithms, the openings they generate are statistically indistinguishable. The proof of (1) uses standard techniques from the analysis of Mahadev's protocol---in particular, the ``swap isometry'' as presented in~\cite{Vid20-course}, but the proof of (2) is new to our work. Our arguments are based on the \emph{collapsing} property of the $\tcf$ functions used to generate $y_1, \dots, y_n$ (in the Hadamard basis case), and $y_0$ (in the standard basis case).  (Jumping ahead, we note that in the succinct setting the situation is reversed. 
 We can obtain (2) basically ``for free'' from the non-succinct setting, whereas the proof of (1) incurs most of the technical burden in this work.)

\paragraph{Our construction: multiple qubits, and succinctness.}
From the single-qubit scheme described above, we construct a {\em non-succinct multi-qubit} commitment scheme, by committing qubit-by-qubit, and thus repeating the single-qubit construction $\ell$-times, where $\ell$ is the number of qubits we wish to commit to. This transformation is generic and can be used to convert {\em any} single-qubit commitment scheme into a {\em non-succinct} multi-qubit one.  We emphasize that in the resulting $\ell$-qubit scheme, both the public key and the commitment string grow with $\ell$, since the former consists of $\ell$ public-keys and the latter consists of $\ell$ commitment strings, where each corresponds to the underlying single-qubit scheme.  We then convert this scheme into a succinct commitment scheme.  This is done in two stages: 
\begin{enumerate}
    \item {\bf Stage 1:}  Reuse the same public key, as opposed to choosing $\ell$ independent ones.  Namely, the public key consists of a single public key $\pk$ corresponding the underlying single-qubit commitment scheme.  To commit to an $\ell$-qubit state, commit qubit-by-qubit while using the same public key $\pk$. We refer to such a commitment scheme as {\em semi-succinct} since the public key is succinct but the commitment is not.

    We note that while this construction is generic, the analysis is not.  In general, reusing the same public-key may break the binding property. We prove that if we start with our specific single-qubit commitment scheme then the resulting semi-succinct multi-qubit scheme remains sound.    
    We recall, that as mentioned above, if we start with Mahadev's single qubit weak commitment protocol and convert it into a multi-qubit weak commitment while reusing the same public key, then the resulting measurement protocol becomes insecure. The reason is that a malicious sender may generate openings $d$ in the Hadamard basis that cause the receiver's ``decoding'' outcomes $d \cdot (1 , x_0 \oplus x_1)$ to leak bits of $\sk$---recall that the receiver must use the secret key to decode, as $x_0$ and $x_1$ cannot be computed efficiently without it.  Indeed, the $\tcf$ function family that we (and Mahadev) use is the $\LWE$-based construction due to \cite{BCMVV18}, which has the property that $d \cdot (1 , x_0 \oplus x_1)=d'\cdot s$, where $s$ is a secret key\footnote{In their construction the public key is an $\LWE$ tuple $(A,As+e)$.  The secret key is actually a trapdoor of the matrix~$A$ but revealing the secret~$s$ is sufficient to break security.}  and $d'$ can be efficiently computed from $d$ and $x_0$.   Once enough information about the secret $s$ has been revealed, the scheme is no longer a secure measurement protocol, let alone a secure commitment: with knowledge of $s$, it becomes easy to distinguish the outcomes of the commitment from outcomes of measuring a true quantum state! Thus, to argue the security of our semi-succinct scheme, we must exploit specific properties of our single-qubit scheme.  Indeed, we crucially use the \emph{binding} property of our scheme to show that the openings $z$ reported by a successful sender must always have high min-entropy, which in our construction implies that $d'$ has min-entropy.  We then use a specific property of the underlying $\tcf$ function family from  \cite{BCMVV18}, which we call the ``distributional strong adaptive hardcore bit'' property. Roughly, this property ensures that if the opening $d$ has min-entropy then $d \cdot (1 , x_0 \oplus x_1)$ (which in their construction is equal to $d'\cdot s$) does not reveal information about $\sk$.

    \item {\bf Stage 2:} Convert any semi-succinct commitment scheme into a succinct one.  This part is generic and shows how to convert {\em any} semi-succinct commitment scheme into a succinct one. Our transformation is almost identical to that from \cite{Bartusek22}, who showed how to convert any semi-succinct interactive argument (which is one where only the verifier's communication is succinct, and where the prover's communication can be long) into a fully succinct one.  We elaborate on the  high-level idea behind this transformation in \Cref{sec:overview}. 
\end{enumerate}

\subsection{Applications}

We show how to use our succinct commitment scheme to construct succinct interactive argument for $\QMA$. As a simpler bonus, we also use it show how to  compile a hypothetical quantum PCP in ``$X/Z$  form'' into a succinct interactive argument. For the $X/Z$ PCP compiler the idea is simple:  In the succinct interactive argument the prover first succinctly commits to the $X/Z$ PCP, then the verifier sends its $X/Z$ queries and finally the prover opens the relevant qubits in the desired basis. 
The succinct interactive argument for $\QMA$ is more complicated, and follows the blueprint from \cite{Mah18a,Bartusek22}.  We elaborate on this in \Cref{sec:overview:Mahadev}.

\subsection{Related Works}\label{sec:related}

Our work is inspired by the measurement protocol of Mahadev \cite{Mah18a}, which has the same correctness guarantee as our commitment scheme.  
However, a measurement protocol (as was formally defined in \cite{Bartusek22}) does not require binding to hold; rather it only requires that an opening is consistent with a qubit.  This qubit may be different for different opening algorithms.   Indeed, the measurement protocol of Mahadev, as well as the ones from followup works, are not binding in the Hadamard basis. Mahadev uses this measurement protocol to construct classical interactive arguments for $\QMA$.  Mahadev's measurement protocol, which was proven to be secure under the post-quantum $\LWE$ assumption, is a key ingredient in our construction.

Mahadev's measurement protocol is not succinct.  In a followup work, Bartusek et~al.\ \cite{Bartusek22} constructed a succinct measurement protocol, by using Mahadev's measurement protocol as a key ingredient, and thus obtaining a succinct classical interactive arguments for $\QMA$. However the security of their protocol, and thus the soundness of the resulting $\QMA$ argument, relies on the existence of a post-quantum secure indistinguishable obfuscation scheme (in addition the post-quantum $\LWE$ assumption). We mention that  Chia, Chung and Yamakawa~\cite{TCC:ChiChuYam20} also construct a succinct 
measurement protocol, which they use to obtain a succinct 2-message argument for $\QMA$. However, in their scheme the prover and verifier share a polynomial-sized structured reference string (which requires a trusted setup to instantiate), and their security is heuristic.\footnote{More specifically, their scheme uses a hash function~$h$, and it is proved to be secure when $h$ is modeled as a random oracle, but the \emph{protocol description itself} explicitly requires the code of $h$ (i.e. uses $h$ in a non-black-box way).}

We improve upon these works by constructing a succinct classical commitment scheme for quantum states that guarantees binding (which is a stronger security condition than the one offered by a measurement protocol), based only on the post-quantum $\LWE$ assumption.  As a result, we obtain a succinct classical interactive arguments for $\QMA$, under the post-quantum $\LWE$ assumption.
Our analysis makes use of techniques developed in \cite{Mah18a,Vid20-course,Bartusek22}, in addition to several new ideas that are needed to obtain our results.  

We mention that our work, as well as all prior works mentioned above, require the receiver (a.k.a\ the verifier) to hold a secret key~$\sk$ which is needed to decode the prover's message and obtain the measurement output.  We mention that the recent work of Bartusek et~al.~\cite{Obfuscation_Bartusek} considers the public-verifiable setting, where decoding can be done publicly.  They construct a publicly verifiable measurement protocol in an oracle model, which is used as a building block in their obfuscation of pseudo-deterministic quantum circuits.

So far we only focused on prior work where the verifier (and hence the communication) is classical.  We mention that recently Gunn et~al.~\cite{GJMZ22} defined and constructed a {\em quantum} commitment scheme to quantum states, where \emph{both} parties are quantum.
In their setting, the quantum committer sends a quantum commitment to the receiver, and later opens by sending a quantum opening.
The receiver then applies some unitary operation to recover the committed quantum state. This is in contrast to the classical setting where the receiver is classical and cannot hope to recover the committed quantum state, and instead only obtains an opening in a particular basis (standard or Hadamard).  We mention that the quantum commitment scheme from \cite{GJMZ22} relies on very weak cryptographic assumptions, and in particular, ones that are implied by the existence of one-way functions.

Finally, simultaneously and using different techniques from this work, a succinct argument system for $\QMA$ based on the assumption of quantum Fully Homomorphic Encryption (qFHE) was achieved by~\cite{MNZ24}. While both papers use common techniques from \cite{Bartusek22} to go from semi-succinctness to full succinctness, the core techniques are essentially disjoint. In particular, \cite{MNZ24} does not use commitments to quantum states, but instead directly analyzes the soundness of the KLVY~\cite{kalai2022quantum}  compilation of a particular semi-succinct two-prover interactive proof for $\QMA$. We leave it as an interesting open question for future work whether their result can yield an alternate construction of our primitive of quantum commitments.


\paragraph{Roadmap} We refer the reader to \Cref{sec:overview} for the high-level overview of our techniques, to \Cref{sec:prelim} for all the necessary preliminaries, to \Cref{sec:XZ-commitments} for the formal definition of a succinct and non-succinct commitment scheme, to \Cref{sec:constructions} for the constructions, to \Cref{sec:analysis} for the analysis, and to \Cref{sec:applications} for the applications.

\section{Technical Overview}\label{sec:overview}

In this section we describe the ideas behind our commitment schemes and their applications in more depth yet still informally. 
Our first contribution is defining the notion of a classical commitment scheme to quantum states.
Let us start with the non-succinct version, and in particular the single-qubit case.  As mentioned in the introduction,  such a commitment scheme consists of algorithms
\[
(\Gen,\Commit,\Open,\Out)
\]
where $\Gen$ is a $\PPT$ algorithm that takes as input the security parameter $1^\secp$ and outputs a pair of keys $(\pk,\sk)$; $\Commit$ is a $\BQP$ algorithm that takes as input a public key $\pk$ and a single-qubit quantum state $\bsigma$ and outputs a classical commitment string $\vecy$ and a post-commitment state $\brho$; $\Open$ is a $\BQP$ algorithm  that takes as input the post-commitment state $\brho$ and a bit $b\in\{0,1\}$, where $b=0$ corresponds to a standard basis opening and $b=1$ corresponds to a Hadamard basis opening, and outputs a classical opening $\vecz$; and $\Out$ is a polynomial-time algorithm that takes as input the secret key $\sk$, a commitment string $\vecy$, a basis $b\in\{0,1\}$ and an opening $\vecz$ and it outputs an element in $\{0,1,\bot\}$. 

We require the scheme to satisfy a correctness and a binding property. 
The correctness property is straightforward and was formalized in prior work~\cite{Bartusek22}.  It is the binding property that is tricky to formulate and achieve.

\paragraph{Defining Binding:  the single qubit setting}  In the classical setting, the binding condition asserts that for any poly-size algorithm $\Commit^*$ that generates a commitment $\vecy$ (to some classical string), and for any two poly-size algorithms $\Open^*_1$ and $\Open^*_2$, the probability that they successfully open to different strings is negligible.  In the quantum setting the analogous property is the following:  For any $\BQP$ algorithm $\Commit^*$ that generates a commitment $\vecy$ (to a quantum state), and for $\BQP$ algorithms $\Open^*_1$ and $\Open^*_2$  (that are accepted with probability~$1$) and every basis choice $b\in\{0,1\}$, the output distributions of $\Open^*_1$ and $\Open^*_2$ are statistically close or computationally indistinguishable.  This is indeed one of the properties we require.\footnote{Jumping ahead, we note that our non-succinct commitment scheme achieves statistical closeness and our succinct commitment scheme achieves computational indistinguishability. We mention that Mahadev's scheme \cite{Mah18a}, as well as its successors \cite{Bartusek22}, do not satisfy this property since these schemes offer no binding on the Hadamard basis.} But this property on its own is not enough. We also need to ensure that the opening is consistent with some qubit.
Namely, we require that there exists a $\BQP$ extractor $\Ext$ such that for every $\BQP$ algorithm $\Open^*$ (that is accepted with probability~$1$),  $\Ext$ given black-box access to $\Open^*$ can extract from $\Open^*$ a quantum state~$\btau$ such that for every basis $b\in\{0,1\}$ the output of $\Open^*$ is computationally indistinguishable from measuring $\btau$ in basis~$b$. We mention that this latter condition was formalized in \cite{Bartusek22} as a security property from a measurement protocol.

We construct a commitment scheme that achieves the above two properties. However, to make this definition meaningful we must consider opening algorithms that are accepted with probability smaller than~$1$.  Indeed, we consider opening algorithms that are accepted with probability $1-\delta$ and obtain  $O(\sqrt{\delta})$-indistinguishability in both the requirements above.  We note that we can assume that $\Open^*$ is accepted with probability $1-\delta$ by repeating the commitment protocol $\Omega(1/\delta)$ times (assuming the committer has many copies of the state they wish to commit to).

\paragraph{The multi-qubit setting.}  So far we focused on the single-qubit setting.  When generalizing the definitions to the multi-qubit setting we distinguish between the non-succinct setting and the succinct setting, starting with the former. The syntax can be generalized to the non-succinct multi-qubit setting in a straightforward way by committing and opening qubit-by-qubit. 
Generalizing the binding definition to the multi-qubit setting is a bit tricky.  In particular, recall that we assumed that $\Open^*$ is accepted with high probability when opening in both bases. As mentioned, this is a reasonable assumption since we can require the committer to commit to its state many ($\Omega(1/\delta)$) times, then open half of the commitments in the standard basis and half of them in the Hadamard basis.  If any of them are rejected then output $\bot$ and otherwise, choose a random one that was opened in the desired basis~$b$ and use that as the opening.  Generalizing this to the $\ell$-qubit setting must be done with care to avoid an exponential blowup in~$\ell$.  Clearly, we do not want to assume that for every basis choice $(b_1,\ldots,b_\ell)\in\{0,1\}^\ell$, $\Open^*$ successfully opens in this basis with high probability, since we cannot enforce this without incurring an exponential blowup. Yet, in order for our extractor to be successful, we need to ensure that $\Open^*$ succeeds in opening each qubit in each basis with high probability. To achieve this, without incurring an exponential blowup, we require that $\Open^*$ succeeds with high probability to open all the qubits in the standard basis (i.e., succeeds with $(b_1,\ldots,b_\ell)=(0,\ldots,0)$) and  succeeds with high probability to open all the qubits in the Hadamard basis (i.e., succeeds $(b_1,\ldots,b_\ell)=(1,\ldots,1)$). This can achieved via repetitions, as in the single qubit setting. Specifically, in this setting we ask $1/3$ of the repetitions to be opened in the $0^\ell$ basis, $1/3$ to be opened in the $1^\ell$ basis, and the remaining $1/3$ to be opened in the desired $(b_1,\ldots,b_\ell)$ basis. Jumping ahead, we note that the extractor $\Ext$ uses $\Open^*$ with basis $(b,\ldots,b)$ to extract the state $\btau$. We refer the reader to \Cref{def:binding} for the formal definition. 


\ifsubmission
\else
\paragraph{Our construction for the single qubit case.}

We start by describing our commitment scheme in the single-qubit case.  
Our starting point is Mahadev's \cite{Mah18a} measurement protocol. Her protocol is binding in the standard basis but offers no binding guarantees, and in fact fails to provide any form of binding, when opening in the Hadamard basis. Moreover, in her protocol the opening basis must be determined ahead of time and the public key $\pk$ used to compute the  commitment string depends on this basis.  Specifically, her protocol uses a family of (noisy) trapdoor claw-free functions, where functions can be  generated either in an {\em injective} mode or in a {\em two-to-one} mode. The public key of the commitment scheme consists of a public key corresponding to an injective function if the verifier wishes to open in the standard basis, and corresponds to a two-to-one function if the verifier wishes to open in the Hadamard basis. 

We first notice that it is not necessary to determine the opening basis in the key generation phase. In fact, we show that one can always use the two-to-one mode, irrespective of the basis we wish to open in.  Moreover, we show that this ``dual mode'' property is not needed altogether.  This observation is quite straightforward and was implicitly used in the analysis in prior work \cite{Vid20-course,Bartusek22}.

Our first instrumental idea is that we can obtain binding in both bases if we compose Mahadev's weak commitment twice!
Namely, to commit to a state $\bsigma$, we first apply Mahadev's measurement protocol, denoted by $\CommitM$, to obtain 
\[
(\vecy,\brho)\gets\CommitM(\pk,\bsigma).
\]
As mentioned, this already guarantees binding when opening in the standard basis, but fails to provide binding when opening in the Hadamard basis.
To fix this we make use of the fact that Mahadev's measurement protocol has the property that the $\Open$ algorithm always measures the post-commitment state in either the standard basis or the Hadamard basis.  We apply to the post-commitment state~$\brho$ the unitary that computes Hadamard opening $\Open(\cdot,1)$, which is simply the Hadamard unitary $H^{\tensor{(n+1)}}$, where $n+1$ is the number of qubits in $\brho$ ($n$ being the security parameter associated with the underlying $\TCF$ family), and we commit to the resulting state.
Namely, we compute
\[
    \brho'\gets H^{\tensor{(n+1)}}[\brho]~~\mbox{ and }~~(\vecy',\brho'')\gets \CommitM(\pk',\brho'),
\]
where $\pk$ and $\pk'$ are independent keys,\footnote{Using different and independent public keys $\pk$ and $\pk'$  is important in our analysis.} and where throughout our paper we use the shorthand
\[
U[\brho]=U\brho U^\dagger
\]
to denote the application of a unitary $U$ to a mixed state $\brho$.

To open the commitment in the Hadamard basis, we just need to measure $\brho'$ in the {\em standard} basis.  Binding in the Hadamard basis follows from the fact that $\brho'$ was committed to via the classical string~$\vecy'$, and from the fact that  Mahadev's measurement protocol provides binding in the standard basis. However, it is no longer clear how to open in the standard basis, since the original post-commitment state $\brho$ is no longer available, and has been replaced with $\brho''$.  Here we use the desired property mentioned above, specifically, that algorithm $\Open$  generates a standard basis opening by measuring the state in the standard basis, and generates a Hadamard basis opening by measuring the state in the Hadamard basis. This implies that measuring $\brho$ in the standard basis is equivalent to measuring $\brho'$ in the Hadamard basis.  

The reader may be concerned that we may have lost the binding in the standard basis, since opening in the Hadamard basis is not protected.  But this is not the case, since it is the commitment string $\vecy$ that binds the standard basis measurement, and the commitment string $\vecy'$ that binds the Hadamard basis measurement.

\paragraph{Multi-qubit commitments} One can use this single qubit commitment scheme to commit to an $\ell$-qubit state, by committing qubit-by-qubit.  This results with a long commitment string of size $\ell\cdot\poly(\secp)$ and with a long public key, since the public key consists of $\ell$ public keys $(\pk_1,\ldots,\pk_\ell)$, where each $\pk_i$ is generated according to the single qubit scheme.  As mentioned in the introduction, our main goal is to construct a succinct commitment scheme.  Following the blueprint of \cite{Bartusek22}, we do this in two steps.  We first construct a {\em semi-succinct} commitment scheme where the commitment string is long, but the public-key is succinct.  We then show how to convert the semi-succinct scheme into a fully succinct one.

\paragraph{Semi-succinct commitments}  In our semi-succinct commitment scheme we generate a single key pair $(\pk,\sk)\gets \Gen(1^\secp)$ corresponding to the single-qubit scheme, and simply use $\pk$ to commit to each and every one of the qubits. The question is whether this is sound.
Let us first describe the main issue that comes up when trying to prove soundness, and then we will show how we overcome it. The issue is that our commitment scheme is privately verifiable, and thus a $\BQP$ algorithm $\Open^*$, which produces an opening~$\vecz$, does not know the corresponding output bit $m=\Out(\sk,\vecy,b,\vecz)$ since $\sk$ is needed to compute~$m$. Therefore, perhaps a malicious $\BQP$ algorithm $\Open^*$ can generate $\vecz$ in a way such that $m$ leaks information about $\sk$. In particular, perhaps $\Open^*$ can generate $\ell$ openings $\vecz_1,\ldots,\vecz_\ell$ such that their corresponding outputs  $m_1,\ldots,m_\ell$ completely leak $\sk$.  

Recall that our binding property consists of two parts: The first asserts that for any $\BQP$ algorithm $\Commit^*$ that commits to an $\ell$-qubit state via a classical commitment string $\vecy$, it holds that for any two $\BQP$ opening algorithms $\Open^*_1$ and $\Open^*_2$ and any basis choice $(b_1,\ldots.b_\ell)$, the output distributions produced by these two opening algorithms are (computationally or statistically) close.  In our construction we get {\em statistical closeness}, and hence the closeness holds even if $\sk$ is leaked.  Indeed, the proof of this property in the semi-succinct setting is the same as the proof in the non-succinct setting.  The issue is with the second part: Given $\sk$, the distributions generated by $\Ext^{\Open^*}$ and $\Open^*$ are no longer computationally indistinguishable. Diving deeper into our scheme and its analysis, we note that the standard basis outputs produced by $\Ext^{\Open^*}$  and $\Open^*$ are actually statistically close, and it is the Hadamard basis outputs that are only computationally indistinguishable. 

We next examine the leakage that the decoded messages $m_1,\ldots,m_\ell$ may contain about the secret key, and argue that even given this leakage, the Hadamard basis outputs produced by $\Ext^{\Open^*}$  and $\Open^*$ remain computationally indistinguishable.  To this end, we will need to use additional properties about Mahadev's measurement protocol, and thus recall it in \cref{sec:overview:Mahadev} below. Jumping ahead, we mention that one property that we rely on is the fact that in Mahadev's protocol, $\Out$ does not use the secret key when generating standard basis outputs (and the secret key is only used to generate Hadamard basis outputs). 

Recall that in our commitment scheme, the secret key consists of two parts, $(\sk,\sk')$, since we apply Mahadev's protocol twice (where $\sk$ is for a single qubit state and $\sk'$ is for an $(n+1)$-qubit state). We mention that when opening in the standard basis, the output $m$ can only leak information about $\sk'$.  This is the case since to open in the standard basis, we first use $\sk'$ to generate a standard basis opening $\vecz$ for Mahadev's protocol, and then use Mahadev's $\Out$ algorithm to decode $\vecz$, which as mentioned above, can be done publicly without the secret key $\sk$ (since it is a standard basis opening).  Importantly, we show that the computational indistinguishability of the Hadamard basis opening only relies on the fact that $\sk$ is secret, and does not rely on the secrecy of $\sk'$.  Thus, the remaining problem, which is at the heart of the technical complication, is the leakage of the Hadamard basis openings on $\sk$. We note that in Mahadev's protocol, the Hadamard basis openings may leak the entire $\sk$. What saves us in our setting is the fact that we tie the hands of the  adversary when opening in the Hadamard basis.  
To explain this in more detail we need to recall Mahadev's measurement protocol. 

\fi


 \subsection{Mahadev's measurement protocol}\label{sec:overview:Mahadev}
As mentioned, Mahadev's measurement protocol \cite{Mah18a} uses a noisy $\tcf$ family.\footnote{As mentioned above, her work, as well as  followup works, use a dual-mode $\tcf$ family; we avoid this technicality.} In this overview, for the sake of simplicity, we describe her scheme assuming we have a noiseless $\tcf$ family, which is a function family associated with algorithms 
\[
(\Gen_\tcf, \Eval_\tcf,\Invert_\tcf)
\]
where $\Gen_\tcf$ is a $\PPT$ algorithm that takes as input the security parameter~$1^\secp$ and outputs a key pair $(\pk,\sk)$; $\Eval$ is a poly-time deterministic algorithm that takes as input the public key~$\pk$, and a pair $(b,\vecx)$ where $b\in\{0,1\}$ is a bit and $\vecx\in\{0,1\}^n$ (where $n=\poly(\secp)$), and outputs a value~$\vecy$, and $\Eval(\pk,\cdot)$ is a two-to-one function where every $\vecy$ in the image has exactly two preimages of the form $(0,\vecx_0)$ and $(1,\vecx_1)$; $\Invert_\tcf$ takes as input the secret key $\sk$ and an element $\vecy$ in the image and it outputs the two preimages $((0,\vecx_0),(1,\vecx_1))$.

In what follows we show how Mahadev uses a $\tcf$ family to construct a measurement protocol.  The following protocol slightly differs from Mahadev's scheme, and in particular the basis choice is not determined during the key generation algorithm.
The measurement protocol consists of algorithms
$(\Gen,\Commit,\Open,\Out)$
defined as follows:
\begin{itemize}
    \item $\Gen$ is identical to $\Gen_\tcf$; it takes as input the security parameter~$1^\secp$ and outputs a key pair $(\pk,\sk)$.
    \item $\Commit$ takes as input $\pk$ and a single-qubit pure state $\ket{\psi}=\alpha_0 \ket{0} + \alpha_1 \ket{1}$ and generates
    \[
    \ket{\psi'}=\alpha_0 \ket{0,x_0} + \alpha_1 \ket{1,x_1}
    \]
    such that $\Eval(\pk,(0,\vecx_0))=\Eval(\pk,(1,\vecx_1))=\vecy$, and outputs $\vecy$ as the commitment string.
    \item $\Open$ takes as input the post-committed state $\ket{\psi'}$ and a basis $b\in\{0,1\}$; if $b=0$ it returns the outcome $\vecz$ of measuring $\ket{\psi'}$ in the standard basis, which is of the form $(b,x_b)$, and if $b=1$ it returns the outcome $\vecz$ of measuring $\ket{\psi'}$ in the Hadamard basis.
    \item $\Out$ takes as input $(\sk,\vecy,b,\vecz)$, and if $b=0$ it checks that $\Eval(\pk,\vecz)=\vecy$ and if this is the case it outputs the first bit of $\vecz$, and otherwise it outputs $\bot$.  If $b=1$ if outputs $\vecz\cdot(1,\vecx_0\oplus\vecx_1)$ where 
    $((0,\vecx_0),(1,\vecx_1))=\Invert_\tcf(\vecy)$.
\end{itemize}
Recall that, as explained in the introduction, Mahadev's measurement protocol is not fully binding.
The issue is that a cheating prover can produce any opening in the Hadamard basis, and will never be rejected. For instance, a cheating prover could commit to $\ket{+}$ honestly, apply a $Z$ to the first qubit of the post-commitment state, and then open to $\ket{-}$.

\subsection{Our Single-Qubit Commitment Scheme}

We convert Mahadev's protocol into a binding commitment scheme by adding another step to the commitment algorithm, as described in the beginning of \Cref{sec:overview}. 
More specifically, our commitment scheme consists of algorithms $(\Gen,\Commit,\Open,\Out)$ defined as follows: 
\begin{itemize}
    \item $\Gen(1^\secp)$ generates $n+2$ $\tcf$ keys $(\pk_i,\sk_i)_{i\in \{0,1,\ldots,n+1\}}$, where each $(\pk_i,\sk_i)\gets\Gen_\tcf(1^\secp)$, and outputs 
$\pk=(\pk_0,\pk_1,\ldots,\pk_{n+1})$ and $\sk=(\sk_0,\pk_1,\ldots,\sk_{n+1})$.
\item $\Commit(\pk,\ket{\psi})$ operates as follows:
\begin{enumerate}
\item Parse $\pk=(\pk_0,\pk_1,\ldots,\pk_{n+1})$.
    \item Apply Mahadev's measurement protocol to commit to $\ket{\psi}=\alpha_0 \ket{0} + \alpha_1 \ket{1}$ w.r.t.\ $\pk_0$; i.e., generate
\[
    \ket{\psi'}=\alpha_0 \ket{0,x_0} + \alpha_1 \ket{1,x_1}
    \]
such that  $\Eval(\pk_0,(0,\vecx_0))=\Eval(\pk_0,(1,\vecx_1))=\vecy_0$.
\item Compute 
\[
H^{\tensor(n+1)}\ket{\psi'}=\sum_{\vecd\in\{0,1\}^{n+1}} \beta_{\vecd}\ket{\vecd},
\]
\item Use Mahadev's measurement protocol to commit qubit-by-qubit to the above $(n+1)$-qubit state, w.r.t.\ public keys $\pk_1,\ldots,\pk_{n+1}$ to obtain the state
\[
\sum_{\vecd\in\{0,1\}^{n+1}}\beta_{\vecd} \ket{\vecd}\ket{\vecx'_{1,\vecd_1}}\ldots\ket{\vecx'_{n+1,\vecd_{n+1}}}
\]
and strings $\vecy_1,\ldots,\vecy_{n+1}$ such that for every $i\in[n+1]$,
\[
\Eval(\pk_i,(0,\vecx'_{i,0}))=\Eval(\pk_i,(1,\vecx'_{i,1}))=\vecy_i.
\]
\item Output $(\vecy_0,\vecy_1,\ldots,\vecy_{n+1})$, and (for simplicity) rearrange the post-commitment state to be 
\[\sum_{\vecd\in\{0,1\}^{n+1}}\beta_{\vecd} \ket{\vecd_1,\vecx'_{1,\vecd_1}}\ldots\ket{\vecd_{n+1},\vecx'_{{n+1},\vecd_{n+1}}}
\]
\end{enumerate}
\item $\Open$ takes as input the post-commitment state $\brho$ and a basis $b\in\{0,1\}$.  If $b=1$ (corresponding to opening in the Hadarmard basis) then it outputs the measurement of the state $\brho$ in the standard.  If $b=0$  (corresponding to opening in the stanadard basis) then it outputs the measurement of the state $\brho$ in the Hadamard basis. 
\item $\Out$ takes as input the secret key $\sk=(\sk_0,\sk_1,\ldots,\sk_{n+1})$, a commitment string $\vecy=(\vecy_0,\vecy_1,\ldots,\vecy_{n+1})$, a basis $b\in\{0,1\}$ and an opening string $\vecz\in\{0,1\}^{(n+1)^2}$ and does the following:
\begin{enumerate}
    \item If $b=1$ then parse 
    \[\vecz=(\vecd_1,\vecx'_{1,\vecd_1},\ldots,\vecd_{n+1},\vecx'_{1,\vecd_1})
    \]
    and check that for every $i\in[n+1]$ it holds that 
    \[\vecy_i=\Eval(\pk_i,(\vecd_i,\vecx'_{i,\vecd_i})).
    \]
    If all these checks pass then output $\vecd\cdot (1,\vecx_0\oplus \vecx_1)$, where $((0,\vecx),(1,\vecx_1))=\Invert(\sk_0\vecy_0)$.  Otherwise, output~$\bot$.
    \item If $b=0$ then parse $\vecz=(\vecz_1,\ldots,\vecz_{n+1})$, and for every $i\in[n+1]$ compute 
    \[
    ((0,\vecx'_{i,0}),(1,\vecx'_{i,1}))=\Invert(\sk_i,\vecy_i)~~\mbox{ and }~~m_i=\vecz_i\cdot(1,\vecx'_{i,0}\oplus \vecx'_{i,1})
    \]
    If $\Eval(\pk_0,(m_1,\ldots,m_{i+1}))=\vecy_0$ then output~$m_1$, and otherwise output~$\bot$.
\end{enumerate}

\end{itemize}

\paragraph{Analyzing the leakage.}  We next analyze the leakage that a cheating $\BQP$ algorithm $\Commit^*$ and a cheating $\BQP$ algorithm $\Open^*$ obtain by, given $\pk=(\pk_0,\pk_1,\ldots,\pk_{n+1})$, generating a commitment string $\vecy=(\vecy_1,\ldots,\vecy_\ell)$, where each $\vecy_i=(\vecy_{i,0},\vecy_{i,1},\ldots,\vecy_{i,n+1})$, a basis $(b_1,\ldots,b_\ell)$ and an opening $\vecz=(\vecz_1,\ldots,\vecz_\ell)$, and obtaining outputs $m_i=\Out(\sk,\vecy_i,b_i,\vecz_i)$ for every $i\in[\ell]$.
Denote by 
\[I=\{i:~b_i=0\}~~\mbox{ and }~~J=\{i:~b_i=1\}.
\]
We distinguish between the leakage obtained from $\{m_i\}_{i\in I}$ and that obtained from $\{m_i\}_{i\in J}$.  As mentioned above, $\{m_i\}_{i\in I}$ only leaks information about $\sk_1,\ldots,\sk_{n+1}$, since $\sk_0$ is not used when computing  $\{m_i\}_{i:~b_i=0}$.
For $i\in J$, it holds that
\[m_i=\vecd_i\cdot (1,\vecx_{i,0}\oplus\vecx_{i,1})~~\mbox{ where }~~((0,\vecx_{i,0}),(1,\vecx_{i,1}))=\Invert_\tcf(\sk_0,\vecy_{i,0}),
\]
where $\vecz_i=(\vecd_{i,1},\vecx'_{i,1,\vecd_1},\ldots,\vecd_{i,n+1},\vecx'_{i,n+1,\vecd_{n+1}})$.
This may leak information about $\sk_0$.
In particular, if we use the underlying (noisy) $\tcf$ family from \cite{BCMVV18}, along with an adversarially chosen $\vecd=(\vecd_1,\ldots,\vecd_{\ell})$ then $\{m_i\}_{i\in J}$ may leak part of the secret key which breaks the indistinguishability between the output produced by $\Open^*$ and $\Ext^{\Open^*}$.  

We get around this problem by arguing that in our scheme if $\vecz$ is accepted then it must be the case that the ``important'' bits of $\vecd$ have min-entropy $\omega(\log \secp)$.\footnote{We emphasize that this is not the case for Mahadev's scheme, since in her scheme every Hadamard opening $\vecd$ is accepted.} 
For this we rely on the fact that the underlying $\tcf$ family has the adaptive hardcore bit property, which the  (noisy) $\tcf$ family from \cite{BCMVV18} was proven to have under the $\LWE$ assumption. 
We actually need the stronger condition that the ``importants'' bits of $\vecd$ have min-entropy $\omega(\log \secp)$ even given some auxiliary input (which comes into play due to the fact that we are opening many qubits). We prove this for the specific $\TCF$ family from  \cite{BCMVV18}.  Specifically, we prove that under the $\LWE$ assumption, the $\TCF$ family from  \cite{BCMVV18} has a property which we refer to as the {\em distributional strong adaptive hardcore bit property}. We argue that this property, together with the min-entropy property of $\vecd$, implies that the leakage obtained from $\vecd_i\cdot(\vecx_{i,0}\oplus\vecx_{i,1})$ is benign and does not break the indistinguishability between the output produced by $\Open^*$ and $\Ext^{\Open^*}$.  

In more detail, for Mahadev's measurement protocol, the proof that the Hadamard outputs of $\Ext^{\Open^*}$ and $\Open^*$ are computationally indistinguishable relies on the adaptive hardcore bit property, which states that for every $\BQP$ adversary $\A$,
\[
\Pr[\A(\pk)=(b,\vecx_b,\vecd,\vecd\cdot(1,\vecx_0\oplus\vecx_1))]\leq \frac12+\negl(\secp)
\]
where $((0,\vecx_0),(1,\vecx_1))=\Invert(\sk,\Eval(\pk,b,\vecx_b))$.
We need to argue that this holds even if $\A$ gets as auxiliary input a bunch of elements of the form \[(b_i,\vecx_{i,b_i},\vecd_i\cdot(1,\vecx_{i,0}\oplus\vecx_{i,1})).
\]
While in general this is not true, we prove that it is true for the (noisy) $\tcf$ family from \cite{C:BBCM92}, if each $\vecd_i$ has $\omega(\log \secp)$ min-entropy (even conditioned on $(\vecd_1,\ldots,\vecd_{i-1})$), under the $\LWE$ assumption.



\subsection{Succinct commitments}  

As mentioned in the introduction, our main result is a {\em succinct} commitment scheme, where $\Commit$ commits to an $\ell$-qubit state by generating a {\em succinct} classical commitment string that consists of only $\poly(\secp)$ many bits,  and $\Open$ generates an opening to any qubit $i\in[\ell]$ in any basis $b\in\{0,1\}$, where the opening consists of only $\poly(\secp)$ many bits.  Importantly, the guarantee we provide is that even if $\Open^*$ only opens to a few qubits, we should still be able to extract the entire $\ell$-qubit quantum state from $\Open^*$.\footnote{ This guarantee is important for our applications, as we will see in \Cref{sec:overview:app}.} This seems impossible to do, since how can we extract information about qubits that were never opened?  Indeed, to achieve this we need to change the syntax.  

We add to the syntax an {\em interactive test phase}.  Similar to the test round in Mahadev's protocol, our test phase is executed with probability $1/2$, and if it is executed then after the test phase the protocol is terminated and the opening phase is never run.  This is the case since the test phase destroys the quantum state. 
Importantly, we allow the test phase to be {\em interactive}.  It is this interaction that allows us to extract a long $\ell$-qubit state from $\Open^*$. Loosely speaking, in this test phase, we choose at random $b\gets\{0,1\}$ and ask the prover to provide an opening to all the $\ell$-qubits in basis $b^\ell$.
To ensure that the protocol remains succinct, we ask for the openings to be sent in a succinct manner, using a Merkle hash. Then the prover and verifier engage in a succinct interactive argument where the prover proves knowledge of the committed openings.  For this we use Kilian's protocol and the fact that it is a proof-of-knowledge even in the post-quantum setting \cite{CMSZ,VidickZ21}. 
Then the verifier sends the prover the secret key $\sk$ and the prover and verifier engage in a succinct interactive argument where the prover proves that the committed openings are accepted (w.r.t.\ $\sk$). This is also done using the Kilian protocol.  

In addition, we allow the commit phase to be interactive.  This allows $\Commit$ to first generate a non-succinct commitment $\vecy$, and send its Merkle hash, denoted by $\rt$.  Then the committer can run a succinct proof-of-knowledge interactive argument, to prove knowledge of a preimage of $\vecy$.  Importantly, the proof-of-knowledge must be {\em state-preserving}, which means that we can extract $\vecy$ without destroying the state. Such a state-preserving proof-of-knowledge protocol was recently constructed in \cite{LombardiMS22}.
This interactive commitment phase allows us to reduce the binding of the succinct commitment scheme to that of the semi-succinct one.  This part of the analysis is similar to \cite{Bartusek22}.



\subsection{Applications} \label{sec:overview:app}
 We construct a succinct interactive argument for $\QMA$ and a compiler that converts any $X/Z$ PCP into a succinct interactive argument, both under the $\LWE$ assumption. For simplicity we do not use our succinct commitment scheme to construct these succinct interactive arguments.  Rather we use our {\em semi-succinct} commitment scheme to construct a {\em semi-succinct} interactive argument.
 We then rely on a black-box transformation from \cite{Bartusek22} which shows a generic transformation for converting any semi-succinct interactive argument for $\QMA$ into a fully succinct one.\footnote{We mention that this transformation was used (in a non-black-box way to convert our semi-succinct commitment scheme into a succinct one. }

An important point to note is that the argument systems we construct have negligible soundness, even though the binding guarantee of our commitment only holds for provers with a probability close to $1$ of being accepted. We do this by applying sequential repetition to our semi-succinct protocol to drive down the soundness error, before applying the black-box transformation of \cite{Bartusek22}. A drawback of this approach (shared with all known succinct argument systems for $\QMA$) is that the honest prover must hold polynomially many copies of the $\QMA$ witness state, or the X/Z PCP state.

The formal statements of these results are contained in \Cref{thm:succinct-qma} and \Cref{thm:succinct-QPCP}.
 
\paragraph{Compiling an $X/Z$ PCP into a semi-succinct interactive argument} 
Our compiler uses a succinct commitment in a straightforward way. 
The succinct interactive argument proceeds as follows:
\begin{enumerate}
    \item The verifier generates a key pair $(\pk,\sk)\gets \Gen(1^\secp)$ corresponding to the underlying semi-succinct commitment scheme.
    \item The prover commits to the $X/Z$ PCP~$\ket{\pi}$ by generating a classical commitment string $\vecy\gets \Commit(\pk,\ket{\pi})$.  It sends $\vecy$ to the verifier.  
   \item With probability $1/2$ the verifier behaves as the PCP verifier and chooses small set of indices $(i_1,\ldots,i_k)$ along with basis choices $(b_1,\ldots,b_k)$; with probability $1/2$ the verifier chooses a random $b\gets\{0,1\}$ and sends $b$ to the prover.
   \item If the prover receives a bit $b$ then it opens the entire PCP in the standard basis if $b=0$ and in the Hadamard basis if $b=1$.  Otherwise, if the prover receives a set of indices $(i_1,\ldots,i_k)$ along with basis choices $(b_1,\ldots,b_k)$ then the prover opens these locations in the desired basis.
\end{enumerate}
Completeness follows immediately from the completeness of the underlying semi-succinct commitment scheme.   
To argue soundness, fix a cheating prover~$P^*$ that is accepted with high probability. We  rely on the soundness property of the underlying commitment scheme to argue that there exists a $\BQP$ extractor that extracts a state $\ket{\pi^*}$ from $P^*$, such that on a random challenge produced by the PCP verifier (for which $P^*$ succeeds in opening with high probability), the output of $P^*$ is close to the the outcome obtained by measuring $\ket{\pi^*}$ directly, which implies that $\ket{\pi^*}$ is an $X/Z$ PCP that is accepted with high probability, implying that the soundness property holds.
 
\paragraph{Semi-succinct interactive argument for $\QMA$} 
 
To obtain a semi-succinct argument, we follow the blueprint of Mahadev~\cite{Mah18a}.  Namely, we first convert the $\QMA$ witness into one that can be verified by measuring only in the $X/Z$ basis. For this we rely on a result  due to Fitzsimons, Hajdu\v{s}ek, and Morimae~\cite{FHM18} which shows how to convert multiple copies of the $\QMA$ witness into an $\ell$-qubit state $\ket{\pi}$ that can be verified by measuring it only in the $X/Z$ basis. Importantly, this state can be verified by measuring it in a random basis $(b_1,\ldots,b_\ell)\gets\{0,1\}^\ell$.  Armed with this tool, the semi-succinct interactive argument proceeds as follows:
\begin{enumerate}
\item The verifier generates a key pair $(\pk,\sk)$ corresponding to the underlying semi-succinct commitment scheme.
\item The prover converts its (multiple copies) of the $\QMA$ witness into a state $\ket{\pi}$ by relying on the \cite{FHM18} result, and computes $\vecy\gets \Commit(\pk,\ket{\pi})$.
\item With probability $1/2$ the verifier chooses at random a seed $s\in\{0,1\}^\secp$ and sends $s$ to the prover, and with probability $1/2$ the verifier chooses a random $b\gets\{0,1\}$ and sends $b$ to the prover.
\item If the prover receives a bit $b$ then it sends the opening of the commitment in the basis $b^\ell$.  If it receives a seed~$s$ then it uses a pseudorandom generator to deterministically expand $s$ to a pseudorandom string $(b_1,\ldots,b_\ell)$ and sends an opening of the commitment in basis $(b_1,\ldots,b_\ell)$.
\item The verifier uses its secret key to compute the output corresponding to this opening. If any of the openings are rejected it rejects.  Otherwise, in the case that it sent a seed, it accepts if the verifier from  \cite{FHM18}  would have accepted. 
\end{enumerate}
To argue soundness, fix a cheating prover~$P^*$ that is accepted with high probability. We first rely on the soundness property of the underlying commitment scheme to argue that the $\BQP$ extractor extracts a state $\ket{\pi^*}$ from $P^*$, such that for any choice of basis $\vecb=(b_1,\ldots,b_\ell)$ for which $P^*$ succeeds in opening with high probability, the output corresponding to these openings are computationally indistinguishable from measuring $\ket{\pi^*}$ in basis $\vecb$. By the soundness of the underlying scheme~\cite{FHM18} we note that for a random basis $(b_1,\ldots,b_\ell)$, the state would be rejected with high probability.  Hence it must also be the case if the basis is pseudorandom, as otherwise one can distinguish a pseudorandom string from a truly random one.   

\section{Preliminaries}\label{sec:prelim}

\paragraph{Notations.} For any random variables $A$ and $B$ (classical variables or quantum states), we use the notation $A\equiv B$ to denote that $A$ and $B$ are identically distributed, and use $A\stackrel{\epsilon}\equiv B$ to denote that $A$ and $B$ are $\epsilon$-close, where closeness is measured with respect to total variation distance for classical variables,  trace distance for mixed quantum states, and $\|\cdot\|_2$ distance for pure quantum states.  For every two ensemble of distributions $A=\{A_\lambda\}_{\lambda\in\mathbb{N}}$ and $B=\{B_{\lambda}\}_{\lambda\in\mathbb{N}}$ we use the notation $A\approx B$ to denote that $A$ and $B$ are computationally indistinguishable, i.e., for every polynomial size distinguisher $\D$ there exists an negligible function $\mu=\mu(\secp)$ such that for every $\secp\in\mathbb{N}$,
\[
|\Pr[\D(a)=1]-\Pr[\D(b)=1]|\leq \mu(\secp)
\]
where the probabilities are over $a\gets A_\secp$ and $b\gets B_\secp$.  For every $\epsilon=\epsilon(\secp)\in[0,1)$, we use the notation $A\stackrel{\epsilon}\approx B$ to denote that for every polynomial size distinguisher $\D$ and for every $\secp\in\mathbb{N}$,
\[
|\Pr[\D(a)=1]-\Pr[\D(b)=1]|\leq \epsilon(\secp)
\]
where the probabilities are over $a\gets A_\secp$ and $b\gets B_\secp$.

For any random variable $\A$, we denote by $\Supp(\A)$ the support of $\A$; i.e., \[\Supp(\A)=\{a: \Pr[\A=a]>0\}.
\]
We denote strings in $\{0,1\}^*$ by bold lower case letters, such as $\vecx$.
We let $\PPT$ denote probabilistic polynomial time, $\BQP$ denote probabilistic quantum polynomial time, and $\QPT$ denote quantum polynomial time. 

Let $\cH$ be a complex Hilbert space of finite dimension~$2^n$.  Thus, $\cH \simeq  \C^{2^n}$ where $\C$ denotes the complex numbers.  
A pure $n$-qubit quantum state is a unit vector $\ket{\Psi}\in \cH$.
Namely, it can be written as 
$$\ket{\Psi}=\sum _{b_1,\ldots,b_n\in\{0,1\}} \alpha_{b_1,\ldots,b_n}\ket{b_1,\ldots,b_n}
$$
where $\{\ket{b_1,\ldots,b_n}\}_{b_1,\ldots,b_n\in\{0,1\}}$ forms an orthonormal basis of $\cH$, and where $\alpha_{b_1,\ldots,b_n}\in \C$ satisfy
$$
\sum _{b_1,\ldots,b_n\in\{0,1\}}|\alpha_{b_1,\ldots,b_n}|^2=1.
$$
We refer to $n$ as the number of qubits in $\ket{\Psi}$.  
We sometimes divide the registers of $\ket{\Psi}$ into named registers.  We often denote these registers by calligraphic upper-case letters, such as $\cA$ and $\cB$, in which case we also divide the Hilbert space into $\cH = \cH_{\cA} \otimes \cH_\cB$, so that each quantum state $\ket{\Psi}$ is a linear combination of quantum states $\ket{\Psi_\cA}\otimes\ket{\Psi_\cB} \in \cH_\cA\otimes\cH_\cB$.\footnote{We sometimes give registers names that correspond to their purpose, such as a $\mathsf{coin}$ register or an $\mathsf{open}$ register.}  We denote by $$\ket{\Psi}_\cA=\Tr_\cB(\ket{\Psi}\bra{\Psi})\in\cH_\cA,$$ 
where $\Tr_\cB$ is the linear operator defined by $$\Tr_\cB(\ket{\Psi_\cA}\bra{\Psi_\cA}\otimes \ket{\Psi_\cB}\bra{\Psi_\cB})=\ket{\Psi_\cA}\bra{\Psi_\cA}\cdot \Tr(\ket{\Psi_\cB}\bra{\Psi_\cB}),$$
where $\Tr$ is the trace operator.

Let $\mathrm{D}(\RegH)$ denote the set of all positive semidefinite operators on $\RegH$ with trace~$1$. A mixed state is an operator $\bm{\rho} \in \mathrm{D}(\RegH)$, and is often called a \emph{density matrix}.
We denote by
$$U[\bsigma]=U\bsigma U^\dagger.
$$
For any binary observable $O$ and bit $b\in\{0,1\}$ we let $\Pi_{O,b}[\bsigma]$ denote the \emph{unnormalized} projection of $\bsigma$ to the state that has value~$b$ when measured in the $O$-basis.
Namely, 
$$\Pi_{O,b}[\bsigma]=\Pi_{O,b}\bsigma \Pi_{O,b}^\dagger.
$$

We let $X$ and $Z$ denote the Pauli matrices:
\[ X = \begin{pmatrix} 0 & 1 \\ 1 & 0 \end{pmatrix}, \quad Z = \begin{pmatrix} 1 & 0 \\ 0 & -1\end{pmatrix}. \]
For any single qubit register $\cA$, we denote by $X_\cA$ (respectively, $Z_\cA$) the unitary that applies the Pauli $X$ (respectively, $Z$) unitary to the $\cA$ register of a given quantum state and applies the identity unitary to all other registers.

\subsection{Quantum information facts}



    


We use the following infant version of the gentle measurement lemma.
\begin{lemma}\label{lem:gentle-infant}
Let $\ket{\psi}$ be a pure state and $\Pi$ be a projector such that $\bra{\psi} \Pi \ket{\psi} = 1- \varepsilon$. Then $\| \Pi \ket{\psi} - \ket{\psi} \|_2 = \sqrt{\varepsilon}$.
\end{lemma}
\begin{proof}
    Calculate:
\[ \| \Pi \ket{\psi} - \ket{\psi} \|_2^2 = 1 - \bra{\psi} \Pi \ket{\psi} = \varepsilon.\]
\end{proof}
We also use the following version for mixed states.
\begin{lemma}
\label{lem:gentle-mixed}
    Let $\rho$ be a mixed state and $\Pi$ be a projector such that $\Tr[\Pi \rho] = 1-\varepsilon$. Then $\frac{1}{2} \|\rho - \Pi[\rho]\|_1 \leq \sqrt{\varepsilon}$.
\end{lemma}
\begin{proof}
    This is Lemma 9.4.2 of~\cite{Wilde11}. 
\end{proof}
\begin{lemma}\label{lem:control-distance-z}
    Suppose $\ket{\psi_1}_{ABC}$ and $\ket{\psi_2}_{ABC}$ are pure states with $A$ being a single-qubit register, and $U$ is a unitary acting on register $B$. Let $CU_{AB}$ denote the controlled version of $U$, with $A$ being the control system and $B$ the target. Then
    \[ \| CU_{AB} \otimes I_C (\ket{\psi_1} - \ket{\psi_2})\|_2 = \| Z_A \otimes I_{BC} \ (\ket{\psi_1} - \ket{\psi_2}) \|_2 ,\]
    where $Z$ is the Pauli $Z$ operator.
\end{lemma}
\begin{proof}
In the following calculation we omit factors of identity that are clear from context.
    \begin{align*}
        \|CU_{AB} (\ket{\psi_1} - \ket{\psi_2}) \|_2^2 &= \| \ket{0}\bra{0}_A (\ket{\psi_1} - \ket{\psi_2}) + \ket{1}\bra{1}_A \otimes U_B (\ket{\psi_1} - \ket{\psi_2}) \|_2^2 \\
        &= \| \ket{0}\bra{0}_A (\ket{\psi_1} - \ket{\psi_2}) \|_2^2 + \| \ket{1}\bra{1}_A \otimes U_B (\ket{\psi_1} - \ket{\psi_2}) \|_2^2 \\
        &= \| \ket{0}\bra{0}_A (\ket{\psi_1} - \ket{\psi_2}) \|_2^2 + \| \ket{1}\bra{1}_A \otimes U_B (\ket{\psi_1} - \ket{\psi_2}) \|_2^2 \\
        &= \| \ket{0}\bra{0}_A (\ket{\psi_1} - \ket{\psi_2}) \|_2^2 + \| \ket{1}\bra{1}_A  (\ket{\psi_1} - \ket{\psi_2}) \|_2^2 \\
        &= \| \ket{0}\bra{0}_A (\ket{\psi_1} - \ket{\psi_2}) - \ket{1}\bra{1}_A  (\ket{\psi_1} - \ket{\psi_2}) \|_2^2 \\
        &= \| Z_A (\ket{\psi_1} - \ket{\psi_2}) \|_2^2.
    \end{align*}
\end{proof}

\subsection{Hash Family with Local Opening}\label{sec:HT}

A hash family with local opening
consists of the following algorithms:
\begin{description}
\item[$\Gen(1^\secp)\to \hk$.]  This $\PPT$  algorithm takes as input a security parameter $\secp$ (in unary) and outputs a hash key $\hk$. 
\item[$\Eval(\hk,\vecx)\to \rt$.] This deterministic poly-time algorithm takes as input a hash key $\hk$ and a string $\vecx\in\{0,1\}^N$ and outputs a hash value (often referred to as hash root) $\rt \in \{0,1\}^{\poly(\secp, \log N)}$.
  \item[$\Open(\hk,\vecx,i)\to(b,\veco)$.]  This deterministic poly-time algorithm takes as input a hash key $\hk$, a string $\vecx\in\{0,1\}^N$ and an index $i\in[N]$. It outputs a bit $b\in\{0,1\}$ and an opening $\veco\in\{0,1\}^{\poly(\secp, \log N)}$.
  \item[$\Ver(\hk,\rt,i,b,\veco)\rightarrow 0/1$.]   This deterministic poly-time algorithm takes as input a hash key $\hk$, a hash root $\rt$, an index $i\in[N]$, a bit $b\in\{0,1\}$ and an opening $\veco\in\{0,1\}^{\poly(\secp, \log N)}$. It outputs a bit indicating whether or not the opening is valid.
\end{description}

\begin{definition}\label{def:HT} A hash family with local opening $(\Gen,\Eval,\Open,\Ver)$ is required to satisfy the following properties.
 \begin{description}
   \item[~~Opening completeness.] For any $\secp\in \mathbb{N}$, any $N = N(\secp)  \leq 2^\secp$, any $\vecx=(x_1,\ldots,x_n)\in\{0,1\}^N$, and any index $i \in [N]$, 
    \[ \probcond{b=x_i \\
    \wedge~\Ver(\hk, \rt, i, b, \veco) = 1}
    {
    \hk \gets \Gen(1^\secp), \\
    \rt = \Eval(\hk, \vecx), \\
    (b,\veco) = \Open(\hk, \vecx, i)} 
    =1 -\negl(\secp). 
    \]
      
    \item[~~Computational binding w.r.t.\ opening.]  For any poly-size adversary $\A$, there exists a negligible function $\negl(\cdot)$ such that for every $\secp \in \mathbb{N}$, 
    \[ \probcond{\Ver(\hk,\rt,i,0,\veco_0) =1\\
    \wedge~ \Ver(\hk,\rt,i,1,\veco_1) =1}{
        \hk \gets \Gen(\secp), \\
        (1^N,\rt,i,\veco_0,\veco_1)\gets \A(\hk)} = \negl(\secp).\]   
  \end{description}
\end{definition}

\begin{theorem}[\cite{C:Merkle87}]
Assuming the existence of a collision resistant hash family there exists a hash family with local opening (according to \Cref{def:HT}).
\end{theorem}

\begin{definition}\label{def:collapsing:local} A hash family with local opening  $(\Gen,\Eval,\Open,\Ver)$ is said to be collapsing if any $\BQP$ adversary $\A$ wins in the following game with probability $\frac{1}{2} + \negl(\lambda)$:
  \begin{enumerate}
        \item The challenger generates $\hk\gets \Gen(1^\secp)$ and sends $\pk$ to $\A$.
            \item $\A(\hk)$  generates a classical value $(\rt,j)$ and a quantum state $\bsigma$.
            
            $\A$ sends $(\rt,j,\bsigma)$ to the challenger.

            \item The challenger does the following:
            \begin{enumerate}
                \item Apply in superposition the algorithm $\Ver(\hk,
                \rt,j,\cdot,\cdot)$ to $\bsigma$, and measure the output.  If the output is~$0$ then send $\bot$ to $\A$. Otherwise, denote the resulting state by $\bsigma'$
                \item   Choose a random bit $b\leftarrow\{0,1\}$.
                \item If $b=0$ then send $\bsigma'$ to $\A$.
                \item If $b=1$ then measure $\bsigma'$ in the standard basis and send the resulting state to $\A$.

            \end{enumerate}

            \item Upon receiving the quantum state (or the symbol $\bot$), $\A$ outputs a bit $b'$.
            \item $\A$ wins if $b'=b$
        \end{enumerate}

  \end{definition}

\begin{theorem}[\cite{AC:Unruh16,CMSZ}]
    There exists a hash family with local opening that is collapsing assuming the post-quantum hardness of $\LWE$.
\end{theorem}

\subsection{Noisy Trapdoor Claw-Free Functions}\label{sec:TCF}

In what follows we define the notion of a  {\em noisy trapdoor claw-free function family}. This notion is simpler than the notion of a {\em dual-mode noisy trapdoor claw-free function family} which was used for certifiable randomness generation in \cite{BCMVV18} and by Mahadev~\cite{Mah18a} in her classical verification protocol for $\QMA$. This simpler notion suffices for our work.\footnote{Our formulation is from~\cite{Bartusek22} (without the dual-mode requirement).}

\begin{definition}\label{def:TCF}
    A noisy trapdoor claw-free function ($\TCF$) family is described by $\PPT$ algorithms: $$(\Gen,\Eval,\Invert, \Check, \Good)$$ with the following syntax:
 \begin{description}
      \item [$\Gen(1^\lambda)\rightarrow (\pk,\sk)$.] This $\PPT$ key generation algorithm takes as input a security parameter $\secp$ (in unary) and outputs a public key $\pk$ and a secret key $\sk$. 
      
      We denote by $\mathsf{D}_{\pk}$ the domain of the (randomized) function defined by $\pk$, and assume for simplicity that  $\mathsf{D}_{\pk}$ is an efficiently verifiable and samplable subset of $\{0,1\}^{n(\secp)}$, where $n(\secp)=\poly(\secp)$. We denote by $\mathsf{R}_{\pk}$ the range of this (randomized) function.
     \item [$\Eval(\pk,b,\vecx)\rightarrow \vecy$.] This $\PPT$ algorithm takes as input a public key $\pk$, a bit $b\in\{0,1\}$ and an element $\vecx \in \mathsf{D}_{\pk}$, and outputs a string $\vecy$ distributed according to some distribution~$\chi=\chi_{\pk,b,\vecx}$.  
      \item [$\Invert(\sk,\vecy)\rightarrow ((0,\vecx_0),(1,\vecx_1))$.] This deterministic polynomial time algorithm takes as input a secret key $\sk$, and an element $\vecy$ in the range $\mathsf{R}_\pk$ and outputs two pairs $(0,\vecx_0)$ and $(1,\vecx_1)$ with $\vecx_0,\vecx_1\in \mathsf{D}_\pk$, or $\bot$. 
      \item [$\Check(\pk, b, \vecx, \vecy)\rightarrow 0/1$.] This deterministic poly-time algorithm takes as input a public key $\pk$, a bit $b\in \{0,1\}$, an element $\vecx \in \mathsf{D}_\pk$ and an element $\vecy \in \mathsf{R}_\pk$ and outputs a bit.  
      \item [$\Good(\vecx_0, \vecx_1, \vecd)\rightarrow 0/1$.] This deterministic poly-time algorithm takes as input two domain elements $\vecx_0, \vecx_1 \in \mathsf{D}_\pk$ and a string $\vecd \in \{0,1\}^{n + 1}$. It outputs a bit that characterizes membership in the set:
      \begin{equation}
        \Good_{\vecx_0, \vecx_1} := \{\vecd : \Good(\vecx_0, \vecx_1, \vecd) = 1\}
      \end{equation}
        We specify that $\Good(\vecx_0, \vecx_1, \vecd)$ ignores the first bit of $\vecd$.    
      
      For the purpose of this work, we allow $\Good$ to output a vector (as opposed to a single bit).\footnote{We use this extension to analyze our succinct commitment scheme.} Specifically, $\vecd'=\Good(\vecx_0,\vecx_1,\vecd)$ is a vector in $\{0,1\}^k$ (for some $k=\poly(\secp)$).  It is computed by partitioning the vector $\vecd$ into $k$ blocks, denoted by $\vecd[1],\ldots,\vecd[k]$  and outputting  $\vecd'=(\langle{\veca[1],\vecd[1]\rangle},\ldots,\langle{\veca[k],\vecd[k]\rangle}) $, where $\veca[1],\ldots,\veca[k]$ are non-zero vectors determined by $\vecx_0$ and $\vecx_1$. We define  \[\Good_{\vecx_0, \vecx_1} := \{\vecd : \Good(\vecx_0, \vecx_1, \vecd) \neq 0^k\}.\]

\end{description}

We require that the following properties are satisfied. 
  \begin{enumerate}
      \item \textbf{\em Completeness:} 
      \begin{enumerate}
            \item\label{def:TCF-a} For all $(\pk, \sk) \in \Supp(\Gen(1^\lambda))$, every $b\in\{0,1\}$,  every $\vecx\in\mathsf{D}_\pk$, and $\vecy \in \Supp(\Eval(\pk,b,\vecx))$, 
      $$ \Invert(\sk,\vecy) = ((0,\vecx_0),(1,\vecx_1))$$ such that $\vecx_b=\vecx$ and  $\vecy\in\Supp(\Eval(\pk,\beta,\vecx_\beta))$ for every $\beta\in\{0,1\}$. 
            \item\label{def:TCF-b} For all $(\pk, \sk) \in \Supp(\Gen(1^\lambda))$, there exists a perfect matching $\mathsf{M}_\pk \subseteq \mathsf{D}_\pk \times \mathsf{D}_\pk$ such that for all $(\vecx_0, \vecx_1)$, it holds that $\Eval(\pk, 0, \vecx_0)  \equiv \Eval(\pk, 1, \vecx_1)$ if and only if $(\vecx_0, \vecx_1) \in \mathsf{M}_\pk$.
            \item\label{def:TCF-c} For all $(\pk, \sk) \in \Supp(\Gen(1^\lambda))$, every $b \in \{0,1\}$ and every $\vecx \in \mathsf{D}_\pk$, 
                  \begin{equation}
                        \Pr[\Check(\pk, b, \vecx, \vecy) = 1] = 1
                  \end{equation}
            if and only if $\vecy \in \Supp(\Eval(\pk, b, \vecx))$.

            \item\label{def:TCF-d} For all $(\pk, \sk) \in \Supp(\Gen(1^\lambda))$ and every pair of distinct domain elements $\vecx_0, \vecx_1$, the density of $\Good_{\vecx_0, \vecx_1}$ is $1 - \negl(\lambda)$.

      \end{enumerate}
       \item \textbf{\em Efficient Range Superposition:} 
       For every $(\pk, \sk) \in \Supp(\Gen(1^\lambda))$ and every $b \in \{0,1\}$, there exists an efficient $\QPT$ algorithm to prepare a state $\ket{\varphi_b}$ such that: 
       \begin{equation}
           \ket{\varphi_b} \stackrel{\mu}\equiv \frac{1}{\sqrt{\mathsf{D}_\pk}} \sum_{\substack{\vecx \in \mathsf{D}_\pk \\ \vecy \in \mathsf{R}_\pk}} \sqrt{p_{\pk}(b, \vecx, \vecy)} \ket{\vecx} \ket{\vecy}
       \end{equation}

       for some negligible function $\mu(\cdot)$. Here, $p_{\pk}(b, \vecx, \vecy)$ denotes the probability density of $\vecy$ in the distribution $\Eval(\pk, b, \vecx)$.
        
    \item\label{item:HCB} \textbf{\em Adaptive Hardcore Bit:}  For every $\BQP$ adversary $\A$ there exists a negligible function~$\mu$ such that for every $\secp\in\mathbb{N}$, 
      \begin{align*}
      & \Pr[\A(\pk)=(\vecy,b,\vecx,\vecd,v): \Check({\pk, b, \vecx, \vecy}) = 1  ~\wedge~ \vecd \in \Good_{\vecx_0, \vecx_1}~\wedge~ \vecd \cdot (1, \vecx_0\oplus \vecx_1)=v] \leq \frac12+\mu(\secp),
      \end{align*}
      where the probability is over $(\pk,\sk)\leftarrow \Gen(1^\secp)$, and where $((0,\vecx_0),(1,\vecx_1))=\Invert(\sk,\vecy)$.
  \end{enumerate}
\end{definition}

\begin{proposition}~\cite{BCMVV18}
There exists a $\TCF$ family assuming the post-quantum hardness of $\LWE$. 
\end{proposition}

In this work we rely on the fact that every $\TCF$ family is {\em collapsing}, as defined below.\footnote{This definition is similar to \Cref{def:collapsing:local} adapted to a $\TCF$ family.}

\begin{definition}\label{def:collapsing} A $\TCF$ family  $(\Gen, \Eval,\Invert,\Check,\Good)$ is said to be collapsing if any $\BQP$ adversary $\A$ wins in the following game with probability $\frac{1}{2} + \negl(\lambda)$:
  \begin{enumerate}
        \item The challenger generates $(\pk,\sk)\gets \Gen(1^\secp)$ and sends $\pk$ to $\A$.
            \item $\A(\pk)$  generates a classical value $\vecy\in{\mathsf{R}_\pk}$ and an $(n(\secp)+1)$-qubit quantum state $\bsigma=\bsigma_{\cS,\cZ}$, where the $\cS$ register contains a single qubit and the $\cZ$ register contains $n(\secp)$ many qubits.  
            
            $\A$ sends $(\vecy,\bsigma)$ to the challenger.

            \item The challenger does the following:
            \begin{enumerate}
                \item Apply in superposition the algorithm $\Check$ to $\bsigma$, w.r.t.\ public key $\pk$ and the image string $\vecy$, and measure the bit indicating whether the output of $\Check$ is~$1$. If the output does not equal $1$, send $\bot$ to $\A$. Otherwise, denote the resulting state by $\bsigma'$
                \item   Choose a random bit $b\leftarrow\{0,1\}$.
                \item If $b=0$ then it send $\bsigma'$ to $\A$.
                \item\label{item:alt-def} If $b=1$ then measure the $\cS$ register of $\bsigma'$ in the standard basis and send the resulting state to $\A$.

            \end{enumerate}

            \item Upon receiving the quantum state (or the symbol $\bot$), $\A$ outputs a bit $b'$.
            \item $\A$ wins if $b'=b$
        \end{enumerate}

  \end{definition}

  \begin{remark}\label{remark:collapsing}
            An equivalent definition of collapsing is obtained by replacing \Cref{item:alt-def} with the following:
            If $b=1$ then send to $\A$ the state $Z_\cS[\bsigma']$.  
            In this work we use both of these formulations, since it is sometimes easier to work with one and other times with the other.
        \end{remark}

\begin{proposition}\label{claim:collaps-binding} \cite{zhandry2022new} Every $\TCF$ family is collapsing.  
\end{proposition}

In what follows we define an extension of the collapsing property and argue that any $\TCF$ family satisfies it.  This extension may appear to be unnatural, but we make use of it when proving the binding property of our commitment schemes.

\begin{proposition}\label{claim:ext-collapsing}
For every polynomial $\ell=\ell(\secp)$, every $\TCF$ family $(\Gen, \Eval,\Invert,\Check,\Good)$ is $\ell$-{extended collapsing}, where the $\ell$-extended collapsing definition asserts that every $\BQP$ adversary $\A$ wins in the following {\em extended collapsing} game with probability  $\frac{1}{2} + \negl(\lambda)$:
  \begin{enumerate}
        \item The challenger generates $\ell$ independent public keys $\pk_1,\ldots\pk_{\ell}\gets \Gen(1^\secp)$ and sends $(\pk_1,\ldots,\pk_{\ell})$ to $\A$.
            \item $\A(\pk_1,\ldots,\pk_{\ell})$  generates a subset $J\subseteq [\ell]$, classical values $\{\vecy_j\}_{j\in J}$ where each $\vecy_j\in\mathsf{R}_{\pk_j}$, and a $|J|\cdot (n(\secp)+1)$-qubit quantum state $\bsigma=\bsigma_{\{\cS_j,\cZ_j\}_{j\in J}}$, where each register $\cS_j$ consists of a single qubit and each register $\cZ_j$ consists of $n(\secp)$-qubits.

            $\A$ sends $(J, \{\vecy_j\}_{j\in J},\bsigma)$ to the challenger.
            \item The challenger does the following:
            \begin{enumerate}
                \item For every $j\in J$ apply in superposition the algorithm $\Check$ to the $(\cS_j,\cZ_j)$ registers of $\bsigma$ w.r.t.\ $\pk_j$ and $\vecy_j$, and check that the output is~$1$.  If this is not the case  send $\bot$ to~$\A$.
                \item Otherwise, choose a random bit $b\leftarrow\{0,1\}$ 
                and measure the registers $\{\cS_j\}_{j\in J}$ in the standard basis if and only if $b=1$.
                \item Send the resulting state to $\A$.
            \end{enumerate}

            \item Upon receiving a quantum state (or the symbol $\bot$), $\A$ outputs a bit $b'$.
            \item $\A$ wins if $b'=b$
        \end{enumerate}

\end{proposition}
 \Cref{claim:ext-collapsing}  follows from \Cref{claim:collaps-binding} together with a straightforward hybrid argument.

\section{The Distributional Strong Adaptive Hardcore Bit Property}

The binding property of our succinct commitment scheme relies on a variant of the adaptive hardcore bit property, which we define next. 

\begin{definition}\label{def:dist-stat-HCB}
 A $\TCF$ family $(\Gen,\Eval,\Invert, \Check, \Good)$ is said to have the  {\bf distributional strong adaptive hardcore bit} property if there exists $\QPT$ algorithms $\A$ and $C$ such that the following holds:  $\A$ takes as input~$\pk$ and a quantum state $\bsigma$, and outputs a tuple $(\vecy, b,\vecx, \brho)$ such that $\Check(\pk,b,\vecx,\vecy)=1$ and  $\brho$ is a state containing at least $n+1$ qubits.  Denote by $\cO_1$ the registers containing the first $n+1$ qubits, and denote by $\cO_2$ all other registers.  $C$ takes as input the state $\brho_{\cO_2}$ and outputs $\aux\gets C(\brho_{\cO_2})$.  Denote by $\brho_\aux$ the post measurement state, and assume that 
  $\brho$ satisfies that with overwhelming probability over $\aux$, 
 for every $\vecd'\in\{0,1\}^{k}$ the probability of measuring registers $\cO_1$ of $\brho_\aux$ in the standard basis and obtaining $\vecd\in\{0,1\}^{n+1}$ such that $\Good(\vecd,\vecx_{0},\vecx_{1})=\vecd'$ is negligible in $\secp$.
Then
\begin{equation}\label{eqn:dist-hcb}
(\pk,\vecy,\vecx_{0}\oplus\vecx_{1},\vecd\cdot(1,\vecx_{0}\oplus\vecx_{1}),\aux)\approx (\pk,\vecy,\vecx_{0}\oplus\vecx_{1}, U, \aux)
\end{equation}
where  $(\pk,\sk)\gets\Gen(1^\secp)$, $(\vecy, b,\vecx,\brho)\gets\A(\pk,\bsigma)$, $((0,\vecx_{0}),(1,\vecx_{1}))=\Invert(\sk,\vecy)$, $(\aux,\brho_\aux)\gets C(\brho_{\cO_2})$, $\vecd$ is obtained by measuring registers $\cO_1$ of $\brho_\aux$ in the standard basis, and $U$ is uniformly distributed in $\{0,1\}$.

Moreover, there exists $\sk_\mathsf{pre}$, which is efficiently computable from $\sk$, such that given $(\pk,\vecy,b,\vecx_b)$ and $\sk_\mathsf{pre}$ one can efficiently compute $\vecx_{1-b}$ such that if $\Check(\pk,b,\vecx_b,\vecy)=1$ then $\Check(\pk,1-b,\vecx_{1-b},\vecy)=1$, and \Cref{eqn:dist-hcb} holds even if $C$ takes as input $\sk_\mathsf{pre}$ in addition to the state $\brho_{\cO_2}$.
\end{definition}


In \Cref{app:zvik}, we show that the construction of an $\TCF$ from \cite{BCMVV18} satisfies the distributional strong adaptive hardcore bit property.

\section{Classical Commitments to Quantum States } \label{sec:XZ-commitments}

In this section we define the notion of a classical commitment to quantum states. 
Our definition is stronger than the notion of a {\em measurement protocol}, originally considered in \cite{Mah18a} and formally defined in \cite{Bartusek22},\footnote{We  refer to this weaker notion as a ``weak classical commitment,'' and recall its  definition in  \Cref{def:WCQ-section} (for completeness).} in several ways.  First, the opening basis is not determined during the key generation phase.  Namely, the key generation algorithm, $\Gen$, takes as input only the security parameter (in unary), as opposed to taking both the security parameter and the opening basis.  In particular, the opening basis can be determined after the commitment phase, and can  be chosen adaptively based on any information that the parties have access to. 
Importantly, our binding property is significantly stronger. It guarantees that for any $\BQP$ cheating committer~$\cC^*.\Commit$ that commits to an $\ell$-qubit quantum state, there is a {\em single} extracted quantum state $\btau$ such that for {\em any} $\BQP$ algorithm $\cC^*.\Open$ and {\em any} basis opening $(b_1,\ldots,b_\ell)$, where $b_i=0$ corresponds to measuring the $i$'th qubit in the standard basis and $b_i=1$ corresponds to measuring it in the Hadamard basis, the opening obtained by $\cC^*.\Open(b_1,\ldots,b_\ell)$ 
is computationally indistinguishable from measuring $\btau$ in basis $(b_1,\ldots,b_\ell)$, assuming the opening of $\cC^*.\Open$ is accepted.
In contrast, in the soundness guarantee of a weak commitment scheme, the extracted state $\btau$ may depend on $\cC^*.\Open$, which can be chosen adaptively after the commitment phase and hence is not truly binding.\footnote{Indeed, the weak classical commitment scheme from \cite{Mah18a} is only binding in the standard basis, and offers no binding guarantees when opening in the Hadamard basis.}

\subsection{Syntax}  \label{sec:syntax}
In what follows we define the syntax of a commitment scheme. We present two definitions.  The first is the syntax for a 
 {\em non-succinct} commitment scheme, where the length of the commitment string that commits to an $\ell$-qubit quantum state, grows polynomially with $\ell$. More specifically, in this definition the length $\ell$ is determined in the key generation algorithm, and the run-time of all the algorithms grow polynomially with $\ell$.  In \Cref{sec:syntax:succinct} we define the syntax for a {\em succinct} commitment scheme, where the verifier's run-time grows poly-logarithmically with $\ell$.
\begin{definition}\label{def:SCQ-syntax}
A (non-succinct) classical commitment scheme for quantum states is associated with algorithms $(\Gen,\Commit,\Open,\Ver,\Out)$ and has the following syntax: 
\begin{enumerate}
    \item $\Gen$ is a $\ppt$ algorithm that takes as input a security parameter $\secp$ and a length parameter $\ell$  (both in unary), and outputs a pair $(\pk,\sk) \gets  \Gen(1^\secp,1^\ell)$, where $\pk$ is referred to as the {\em public key} and $\sk$ is referred to as the {\em secret key}.
    \item $\Commit$ is a $\BQP$ algorithm that takes as input a public key $\pk$ and an $\ell$-qubit quantum state~$\bsigma$ and outputs a pair $(\vecy, \brho) \gets \Commit(\pk,\bsigma)$, where $\vecy$ is a classical string referred to as the {\em commitment string} and $\brho$ is a quantum state. 

    \item $\Open$ is a $\BQP$ algorithm that takes as input a quantum state $\brho$ and a basis $(b_1,\ldots,b_\ell)\in\{0,1\}^\ell$ (where $b_j=0$ corresponds to opening the $j$'th bit in the standard basis and $b_j=1$ corresponds to opening it in the Hadamard basis). It outputs a pair $(\vecz,\brho') \gets \Open(\brho,(b_1,\ldots,b_\ell))$, where $\vecz$ is a classical string, referred to as the  {\em opening string}, and $\brho'$ is the residual state (which is sometimes omitted). 
    \item $\Ver$ is a polynomial time algorithm that takes a tuple $(\sk, \vecy, (b_1,\ldots,b_\ell), \vecz)$, where $\sk$ is a secret key, $\vecy$ is a commitment string (to the quantum state), $(b_1,\ldots,b_\ell)\in\{0,1\}^\ell$ is a string specifying the opening basis, and $\vecz$ is an opening string. It outputs $0$ (if $\vecz$ is not a valid opening) and outputs~$1$ otherwise.
    \item $\Out$ is a polynomial time algorithm that takes a tuple $(\sk, \vecy, (b_1,\ldots,b_\ell), \vecz)$ (as above), and outputs an $\ell$-bit string $\vecm\leftarrow \Out(\sk, \vecy, (b_1,\ldots,b_\ell), \vecz)$.
    
\end{enumerate}

The protocol associated with the tuple $(\Gen,\Commit,\Open,\Ver,\Out)$ is a two party protocol between a $\BQP$ committer $\cC$ and a $\PPT$ verifier $\cV$ and consists of two phases, $\COMMIT$ and $\OPEN$. During the $\COMMIT$ phase, $\cV$ takes as input security parameter $\lambda$ and a length parameter $\ell$ and $\cC$ takes in an arbitrary quantum state $\bsigma$. During the $\OPEN$ phase, $\cV$ takes as input a basis bit $(b_1,\ldots,b_\ell)\in\{0,1\}^\ell$. The protocol proceeds as follows:
\begin{itemize}
    \item $\mathsf{COMMIT}$ phase:
\begin{enumerate}
    \item $[\cC \leftarrow \cV]$: $\cV$ samples $(\pk,\sk) \gets  \Gen(1^\secp, 1^\ell)$ and sends the public key $\pk$ to $\cC$.
    \item $[\cC \rightarrow \cV]$: $\cC$ computes $(\vecy,\brho) \gets \Commit(\pk,\bsigma)$ and sends the commitment string~$\vecy$ to $\cV$.
\end{enumerate}
    \item $\mathsf{OPEN}$ phase:
\begin{enumerate}  
    \item  $[\cC \leftarrow \cV]$: $\cV$ sends an opening basis $(b_1,\ldots,b_\ell)$ to $\cC$. 
    \item $[\cC \rightarrow \cV]$: $\cC$ computes $(\vecz,\brho') \gets \Open(\brho,(b_1,\ldots,b_\ell))$ and sends $\vecz$ to $\cV$.
    \item $[\cV]$: $\cV$  checks that $\Ver(\sk, \vecy, (b_1,\ldots,b_\ell), \vecz)=1$, and if so it outputs $\vecm \gets \Out(\sk, \vecy, (b_1,\ldots,b_\ell), \vecz)$ as the decommitment. Otherwise, it outputs $\bot$.
 \end{enumerate}
\end{itemize}
\end{definition}

\begin{remark}\label{remark:bit-by-bit}
One could define $\Open,\Ver,\Out$ to operate on one qubit at a time.  Namely, one could define $\Open$ to take as input a quantum state $\brho$ an index $j\in[\ell]$ and a basis $b\in\{0,1\}$, and output a pair $(\vecz,\brho')\gets \Open(\brho,(j,b))$, and define $\Ver$ and $\Out$ to take as input $(\sk,\vecy,(j,b),\vecz)$ and output a bit (indicating accept/reject for $\Ver$ and indicating an output bit for $\Out$).  Indeed, in our  definition of a {\em succinct} classical commitment to quantum state, stated in \Cref{sec:syntax:succinct} below, $\Open$, $\Ver$ and $\Out$ operate on one qubit at a time.  In addition,  our constructions in \Cref{sec:constructions} are defined where $\Open$, $\Ver$ and $\Out$ operate on one qubit at a time.  
\end{remark}

Note that in the syntax above the length of the public key~$\pk$ as well as the length of the commitment~$\vecy$ grows with the number of qubits in the committed state (denoted by $\ell$).  In this work we also construct {\em succinct} commitments where the length of $\pk$ and the commitment $\vecy$ grow only with the security parameter (and grow only poly-logarithmically with $\ell$). 
In what follows we define the syntax of a {\em succinct} classical commitment scheme for multi-qubit quantum states.  

\subsubsection{Syntax for Succinct Commitments}\label{sec:syntax:succinct}

The syntax of a succinct commitment is similar to that of a non-succinct commitment scheme (defined above), with the following main differences:
\begin{enumerate}
    \item The key generation algorithm ($\Gen$) takes as input only the security parameter $\secp$ and does not depend on the size $\ell$ of the committed quantum state.  
    
    This change ensures that the runtime of $\Gen$ does not grow with $\ell$.\footnote{One could give $\Gen$ the parameter $\ell$ in binary, but our scheme does not require it.}
    
    \item The opening algorithm ($\Open$) opens one qubit at a time. 
    Namely, it takes as input the post-commitment quantum state $\brho$, a single index $j\in[\ell]$ and a basis $b\in\{0,1\}$, and it outputs an opening to the $j$'th qubit. 
    
    The reason for this change is that in some of our applications we commit to a long quantum state but open only a small portion of it.  For example, this is the case in our compilation of a $X/Z$ quantum PCP into a succinct argument (see \Cref{sec:applications}).  

    \item The succinct commitment has two additional components. The first is an  interactive protocol that verifies that the prover ``knows'' a non-succinct commitment string corresponding to this succinct commitment. This protocol is referred to as $\Ver.\Commit$, and is a protocol between a poly-time (classical) prover $P$ and a $\PPT$ verifier $V$.\footnote{We note that since $P$ is a classical algorithm the quantum state remains unchanged. } The second is a $\Test$ protocol that tests that the committer can open all the qubits in a valid manner.  We note that $\Test$  is executed with probability $1/2$, and if it is executed then $\Open$ is not executed (since $\Test$  destroys the quantum state needed for the $\Open$ algorithm).\footnote{This is also the case in the  $\Test$ phase in Mahadev's measurement protocol~\cite{Mah18a}.}
\end{enumerate}

Formally a succinct commitment scheme for quantum states consists of  
\[(\Gen,\Commit,\Ver.\Commit, \Test, \Open,\Ver,\Out)\] such that 
\begin{enumerate}
    \item $\Gen$ is a $\ppt$ algorithm that takes as the security parameter $\secp$ (in unary)  and outputs a pair $(\pk,\sk) \gets  \Gen(1^\secp)$.
    \item $\Commit$ is a $\BQP$ algorithm that takes as input a public key $\pk$ and an $\ell$-qubit quantum state~$\bsigma$ and outputs a tuple $(\rt, \vecy,\brho) \gets \Commit(\pk,\bsigma)$, where $\rt$ is a succinct classical commitment to $\bsigma$ (of size $\poly(\secp,\log \ell)$), $\vecy$ is its non-succinct counterpart, and $\brho$ is the residual quantum state. 
\item $\Ver.\Commit$ is an interactive protocol between a poly-time prover $P$ with input $(\pk,\rt, \vecy)$ and a $\PPT$ verifier $V$ with input $(\sk,\rt)$. At the end of the protocol, $V$ outputs a verdict bit in $\{0,1\}$, corresponding to accept or reject ($1$ corresponding to accept and $0$ corresponding to reject). The communication complexity is $\poly(\secp,\log \ell)$.
\item $\Test$ is an interactive protocol between a  $\BQP$ prover $P_\Test$ with input $(\pk,\rt, \vecy,\brho)$ and a $\BPP$ verifier $V_\Test$ with input $(\sk,\rt)$.  At the end of the protocol, $V$ outputs a verdict bit in $\{0,1\}$, corresponding to accept or reject. 
    \item $\Open$ is a $\BQP$ algorithm that takes as input a quantum state $\brho$, an index $j\in[\ell]$, and a basis $b_j\in\{0,1\}$ (where $b_j=0$ corresponds to measuring the $j$'th qubit in the standard basis and $b_j=1$ corresponds to measuring it in the Hadamard basis). It outputs a pair  $(\vecz,\brho') \gets \Open(\brho,(j,b_j))$, where $\vecz$ is a classical string of length $\poly(\secp,\log \ell)$, referred to as the  {\em opening string}, and $\brho'$ is the residual state (which is sometimes omitted). 
    \item $\Ver$ is a polynomial time algorithm that takes a tuple $(\sk, \rt, (j,b_j), \vecz)$, where $\sk$ is a secret key, $\rt$ is a succinct classical commitment string to an $\ell$-qubit quantum state, $j\in[\ell]$, $b_j\in\{0,1\}$ is a bit specifying the opening basis, and $\vecz$ is an opening string. It outputs $0$ (if $\vecz$ is not a valid opening) and outputs~$1$ otherwise.
    \item $\Out$ is a polynomial time algorithm that takes a tuple $(\sk, \rt, (j,b_j), \vecz)$, and outputs a bit $m$.
    
\end{enumerate}
\begin{remark}
    We extend $\Ver$ and $\Out$ to take as input  $(\sk, \rt, (J,\vecb_J), \vecz)$ instead of $(\sk, \rt, (j,b_j), \vecz)$, where $J\subseteq[\ell]$ and $\vecb_J\in\{0,1\}^{|J|}$, in which case the algorithms run with input $(\sk, \rt, (j,b_j), \vecz)$ for every $j\in J$. We extend $\Open$ in a similar manner.
\end{remark}
The succinct commitment protocol associated with the tuple \[(\Gen,\Commit,\Ver.\Commit,\Test,\Open,\Ver,\Out)\] is a two party protocol that consists of three  phases, $\COMMIT$, $\mathsf{CHECK}$ and $\OPEN$, as follows:
\begin{itemize}
    \item $\mathsf{COMMIT}$ phase:
\begin{enumerate}
    \item $[\cC \leftarrow \cV]$: $\cV$ samples $(\pk,\sk) \gets  \Gen(1^\secp)$ and sends the public key $\pk$ to $\cC$.
    \item $[\cC \rightarrow \cV]$: $\cC$ computes $(\rt,\vecy,\brho) \gets \Commit(\pk,\bsigma)$ and sends the succinct commitment string~$\rt$ to~$\cV$.
\end{enumerate}
\item $\mathsf{CHECK}$ phase:
\begin{enumerate}
    \item Run $\Ver.\Commit$ protocol between the prover $P$ with input $(\pk,\rt,\vecy)$ and the verifier $V$ with input $(\sk,\rt)$. If $V$ rejects then the commitment $\rt$ is rejected and the protocol ends. 
    \item Otherwise, choose at random $c\gets\{0,1\}$.
    \item If $c=0$ then go to the $\mathsf{OPEN}$ phase.
    \item If $c=1$ then run the $\Test$ protocol, where the prover $P_\Test$ takes as input $(\pk,\rt,\vecy,\brho)$ and the verifier $V_\Test$ takes as input $(\sk,\rt)$. If  $V_\Test$ rejects then the commitment $\rt$ is rejected and otherwise it is accepted.  At the end of the $\Test$ protocol the commitment protocol ends (the $\Open$ phase is not executed).

\end{enumerate}
    \item $\mathsf{OPEN}$ phase:
\begin{enumerate}  
    \item  $[\cC \leftarrow \cV]$: $\cV$ sends a subset $J\subseteq [\ell]$ and an opening basis $\vecb_J\in \{0,1\}^{|J|}$ to $\cC$. 
    \item $[\cC \rightarrow \cV]$: $\cC$ computes $(\vecz,\brho') \gets \Open(\brho,(J,\vecb_J))$ and sends $\vecz$ to $\cV$.
    \item $[\cV]$: $\cV$  checks that $\Ver(\sk, \rt, (J,\vecb_J), \vecz)=1$, and if so outputs $\vecm \gets \Out(\sk, \rt,  (J,\vecb_J), \vecz)$ as the decommitment bit. Otherwise, it outputs $\bot$.
 \end{enumerate}
\end{itemize}

\subsection{Properties}
We require that a commitment scheme satisfies two properties, {\em correctness} and {\em binding}, defined below. 

\subsubsection{Correctness}
We define the correctness guarantee separately for the non-succinct and the succinct setting, starting with the former. 
\begin{definition}[Correctness]\label{def:SCQ-correctness}
A (non-succinct) classical commitment scheme is {\em correct} if for any $\ell$-qubit quantum state $\bsigma$, and any basis $\vecb=(b_1, \dots, b_\ell) \in\{0,1\}^\ell$,
\begin{equation}
    \Real(1^\lambda, \bsigma, \vecb) \equiv \bsigma(\vecb),
\end{equation}
where $\bsigma(\vecb)$ is the distribution obtained by measuring each qubit $j$ of $\bsigma$ in the basis specified by $b_j$ (standard if $b_j = 0$, Hadamard if $b_j = 1$), and $\Real(1^\lambda, \bsigma, \vecb)$ is the distribution resulting from the following experiment:  

\begin{enumerate}
        \item Generate $(\pk, \sk) \leftarrow \Gen(1^\lambda,1^\ell)$.
        \item Generate $(\vecy, \brho) \leftarrow \Commit(\pk, \bsigma)$.
        \item Compute $(\vecz,\brho')\leftarrow \Open(\brho,\vecb)$.  
        \item If $\Ver(\sk, \vecy,\vecb, \vecz)=0$ then output $\bot$.
        \item Otherwise, output $\Out(\sk, \vecy, \vecb, \vecz)$.
\end{enumerate}
\end{definition}

\begin{definition}[Succinct Correctness]\label{def:SCQ-succinct-correctness}
    A succinct classical commitment scheme is correct if for any $\ell$-qubit quantum state $\bsigma$, any basis $\vecb = (b_1, \dots, b_\ell) \in \{0,1\}^\ell,$ and any subset $J \subseteq [\ell]$, 
the following two conditions holds:
    \begin{equation}\label{succinct-correctness-equivalence}
    \Real_{c=0}(1^\lambda, \bsigma, J, \vecb_J) \equiv \bsigma(J, \vecb_J)~~\mbox{and}~~\Real_{c=1}(1^\lambda, \bsigma, J, \vecb_J)\equiv 1
    \end{equation}
    where:
    \begin{itemize}
        \item $\bsigma(J, \vecb_J)$ is the distribution obtained by measuring each qubit $j \in J$ of $\bsigma$ in the basis specified by $b_j$ (standard if $b_j = 0$ and Hadamard if $b_j = 1$). 
    \item     $\Real_{c=0}(1^\lambda, \bsigma, J, \vecb_J$) is the distribution resulting from the following experiment: 

    \begin{enumerate}
        \item Generate $(\pk, \sk) \leftarrow \Gen(1^\lambda)$.
        \item Generate $(\rt, \vecy, \brho) \leftarrow \Commit(\pk, \bsigma)$.
        \item Run the protocol $\Ver.\Commit$ between $P$ with input $(\pk, \rt, \vecy)$ and $V$ with input $(\sk,\rt)$.  If $V$ rejects then then output $\bot$.
        \item Otherwise, compute $(\vecz,\brho')\leftarrow \Open(\brho,(J, \vecb_J))$. 
        \item If $\Ver(\sk, \rt,(J, \vecb_J), \vecz)=0$ then output $\bot$. 
        \item Otherwise, output $\Out(\sk, \rt, (J, \vecb_J), \vecz)$.
        \end{enumerate}
   
        \item $\Real_{c=1}(1^\lambda, \bsigma, J, \vecb_J$) is the distribution resulting from the following experiment: 
     
        \begin{enumerate}
        \item Generate $(\pk, \sk) \leftarrow \Gen(1^\lambda)$.
        \item Generate $(\rt, \vecy, \brho) \leftarrow \Commit(\pk, \bsigma)$.
        \item Run the protocol $\Ver.\Commit$ between $P$ with input $(\pk, \rt, \vecy)$ and $V$ with input $(\sk,\rt)$.  If $V$ rejects then then output $\bot$.
        \item Execute $\Test$ where the prover $P_\Test$ takes as input $(\pk, \rt, \vecy,\brho)$ and the verifier $V_\Test$ takes as input $(\sk, \rt)$. If  $V_\Test$ rejects then output~$\bot$ and if it accepts then output~$1$.
        \end{enumerate}
\end{itemize}
\end{definition}

\subsubsection{Binding}
In what follows we define the binding condition.  Intuitively, our binding guarantee is that a cheating committer cannot change the way they open based on \emph{any} information they learn after the commitment phase, and that the opening distribution is consistent with the distribution of a qubit.

For simplicity, we consider only cheating algorithms that are accepted with high probability.  This can be ensured by repetition.
Namely, for every $\epsilon, \delta>0$ by repeating the commitment and opening protocol $O\left(\frac{\log(1/\epsilon)}{\delta}\right)$ times, if a cheating $\cC^*$ is accepted in all of executions with probability at least $\epsilon$ then a random execution is accepted with probability at least $1-\delta$.


We first define the binding property for the (non-succinct) commitment and then define it for the succinct commitment.

\begin{definition}[Binding]\label{def:binding}
A classical (non-succinct) commitment scheme to a multi-qubit quantum state is said to be {\em computationally binding} if there exists a $\QPT$ oracle machine $\Ext$ such that for any $\BQP$ algorithm $\cC^*.\Commit$, any $\poly(\secp)$-size quantum state $\bsigma$, any polynomial $\ell=\ell(\secp)$, any basis $\vecb=(b_1,\ldots,b_\ell)$, and any $\BQP$ algorithms $\cC^*_1.\Open$ and $\cC^*_2.\Open$, for every $i\in\{1,2\}$
\begin{equation}\label{eqn:binding1} 
\Real^{\cC^*.\Commit,\cC^{*}_i.\Open}(\lambda,\vecb,\bsigma)\stackrel{\eta}\approx\Ideal^{\Ext,\cC^*.\Commit,\cC^*_i.\Open}(\lambda,\vecb,\bsigma)
\end{equation}
and 
\begin{equation}\label{eqn:binding2}
 \Real^{\cC^*.\Commit,\cC^{*}_1.\Open}(\lambda,\vecb,\bsigma)
 \stackrel{\eta}\approx\Real^{\cC^*.\Commit,\cC^{*}_2.\Open}(\lambda,\vecb,\bsigma)
\end{equation}
where $\eta=O\left(\sqrt{\delta}\right)$ and
     \begin{equation}\label{deltaj}
    \delta= \E_{\substack{(\pk,\sk) \leftarrow \Gen(1^\lambda, 1^\ell) \\ (\vecy, \brho) \leftarrow \cC^*.\Commit(\pk, \bsigma)}}\max_{\substack{i\in\{1,2\},\\ \vecb'\in\{\vecb,{\bf 0},{\bf 1\}}}}\Pr[\Ver(\sk,\vecy,\vecb',\cC_i^*.\Open(\brho,\vecb'))=0].
    \end{equation}
    and where $\Real^{\cC^*.\Commit,\cC^*.\Open}(\lambda,\vecb,\bsigma)$ is defined as follows:
\begin{itemize}
    \item $(\pk,\sk)\gets\Gen(1^\secp,1^\ell)$. 
    \item $(\vecy,\brho)\gets \cC^*.\Commit(\pk,\bsigma)$.

\end{itemize}

         \begin{enumerate}
        \item Compute $(\vecz,\brho')\leftarrow \cC^*.\Open(\brho,\vecb)$.
        \item \label{item:real:ver} If $\Ver(\sk,\vecy,\vecb,\vecz)=0$ then output $\bot$. 
        \item Otherwise, let $\vecm=\Out(\sk,\vecy, \vecb,\vecz)$.
        \item Output $(\pk,\vecy,\vecb,\vecm)$.
        \end{enumerate}
$\Ideal^{\Ext,\cC^*.\Commit,{\cC}^*.\Open}(\secp,\vecb,\bsigma)$ is defined as follows:
         \begin{enumerate}
          \item $(\pk,\sk)\gets\Gen(1^\secp,1^\ell)$. 
    \item $(\vecy,\brho)\gets \cC^*.\Commit(\pk,\bsigma)$.
            \item Let $\btau_{\cA,\cB}=\Ext^{\cC^*.\Open}(\sk, \vecy,\brho)$. 
             \item Measure $\btau_\cA$ in the basis $\vecb=(b_1,\ldots,b_\ell)$ to obtain $\vecm\in\{0,1\}^\ell$.
        \item Output $(\pk,\vecy,\vecb,\vecm)$.            
           \end{enumerate}
\end{definition}


\begin{remark}
Throughout this write-up to avoid cluttering of notation we omit the superscript $\cC^*.\Commit$ from 
\[\Real^{\cC^*.\Commit,\cC^{*}.\Open}(\lambda,\vecb,\bsigma)~~\mbox{ and }~~ \Ideal^{\Ext,\cC^*.\Commit,{\cC}^*.\Open}(\secp,\vecb,\bsigma),
\]and denote these by 
\[\Real^{\cC^{*}.\Open}(\lambda,\vecb,\bsigma)~~\mbox{ and }~~\Ideal^{\Ext,{\cC}^*.\Open}(\secp,\vecb,\bsigma),
\]respectively.  
    
\end{remark}
\begin{remark}  We prove that our commitment scheme is sound with $\eta\leq 10\sqrt{\delta}$.
    Note that $\delta$ is a bound on the probability that the openings of $\cC^*_i.\Open$ are rejected not only on basis $\vecb$, but also on basis ${\bf 0}$ and ${\bf 1}$. We note that for \Cref{eqn:binding2} we do not need to bound the probability that $\cC^*_i.\Open$ is rejected on basis ${\bf 0}$ and basis ${\bf 1}$, and indeed we do not bound these probabilities in the proof (see \Cref{lemma:soundness-single1}).  The reason we need to bound these probabilities to prove \Cref{eqn:binding1} is that our extractor uses the openings of $\cC^* _i.\Open$ on basis ${\bf 0}$ and ${\bf 1}$ to extract the quantum state.
    
\end{remark}

\begin{remark}\label{stronger-binding-remark}
We mention that we prove a stronger condition than the one given in \Cref{eqn:binding2}.  This is done in \Cref{lemma:soundness-single1}.  The strengthening is due to two reasons. 
 First, we prove \Cref{eqn:binding2} by induction on $\ell$, and for the induction step to go through we need to strengthen the induction hypothesis, and as a result we prove a stronger guarantee.  Second, we allow the cheating algorithm $\cC^*.\Open$ to depend on a part of the secret key $\sk$.  This is needed to obtain our succinct interactive argument for $\QMA$ in \Cref{sec:succint-QMA} and is needed for our applications in \Cref{sec:applications}. 
\end{remark}

\Cref{def:binding} assumes that $\cC^*.\Open$ opens all the qubits of the committed state.  Indeed, we use  $\cC^*.\Open$ to extract an $\ell$-qubit quantum state. In what follows we define the notion of binding for a succinct commitment, where $\cC^*.\Open$ may only open to a subset $J\subseteq [\ell]$ of the qubits, and hence cannot be used to extract the entire state (as was done in \Cref{def:binding}). While we can extract a state consisting of $|I|$ qubits from $\cC^*.\Open$, for our applications we will need to extract the entire $\ell$-qubit state, even if $\cC^*.\Open$ only opens to the qubits in $J$  (without blowing up the communication).  This is precisely the purpose of  the 
 $\Ver.\Commit$ protocol; instead of extracting from $\cC^*.\Open$ we extract from the (cheating) prover $P^*$ of the $\Ver.\Commit$ protocol.  We note that even though $\Ver.\Commit$ is a succinct protocol, its interactive nature will allow us to extract the (non-succinct) $\ell$-qubit state from~$P^*$.

\begin{definition}[Succinct Binding]\label{def:binding:succ}
A succinct classical commitment scheme to a multi-qubit quantum state is said to be {\em computationally binding} if there exists a $\QPT$ oracle machine $\Ext$ such that for any $\BQP$ algorithm $\cC^*.\Commit$, any $\poly(\secp)$-size quantum state $\bsigma$, any polynomial $\ell=\ell(\secp)$, any $\QPT$ prover $P^*$ for the $\Ver.\Commit$ protocol, any $\QPT$ prover $P^*_\Test$ for the $\Test$ protocol, any $J\subseteq[\ell]$ and $\vecb_J=(b_j)_{j\in J}\in\{0,1\}^{|J|}$, any $\epsilon>0$, and any $\QPT$ algorithms $\cC^*_1.\Open$ and $\cC^*_2.\Open$, for every $i\in\{1,2\}$
\begin{equation}\label{eqn:binding1:succinct} 
\Real^{\cC^*.\Commit,P^*,\cC^{*}_i.\Open}(\secp, (J,\vecb_J),\bsigma)\stackrel{\zeta}\approx\Ideal^{\Ext,\cC^*.\Commit,P^*,P^*_\Test}(\secp, (J,\vecb_J),\bsigma,\epsilon)
\end{equation}
and 
\begin{equation}\label{eqn:binding2:succinct}
 \Real^{\cC^*.\Commit,P^*,\cC^{*}_1.\Open}(\secp, (J,\vecb_J),\bsigma)
 \stackrel{\eta}\approx\Real^{\cC^*.\Commit,P^*,\cC^{*}_2.\Open}(\secp, (J,\vecb_J),\bsigma)
\end{equation}
where $\Real^{\cC^*.\Commit,P^*,\cC^*.\Open}(\secp, (J,\vecb_J),\bsigma)$ is defined as follows:
\begin{enumerate}
    \item Generate $(\pk,\sk)\gets\Gen(1^\secp)$. 
    \item Compute $(\rt,\brho)\gets \cC^*.\Commit(\pk,\bsigma)$.\footnote{Note that a malicious $\cC^*.\Commit$ may choose $\rt$ maliciously without a corresponding non-succinct commitment string $\vecy$. We assume without loss of generality that all the auxiliary information it has about $\rt$ is encoded in $\brho$.}
      \item\label{ver-commit-reject} Compute the $\Ver.\Commit$ protocol between $P^*(\pk,\rt,\brho)$ and $V(\sk,\rt)$.  If $V$ rejects then output~$\bot$.
         Denote the resulting quantum state of $P^*$ at the end of this protocol by $\brho_\mathsf{post}$.\footnote{If $P^*$ was honest then it would have been classical and hence $\brho$ would have remained unchanged.  But since we are considering a malicious $P^*$ it may alter its quantum state during the $\Ver.\Commit$ protocol.}
        \item Compute $(\vecz_J,\brho')\leftarrow \cC^*.\Open(\brho_\mathsf{post},(J,\vecb_J))$.  
        \item If $\Ver(\sk,\rt,(J,\vecb_J),\vecz_J)=0$ then output $\bot$.
        \item Otherwise, let $\vecm_J=\Out(\sk,\rt, (J,\vecb_J),\vecz)$.
        \item Output $(\pk,\rt,(J,\vecb_J),\vecm_J)$.
        \end{enumerate}
Let $\delta_0$ be the probability that at the end of the protocol $\Ver.\Commit$ the verifier rejects.  Namely, $\delta_0$ is the probability that \Cref{ver-commit-reject} above outputs $\bot$. Denote by $\brho_\mathsf{post}$ the state of $P^*$ after the $\Ver.\Commit$ protocol. Let $\delta'_{0}$ be the probability that the verifier $V_\Test(\sk,\rt)$ outputs $\bot$ in the $\Test$ protocol when interacting with $P^*_\Test(\pk,\rt,\brho_\mathsf{post})$. Let
   \begin{equation}\label{eqn:delta:succinct}
    \delta= \max_{i\in\{1,2\}}\Pr[\Ver(\sk,\rt,\vecb_J,\vecz_J)=0]~~\mbox{ and }~\vecz=\cC^*_i.\Open(\brho_\mathsf{post},\vecb_J),
   \end{equation}
   and let
    \begin{equation}\label{eqn:delta0}
    \eta=O\left(\sqrt{\delta_0+\delta} \right)~~\mbox{ and }~~\zeta=O\left(\sqrt{\delta_0+\delta'_0+\delta}\right)+\epsilon.
    \end{equation}

        $\Ideal^{\Ext,\cC^*.\Commit,P^*,P^*_\Test}(\secp, (J,\vecb_J),\bsigma,\epsilon)$ is defined as follows:
         \begin{enumerate}
         \item Generate $(\pk,\sk)\gets\Gen(1^\secp)$. 
    \item Compute $(\rt,\brho)\gets \cC^*.\Commit(\pk,\bsigma)$.
            \item Let $\btau_{\cA,\cB}=\Ext^{P^*, P^*_\Test}(\sk, \rt,\brho,1^{\lceil1/\epsilon\rceil})$. 
             \item Measure the $J$ qubits of $\btau_\cA$ in the basis $\vecb_J$ to obtain $\vecm_J\in\{0,1\}^{|J|}$.
        \item Output $(\pk,\rt,(J,\vecb_J),\vecm_J)$.            
           \end{enumerate}
    
\begin{remark}
   In the succinct soundness definition above we assume that $\ell=\ell(\secp)$ is polynomial in the security parameter.  We could also consider $\ell$ that is super-polynomial in $\secp$, in which case we will obtain binding assuming the post-quantum $\ell$-security of the $\LWE$ assumption; i.e., assuming that a $\poly(\ell)$ size quantum circuit cannot break the $\LWE$ assumption. The proof for a general (super-polynomial) $\ell$ is exactly the same as the one where $\ell=\poly(\secp)$, the only difference is that now we consider adversaries that run in $\poly(\ell)$ time.  
\end{remark}


\end{definition}

\newpage
\appendix
\section{Constructions}\label{sec:constructions}
In this section we present our constructions.  We first construct a classical commitment scheme for committing to a {\em single} qubit state. This can be found in \Cref{sec:construction:single}.
Then, we show a generic transformation that converts any single-qubit commitment scheme into a multi-qubit commitment scheme.  This can be found in \Cref{def:SCQ-construction-multi}. In this scheme the size of the public key and the size of the commitments grow with the length of the quantum state committed to.  Finally, in \Cref{sec:succinct-multi-qubit-com} we show how to construct a succinct multi-qubit commitment scheme, where the size of the public key as well as the size of the commitment grows only with the security parameter (and poly-logarithmically with the length of the quantum state committed to).  We analyze these schemes in \Cref{sec:analysis}.

\subsection{Construction for Single Qubit States}

\label{sec:construction:single}
In this subsection, we describe our commitment scheme for a quantum state that consists of a single qubit, denoted by $\alpha_0\ket{0}+\alpha_1\ket{1}$. We use as a building block the commitment algorithm $\CommitM$ from \cite{Mah18a} for the {\em multi-qubit} case.  This algorithm makes use of a $\TCF$ family 
\[
(\Gen_\TCF, \Eval_\TCF, \Invert_\TCF, \Check_\TCF, \Good_\TCF).\] 
The public key $\pk$ used by $\CommitM$ to commit to an $\ell$-qubit state is of the form $\pk=(\pk_1,\ldots,\pk_\ell)$ where each $\pk_j$ is a public key generated by $\Gen_\TCF(1^\secp)$.\footnote{We mention that \cite{Mah18a} used a dual mode $\TCF$ family, where each $\pk_i$ is generated either in an injective mode in a two-to-one mode, depending on the opening basis which is assumed to fixed ahead of time.}  The $\BQP$ algorithm $$\CommitM\left((\pk_1,\ldots,\pk_\ell), \sum_{\vecs \in \{0,1\}^\ell} \alpha_{\vecs} \ket{\vecs}_\cS\right)$$ outputs the following:
    \begin{enumerate}
        \item A measurement outcome $\vecy = (\vecy_1, \dots, \vecy_\ell)$, where each $\vecy_j \in \mathsf{R_{\pk_j}}$.
        \item A state $\ket{\varphi}$ such that 
    \begin{equation}\label{mahadev-post-commit-state}
         \ket{\varphi}  \equiv \sum_{\vecs \in \{0,1\}^\ell} \alpha_{\vecs} \ket{\vecs}_\cS \ket{\vecx_\vecs}_\cZ,
    \end{equation} 
    where $\vecx_{\vecs} = (\vecx_{s_1} \dots, \vecx_{s_\ell})$ where each $\vecx_{s_j} \in \mathsf{D_{\pk_j}}$ and is such that  
    \[\vecy_{j} \in \Supp(\Eval(\pk_j, s_j, \vecx_{s_j})).
    \]
    \end{enumerate}

\begin{construction}[Commitment Scheme]\label{construction:single}

Our construction uses a noisy trapdoor claw-free ($\TCF$) function family $(\Gen_\TCF, \Eval_\TCF, \Invert_\TCF, \Check_\TCF, \Good_\TCF)$ and the algorithm
$\CommitM$ defined above.
Our algorithms are defined as follows: 
\begin{itemize}
    \item $\Gen(1^\lambda): $ 
    \begin{enumerate}
        \item For every $i\in\{0,1,\ldots,n+1\}$ sample $(\pk_i, \sk_i) \leftarrow \Gen_{\TCF}(1^\lambda)$, where $n=n(\secp)$ is such that the domain of each trapdoor claw-free function is a subset of $\{0,1\}^n$.
        \item Let $\pk=(\pk_0,\pk_1,\ldots,\pk_{n+1})$  and $\sk=(\sk_0,\sk_1,\ldots,\sk_{n+1})$.
        \item Output $(\pk,\sk)$.
    \end{enumerate} 
    \item $\Commit(\pk, \alpha_0\ket{0}+\alpha_1\ket{1}):$ 
    \begin{enumerate}
    \item Parse $\pk=(\pk_0,\pk_1,\ldots,\pk_{n+1})$ 
        \item Compute $(\vecy_0,\ket{\varphi_0})\leftarrow\CommitM(\pk_0, \alpha_0\ket{0}+\alpha_1\ket{1})$, where 
    \begin{equation*}\label{committed-state-1}
        \ket{\varphi_0}_{\cS, \cZ} \equiv   \sum_{s\in\{0,1\}}\alpha_s\ket{s}_{\cS} \ket{\vecx_s}_{\cZ}
    \end{equation*}
    Here, $\vecx_s \in \{0,1\}^n$ and $\vecy_0 \in \Supp(\Eval_\TCF(\pk_0, s, \vecx_s))$ for every $s \in \{0,1\}$.
    Note that register $\cS$ consists of $1$ qubit and $\cZ$ consists of $n$ qubits. 
 \item Apply the Hadamard unitary $H^{\tensor(n+1)}$ to $\ket{\varphi_0}$ to obtain
 \begin{align*}
\ket{\varphi_1}_{\cS, \cZ} &= H^{\tensor(n+1)}\ket{\varphi_0} \\
&\equiv \frac{1}{\sqrt{2^{n+1}}}\sum_{\vecd\in\{0,1\}^{n+1}}\underbrace{(-1)^{\vecd\cdot(0,\vecx_0)}(\alpha_0+(-1)^{\vecd\cdot (1,\vecx_0\oplus \vecx_1)}\alpha_1)}_{\beta_{\vecd}} \ket{\vecd}_{\cS, \cZ}
 \end{align*} 
 
\item Apply the algorithm $\CommitM$ with $\pk_1$ to register $\cS$ of  $\ket{\varphi_1}_{\cS, \cZ}$, and for every $i\in\{2,\ldots,n+1\}$ apply $\CommitM$ with $\pk_i$ to register $\cZ_i$ of  $\ket{\varphi_1}_{\cS, \cZ}$. Denote the output by  $(\vecy_1,\ldots,\vecy_{n+1})$ and the resulting state by 
\begin{equation*}
\ket{\varphi_2}_{\cS, \cZ, \cZ'} \equiv \frac{1}{\sqrt{2^{n+1}}}\sum_{\vecd \in\{0,1\}^{n+1}}\beta_\vecd \ket{\vecd}_{\cS, \cZ} \ket{\vecx'_{1,d_1},\vecx'_{2, d_2},\ldots,\vecx'_{n+1, d_{n+1}}}_{\cZ'}
\end{equation*}
where 
for every $i\in\{1,\ldots,n+1\}$ and every $d_i\in\{0,1\}$,
$$\vecy_i \in \Supp(\Eval_\TCF(\pk_i, d_i, \vecx'_{i,d_i})).$$
Note that the $\cZ'$ register consists of $n\cdot(n+1)$ qubits, and we partition these qubits to registers $\cZ'_1,\ldots,\cZ'_{n+1}$, each consisting of $n$ qubits. 

\item Rename the register $\cS$ to $\cZ_1$ and split the register $\cZ$ into registers $\cZ_2, \dots, \cZ_{n+1}$ of $1$ qubit each.
Permute the registers to obtain a state $\ket{\varphi_3}$ such that
\begin{equation*}
    \ket{\varphi_3} \equiv \frac{1}{\sqrt{2^{n+1}}} \sum_{\vecd \in \{0,1\}^{n+1}} \beta_{\vecd} \ket{d_1}_{\cZ_1} \ket{\vecx'_{1,d_1}}_{\cZ'_1}  \dots \ket{d_{n+1}}_{\cZ_{n+1}} \ket{\vecx'_{n+1, d_{n+1}}}_{\cZ'_{n+1}}
\end{equation*}

\item Output $(\vecy_0,\vecy_1,\ldots,\vecy_{n+1})$ and $\ket{\varphi_3}$.

\end{enumerate}

\item $\Open(\ket{\varphi},b)$:

    If $b=1$ (corresponding to an opening in the Hadamard basis) then output the measurement of  $\ket{\varphi}$ in the standard basis, and if $b=0$ (corresponding an opening in the standard basis) then output the measurement of $\ket{\varphi}$ in the Hadamard basis.  

\item $\Ver(\sk,\vecy,b,\vecz)$: 
\begin{enumerate}
\item Parse $\sk=(\sk_0,\sk_1,\ldots,\sk_{n+1})$.
\item Parse $\vecy=(\vecy_0,\vecy_1,\ldots,\vecy_{n+1})$.

\item If $b=1$ then do the following: 
\begin{enumerate}
    \item Parse $\vecz=(d_1, \vecx'_1, \dots, d_{n+1}, \vecx'_{n+1})$  and let $\vecd=(d_1,\ldots,d_{n+1})$.
    \item 
    Compute  $((0,\vecx_{0}),(1,\vecx_{1}))=\Invert_{\TCF}(\sk_0,\vecy_0)$.
    \item \label{item:Ver-abort-1}  Verify that 
    \begin{itemize}
        \item  $\Check_\TCF({\pk_i, d_i, \vecx'_i, \vecy_i}) = 1$ for every $i\in\{1,\ldots,n+1\}$. 
        \item  $\vecd \in \Good_{\vecx_{0}, \vecx_{1}}$.
    \end{itemize}
    If any of these checks does not hold output~$0$ and otherwise output~$1$.
\end{enumerate}
\item If $b=0$ then do the following:
\begin{enumerate}
    \item Parse $\vecz=(\vecz_1,\ldots,\vecz_{n+1})$ where each $\vecz_i \in \{0,1\}^{n+1}$.
    \item For every $i\in\{1,\ldots,n+1\}$ compute $((0,\vecx'_{i,0}),(1,\vecx'_{i,1}))=\Invert_{\TCF}(\sk_i,\vecy_i)$.
    \item If there exists $i\in[n+1]$ such that $\vecz_i \notin \Good_{\vecx'_{i, 0}, \vecx'_{i, 1}}$ then output~$0$.
    \item \label{item:comp:m} Else, for every $i\in[n+1]$ let $m_i=\vecz_i\cdot (1,\vecx'_{i,0}\oplus \vecx'_{i,1})$.   

    \item \label{item:Ver-abort-0} 
    
    If $\Check_\TCF(\pk_0, m_{1}, (m_2, \dots, m_{n+1}), \vecy_0) \neq 1$, output~$0$. Else, output~$1$
\end{enumerate}
\end{enumerate}

\item $\Out(\sk, \vecy,b,\vecz)$: 
\begin{enumerate}
\item Parse $\sk=(\sk_0,\sk_1,\ldots,\sk_{n+1})$.
\item Parse $\vecy=(\vecy_0,\vecy_1,\ldots,\vecy_{n+1})$.

    \item If $b = 1$:
 \begin{enumerate}   
       \item Compute  $((0,\vecx_{0}),(1,\vecx_{1}))=\Invert_{\TCF}(\sk_0,\vecy_0)$.
        \item Parse $\vecz=(d_1, \vecx'_1, \dots, d_{n+1}, \vecx'_{n+1})$ and let $\vecd=(d_1,\ldots,d_{n+1})$.
        \item Output $m=\vecd\cdot (1,\vecx_{0}\oplus \vecx_{1})$.
    \end{enumerate}
    \item If $b = 0$:
    \begin{enumerate}
    \item Parse $\vecz=(\vecz_1,\ldots,\vecz_{n+1})$.
    \item Compute $((0,\vecx'_{1,0}),(1,\vecx'_{1,1}))=\Invert_{\TCF}(\sk_1,\vecy_1)$.
    \item Output $m_1=\vecz_1\cdot (1,\vecx'_{1,0}\oplus \vecx'_{1,1})$.
    \end{enumerate}
\end{enumerate}
\end{itemize}
\end{construction}

\subsection{Construction of Commmitments for Multi-Qubit States}
\label{def:SCQ-construction-multi}
There are two ways one can extend our single-qubit commitment scheme to the multi-qubit setting. The first is to commit to an $\ell$-qubit state qubit-by-qubit by generating $\ell$ key pairs and using the $i$'th key pair to commit and open to the $i$'th qubit.  
This construction results with key size and commitment size that grow linearly with $\ell$, and is presented below.
The second approach is to extend our single qubit commitment scheme to a {\em succinct} multi-qubit commitment scheme. This is done in two steps.  First, we construct a {\em semi-succinct} multi-qubit commitment scheme, which is the same as the non-succinct one, except that we generate a {\em single} key pair $(\pk,\sk)$ and commit to each of the $\ell$ qubits using the same public key $\pk$. This results with a commitment string $(\vecy_1,\ldots,\vecy_\ell)$.  Then we show how to convert the semi-succinct commitment scheme  into a succinct one.  We elaborate on this approach in \Cref{sec:succinct-multi-qubit-com}.

\begin{construction}[Scheme for Multi-Qubit States] \label{construction:multi}

Given any single-qubit commitment scheme $$(\Gen_1, \Commit_1, \Open_1, \Ver_1, \Out_1)$$
we construct a multi-qubit commitment scheme 
consisting of algorithms 
\[
(\Gen, \Commit, \Open, \Ver, \Out)
\]
defined as follows, where we define $(\Open, \Ver, \Out)$ to operate one qubit at a time (see \Cref{remark:bit-by-bit}): 
\begin{itemize}
    \item $\Gen(1^\lambda, 1^\ell)$: 
    \begin{enumerate}
        \item For every $i\in [\ell]$ sample $(\pk_i, \sk_i) \leftarrow \Gen_1(1^\lambda)$.
        \item Let $\pk=(\pk_1, \dots, \pk_\ell)$  and $\sk=(\sk_1, \dots ,\sk_\ell)$.
        \item Output $(\pk,\sk)$.
    \end{enumerate} 
    \item $\Commit(\pk, \bsigma)$: 
    \begin{enumerate}
    \item Parse $\pk=(\pk_1, \dots, \pk_\ell)$.
\item We assume that $\bsigma$ is an $\ell$-qubit state, and we denote the $\ell$ registers of $\bsigma$ by $\cS_1,\ldots,\cS_\ell$.
\item  Execute the following steps: 
        \begin{enumerate}
                \item Let $\brho_0 =\bsigma$.
        \item For every $j \in \{1,\ldots ,\ell\}$, apply $\Commit_1$ with key $\pk_j$ to register $\cS_j$ of the state $\brho_{j-1}$, obtaining an outcome $\vecy_j$ and a post-measurement state $(\brho_{j})_{\cS_1, \dots, \cS_{\ell}, \cZ_{1} \dots, \cZ_{j}}$. 
        \end{enumerate}
        
    \item Output $(\vecy, (\brho_\ell)_{\cS_1\dots \cS_\ell, \cZ_1, \dots \cZ_\ell})$, where $\vecy = (\vecy_1, \dots, \vecy_\ell)$.
\end{enumerate}

\item $\Open(\brho_{\cS_1\dots \cS_\ell, \cZ_1, \dots \cZ_\ell},(j,b_j))$: 
\begin{enumerate}
    
    \item Apply $\Open_1$ with basis $b_j$ to registers $\{\cS_j, \cZ_j\}$ of $\brho_{\cS_1\dots \cS_\ell, \cZ_1, \dots \cZ_\ell}$, obtaining an outcome $\vecz_j$ and post-measurement state $\brho'_{j}$ .
    \item Output $(\vecz_j,\brho'_j)$. 
\end{enumerate}

\item $\Ver(\sk,\vecy, (j,b_j), \vecz_j)$: 
\begin{enumerate}
    \item Parse $\sk = (\sk_1, \dots, \sk_\ell)$ and $\vecy = (\vecy_1, \dots, \vecy_\ell)$.
    \item Output $\Ver_1(\sk_j, \vecy_j, b_j, \vecz_j)$. 
\end{enumerate}

\item $\Out(\sk, \vecy, (j,b_j), \vecz_j)$: 

\begin{enumerate}
    \item Parse $\sk = (\sk_1, \dots, \sk_\ell)$, $\vecy = (\vecy_1, \dots, \vecy_\ell)$.
    \item Output $m_j \leftarrow \Out_1(\sk_j, \vecy_j, b_j, \vecz_j)$.
\end{enumerate}
\end{itemize}
\end{construction}


We consider a {\em semi-succinct} version of the commitment scheme described above, which  is used as a stepping stone for proving soundness of the fully succinct commitment scheme constructed in \Cref{sec:succinct-multi-qubit-com}, below.

\begin{definition}\label{remark:semi-succinct}
    A semi-succinct classical commitment scheme to quantum states is a commitment scheme obtained from a single qubit commitment scheme by applying the algorithms qubit-by-qubit, as in \Cref{construction:multi},  but with a  single public key $\pk$. Namely, it is similar to \Cref{construction:multi} but where $\Gen$ generates a single key pair $(\pk,\sk)\gets \Gen_1(1^\secp)$ (as opposed to $\ell$ such pairs), and sets $(\pk_i,\sk_i)=(\pk,\sk)$ for every $i\in[\ell]$. 
\end{definition}

\begin{remark}\label{remark:ss-with-repetition}
    In \Cref{sec:binding:single}, when we prove that our (non-succinct and semi-succinct) multi-qubit commitment scheme satisfies the binding property, we  assume that $\cC^*.\Open$ successfully opens all the $\ell$ qubits in the standard basis with high probability, and successfully  opens all the $\ell$ qubits in the Hadamard basis with high probability $1-\delta$. Namely, we assume that for every $b\in\{0,1\}$, $\cC^*.\Open\left(\brho,b^\ell\right)$ generates an accepting opening with probability~$1-\delta$. To ensure that this assumption holds, we repeat the commitment phase  $O(\lfloor{1/\delta\rfloor})$ times.\footnote{We assume that $\cC.\Commit$ can generate $\bsigma^{\tensor \lfloor{1/\delta}\rfloor}$.} For each of these commitments, we  ask $\cC^*.\Open$ with probability $1/3$ to open all the qubits in the standard basis, with probability $1/3$ to open all the qubits in the Hadamard basis, and with probability $1/3$ to open in desired basis $\vecb\in\{0,1\}^\ell$.  
    
    We emphasize that even if $\cC^*.\Open$ opens only a small subset of the qubits, we still require that $\cC^*.\Open\left(\brho,b^\ell\right)$ generates a valid opening for every $b\in\{0,1\}$.  The reason is that our binding property states that there exists an extractor $\cE$  that uses any $\cC^*.\Open$ to extract a state $\btau$.  We want the guarantee that even if $\cC^*.\Open$ only opens a small subset, still the extractor can extract an $\ell$-qubit state $\btau$.  This is important in some applications, such as compiling any quantum $X/Z$ $\mathsf{PCP}$ into a succinct interactive argument (see \Cref{sec:app:QPCP}).
    Importantly, we need to do this without increasing the  communication complexity in the opening phase, and in particular it should not grow with~$\ell$.
    In what follows we  show how this can be done succinctly.

\end{remark}
\subsection{Construction of Succinct Multi-Qubit Commitments}\label{sec:succinct-multi-qubit-com}

Before we present our construction of a succinct multi-qubit commitment scheme, we define one of the main building blocks used in our construction.

\paragraph{State-Preserving Succinct Arguments of Knowledge}
\label{sec:state-preserving}

Our scheme uses a state-preserving succinct argument of knowledge system, defined and constructed in \cite{LombardiMS22}.

\begin{definition}\label{def:state-preserving-aok}\cite{LombardiMS22}
  A publicly verifiable argument system $\Pi$ for an $\mathsf{NP}$ language $\cL$ (with witness relation $\cR$) is an $\epsilon$-state-preserving succinct argument-of-knowledge if it satisfies the following properties.
  
  \begin{itemize}
      \item {\bf Succinctness}: when invoked on a security parameter $\secp$, instance size $n$, and a relation~$\cR$ decidable in time $T$, the communication complexity of the protocol is $\poly(\secp, \log T)$. The verifier computational complexity is $\poly(\secp, \log T) + \tilde O(n)$. 
      \item {\bf State-Preserving Extraction}. There exists an extractor $\cE$, with oracle access to a cheating prover $P^*$ and a corresponding quantum state $\brho$, with the following properties:
      \begin{itemize}
          \item {\bf Efficiency}: $\cE$ on input $(\vecx, 1^\secp,\epsilon)$ runs in time $\poly(|\vecx|, \secp, 1/\epsilon)$, and outputs a classical transcript $\mathbb{T}_\Sim$, a classical string $\vecw$, and a residual state $\brho_\Sim$.
          \item {\bf State-preserving}: The following two games are $\varepsilon$-indistinguishable to any $\BQP$ distinguisher:
          \begin{itemize}
              \item{\bf Game 0} (Real): Generate a transcript $\mathbb{T}$ by running $P^*(\brho,\vecx)$ with the honest verifier $V$. Output $\mathbb{T}$  along with the residual state $\brho'$.
              \item {\bf Game 1} (Simulated): Generate $((\mathbb{T}_\Sim,\vecw),\brho_\Sim) \gets \cE^{P^*,\brho}(\vecx,1^\secp,\epsilon)$. Output $(\mathbb{T}_\Sim,\brho_\Sim)$.
          \end{itemize}
          
          \item {\bf Extraction correctness}: for any $P^*$ as above, the probability that $\mathbb{T}_\Sim$ is an accepting transcript but $\vecw$ is \emph{not} in $\cR_\vecx$ is at most $\epsilon + \negl(\secp)$. 
          
      \end{itemize}
  \end{itemize}
\end{definition}
The following is an immediate corollary.
\begin{corollary}\label{cor:LMS}
    An $\epsilon$-state-preserving succinct argument-of-knowledge protocol for an $\mathsf{NP}$ language $\cL$ (with witness relation $\cR$) and extractor $\cE$ satisfies that for every cheating prover $P^*$  and a corresponding quantum state $\brho$, and every $\vecx$, if $P^*(\brho,\vecx)$ convinces the verifier $V$ to accept with probability $1-\delta$ then for
    \[
    ((\mathbb{T}_\Sim,\vecw),\brho_\Sim) \gets \cE^{P^*,\brho}(\vecx,1^\secp,\epsilon),
    \]
    $\Pr[(\vecx,\vecw)\in{\cal R}]\geq 1-\delta-2\epsilon-\negl(\secp)$.
\end{corollary}

\begin{theorem}[\cite{LombardiMS22}]\label{thm:kilian}
Assuming the post-quantum $\poly(\secp, 1/\epsilon)$ hardness of learning with errors, there exists a (4-message, public coin) $\epsilon$-state preserving succinct argument of knowledge for $\NP$. 
\end{theorem}

\subsubsection{Construction}
We are now ready to present our construction of a succinct classical commitment for multi-qubit quantum states. Our construction uses the following ingredients:

\begin{itemize}
    \item Collapsing hash family with local opening  $(\Gen_{\mathsf{H}},\Eval_{\mathsf{H}},\Open_\mathsf{H},\Ver_\mathsf{H})$, as defined in \Cref{def:HT,def:collapsing:local}.

    \item A semi-succinct commitment scheme $(\Gen_\ss, \Commit_\ss, \Open_\ss, \Ver_\ss, \Out_\ss)$, as defined in \Cref{remark:semi-succinct}, corresponding to an underlying single-qubit commitment scheme 
    \[(\Gen_1, \Commit_1, \Open_1, \Ver_1, \Out_1)
    \]
   \item An $\epsilon$-state-preserving succinct argument of knowledge protocol $(P,V)$, as defined in \Cref{def:state-preserving-aok}, for the $\NP$ languages $\cL^*$ and $\cL^{**}$ with a corresponding $\NP$ relations $\cR_{\cL^*}$ and $\cR_{\cL^{**}}$, respectively, defined as follows:
    \begin{equation}\label{eqn:def:L*}
    ((\hk,\rt),\vecy)\in \cR_{\cL^*}~\mbox{ if and only if }~\Eval_{\mathsf{H}}(\hk,\vecy)=\rt
    \end{equation} 
    and 
    \[((\sk_1,\hk,\rt,\rt',b),(\vecy,\vecz))\in \cR_{\cL^{**}}\] 
        if and only if 
     \[
   \Eval_{\mathsf{H}}(\hk,\vecy)=\rt~\wedge~ 
\Eval_{\mathsf{H}}(\hk,\vecz)=\rt'~\wedge~ \Ver_\ss\left(\sk_1,\vecy,b^\ell,\vecz\right)=1. 
    \] 
   \end{itemize}

\begin{construction}[Succinct Commitment to Multi-Qubit Quantum States]\label{def:SCQ-succinct-construction-multi}

In what follows we use the ingredients above to construct a succinct commitment scheme to  multi-qubit quantum states.

\begin{itemize}
    \item $\Gen(1^\lambda)$: 
    \begin{enumerate}
        \item Sample $(\pk_1, \sk_1) \leftarrow \Gen_1(1^\lambda)$.
        \item Sample $\hk\gets\Gen_{\mathsf{H}}(1^\secp)$.
        \item Let $\pk=(\pk_1,\hk)$ and $\sk=(\sk_1,\hk)$.
        \item Output $(\pk,\sk)$.
    \end{enumerate} 
    \item $\Commit(\pk, \bsigma)$: 
    \begin{enumerate}
       \item Parse $\pk=(\pk_1,\hk)$.
        \item Compute $(\vecy,\brho)\gets \Commit_\ss(\pk_1,\bsigma)$.
        \item Let $\rt=\Eval_{\mathsf{H}}(\hk,\vecy)$.

    \item Output $(\rt,\vecy, \brho)$.

\end{enumerate}
\item \label{item:ver.commit1} $\Ver.\Commit$ runs the $\epsilon$-state-preserving succinct argument of knowledge protocol for the $\NP$ language $\cL^*$, where $P$ and $V$ take as input the instance $(\hk,\rt)$ and $P$ takes an additional input the witness~$\vecy$.\footnote{Note that in this protocol both the prover and the verifier are classical.} If $V$ rejects then this commitment string~$\rt$ is declared invalid, and the protocol aborts. 

\item $\Test$ is an interactive protocol between $P_\Test$ with input  $(\pk,\rt,\vecy,\brho)$ and $V_\Test$ with input $(\sk,\rt)$ that proceeds as follows:
    \begin{enumerate}
    \item $V_\Test$ samples a uniformly random bit $b\gets\{0,1\}$, and sends $b$ to $P_\Test$.
     \item  $P_\Test$ generates $\vecz\gets \Open_\ss\left(\brho,b^\ell\right)$ and sends  $\rt'=\Eval_\mathsf{H}(\hk,\vecz)$ to $V_\Test$.
       \item Run the $\epsilon$-state-preserving succinct argument of knowledge protocol for the $\NP$ language $\cL^{*}$, where $P$ and $V$ take as input the instance $(\hk,\rt')$ and $P$ takes the additional input $\vecz$. If $V$ rejects then this commitment is declared invalid and the protocol aborts. 
       \item $V_\Test$ sends $\sk_1$ to the prover.
        \item Run the $\epsilon$-preserving succinct argument of knowledge protocol $(P,V)$ for the $\NP$ language $\cL^{**}$, where $P$ and $V$ take as input the instance  $(\sk_1,\hk,\rt,\rt',b)$ and $P$ takes as additional input the witness $(\vecy,\vecz)$. 
       If $V$ rejects then this commitment is declared invalid. 
       \item If $V$ accepts in both steps 2 and 4 above then $V_\Test$ outputs~$1$, and otherwise it outputs~$0$.
         
    \end{enumerate}
\begin{remark} \yael{check}
Note that the $\Test$ protocols runs the $\epsilon$-state-preserving succinct argument of knowledge {\em twice}: once to prove knowledge of~$\vecz$ and once to prove knowledge of $(\vecy,\vecz)$.  It may seem that it suffices to run this protocol once, since there is no need to prove knowledge of $\vecz$ twice.  However, in the first protocol the cheating prover does not know $\sk_1$ and hence the  $\vecz$ that is extracted, using the extractor from \Cref{def:state-preserving-aok},  can be efficiently computed without knowing $\sk_1$. Then we use the security guarantee of the underlying collision resistant hash family to argue that the vector $\vecz$ extracted from the second protocol is identical to that extracted from the first protocol, and hence can be computed efficiently without knowing $\sk_1$.  This fact is crucial for the soundness proof to go through.
\end{remark}

\item $\Open((\brho,\hk,\vecy),(j,b_j))$: 
\begin{enumerate}
\item Parse $\vecy=(\vecy_1,\ldots,\vecy_\ell)$.
\item Compute $\veco_j=\Open_\mathsf{H}(\hk,\vecy,j)$.\footnote{$\Open_\mathsf{H}(\hk,\vecy,j)$ denotes an opening to the $j$'th chunk of the preimage $\vecy$, consisting of $\vecy_j$ which is the commitment to the $j$'th qubit.}
\item Compute $(\vecz_j,\brho')\gets \Open_\ss(\brho,(j,b_j))$.

    \item Output $((\vecy_j,\veco_j,\vecz_j),\brho')$.
\end{enumerate}

\item $\Ver(\sk,\rt, (j,b_j),(\vecy_j,\veco_j, \vecz_j))$: 
\begin{enumerate}
    \item Parse $\sk = (\sk_1, \hk)$.
    \item Let $v_0=\Ver_{\mathsf{H}}(\hk,\rt,j,\vecy_j,\veco_j)$.\footnote{ $\Ver_{\mathsf{H}}(\hk,\rt,j,\vecy_j,\veco_j)$ denotes the verification of the opening for $\vecy_j$, which is the $j$'th chunk of the hashed preimage.}
    \item Let $v_1=\Ver_1(\sk_1, \vecy_j, b_j, \vecz_j)$.
    \item Output $v_0~\wedge~v_1$.
\end{enumerate}

\item $\Out(\sk,\rt, (j,b_j),(\vecy_j,\veco_j, \vecz_j))$: 

\begin{enumerate}
    \item Parse $\sk = (\sk_1, \hk)$.
    \item Output $m_j =\Out_1(\sk_1, \vecy_j, b_j, \vecz_j)$.
    
\end{enumerate}
\end{itemize}
\end{construction}

\section{Analysis of the Multi-Qubit Commitment Schemes from \Cref{sec:constructions}}
\label{sec:analysis}

\subsection{Correctness}\label{correctness-proof}

In this section we prove the correctness of \Cref{construction:multi} and \Cref{def:SCQ-succinct-construction-multi}. 

\begin{theorem}\label{def:SCQ-correctness-theorem}
   The multi-qubit commitment scheme described in \Cref{construction:multi} satisfies the correctness property given in \Cref{def:SCQ-correctness}.  
\end{theorem}

 \Cref{def:SCQ-construction-multi} commits to each qubit of a multi-qubit state independently by using the single-qubit protocol given in construction \Cref{construction:single} as a black-box. Therefore, to prove \Cref{def:SCQ-correctness-theorem} it suffices to prove the following theorem.
 
\begin{theorem}\label{def:SCQ-correctness-theorem-single}
   The single-qubit commitment scheme described in \Cref{construction:single} satisfies the correctness property given in \Cref{def:SCQ-correctness}.  
\end{theorem}

We make use of the following lemma about $\CommitM$ throughout the proof.

\begin{lemma}[Correctness of $\CommitM$]\label{lem:mahadev-correctness} For any $\ell$-qubit quantum state $\ket{\varphi} = \sum_{\vecs \in \{0,1\}^\ell} \alpha_\vecs \ket{\vecs}_\cS$ and any basis $\vecb = (b_1, \dots, b_\ell) \in \{0\}^\ell \cup \{1\}^\ell$,

\begin{equation}
    \Real_W(1^\lambda, \ket{\varphi}, \vecb) \stackrel{\negl(\lambda)}{\equiv} \bsigma(\vecb)
\end{equation}
where $\bsigma(\vecb)$ is the distribution obtained by measuring each qubit $j$ of $\ket{\varphi}$ in the basis specified by $b_j$ (standard if $b_j = 0$, Hadamard if $b_j = 1$), and $\Real_W(1^\lambda, \ket{\varphi}, \vecb)$ is the distribution resulting from honestly opening the commitment. Specifically, $\Real_W(1^\lambda, \ket{\varphi}, \vecb)$ is defined by:

\begin{enumerate}
    \item For every $i \in \{0,1,\dots, \ell\}$, sample $(\pk_i, \sk_i) \leftarrow \Gen_\TCF(1^\lambda)$. 
    \item Compute $(\vecy = (\vecy_1, \dots, \vecy_\ell), \ket{\varphi'}) \leftarrow \CommitM((\pk_1, \dots, \pk_\ell), \ket{\varphi})$, where $\ket{\varphi'}$ is of the same form as \Cref{mahadev-post-commit-state}.
    \item If $\vecb = \{0\}^\ell$:
    \begin{enumerate}
        \item Measure $\ket{\varphi'}$ in the standard basis to get $\vecz = (\vecz_1, \dots, \vecz_\ell)$. Parse each $\vecz_i = (s_i, \vecx_{i, s_i})$. 
        
        If $\Check_\TCF(\pk_i, s_i, \vecx_{i, s_i}, \vecy_i) \neq 1$ for some $i \in [\ell]$, output $\bot$. Otherwise, output $\vecs = (s_1, s_2, \dots, s_\ell)$.
    \end{enumerate}
    \item If $\vecb = \{1\}^\ell$:
    \begin{enumerate}
        \item Measure $\ket{\varphi'}$ in the Hadamard basis to get $\vecd = (\vecd_1, \dots, \vecd_\ell) \in \{0,1\}^{\ell(n+1)}$. For each $j \in [\ell]$, compute $((0, \vecx_{j, 0}), (1, \vecx_{j, 1})) = \Invert_\TCF(\sk_j, \vecy_j)$. If $\vecd_j \notin \Good_{\vecx_{j,0}, \vecx_{j,1}}$ for some $j \in [\ell]$, output $\bot$. Otherwise, output $\vecm = (m_1, \dots, m_\ell)$, where each $m_j = \vecd_j \cdot (1, \vecx_{j,0} \oplus \vecx_{j,1})$.
    \end{enumerate}
\end{enumerate}
\end{lemma}

\begin{proof}
    This follows directly from the proof of Lemma 5.3 in \cite{Mah18a}, where the $\vecb = \{0\}^\ell$ case corresponds to the \emph{Test} round and the $\vecb = \{1\}^\ell$ case corresponds to the \emph{Hadamard} round.
\end{proof}

Recall that in \Cref{construction:single}, the  final state $\ket{\phi_3}$
is the result of applying $\CommitM$ to the state
\[ \ket{\phi_2} = \frac{1}{\sqrt{2^{n+1}}} \sum_{\vecd \in \{0,1\}^{n+1}} \beta_\vecd \ket{\vecd}. \]
in the commitment procedure (pre-measurement)


We now show that the outcome of opening $\ket{\varphi_3}$ in a basis $b \in \{0,1\}$ is statistically indistinguishable from measuring the initial state, $\ket{\psi}$, in the basis $b$. We proceed with the proof for pure states, which extends to the case of mixed and entangled states by linearity, and show correctness for each basis separately.  We treat the correctness of $\CommitM$ as a black-box. Namely, we make use of \Cref{lem:mahadev-correctness} throughout the proof.

\begin{lemma}[Opening in the Hadamard basis, $b = 1$]\label{def:SCQ-correctness-hadamard}
    For any pure single-qubit quantum state $\ket{\psi} = \alpha_0 \ket{0} + \alpha_1 \ket{1}$ and any $\TCF$ family, the distribution over the outcomes of the following two experiments are statistically indistinguishable under \Cref{construction:single}:
    \begin{itemize}
        \item \textbf{Experiment 1.} Measure $\ket{\psi}$ in the Hadamard basis and report the outcome.
        \item \textbf{Experiment 2.} Execute $\Real(1^\lambda, \ket{\psi}, b_1 = 1)$, as described in \Cref{def:SCQ-correctness}.
    \end{itemize}
\end{lemma}

\begin{proof}
 By inspection, it can be seen that the distribution of outcomes obtained from $\Real(1^\lambda, \ket{\psi}, b_1 = 1)$ here is the same as the outcome obtained from the following procedure:
 \begin{enumerate}
     \item Generate keys $(\sk_0, \pk_0)$.
     \item Apply the weak commitment once to get $\vecy_0, \ket{\phi_1} \leftarrow \Commit_W(\pk_0, \ket{\psi})$.
     \item Apply a Hadamard transform to the state to get
     $\ket{\phi_2} = H^{\otimes (n+1)} \ket{\phi_1}$.
     \item Execute $\Real_W(1^\lambda, \ket{\phi_2}, 0^{n+1})$ to obtain an outcome $\vecd = (d_0, \dots, d_{n+1})$.
     \item Report an outcome $\vecd \cdot (1, \vecx_0 \oplus \vecx_1)$, where $\{(b,\vecx_b)\}_{b = 0,1} = \Invert_\TCF(\sk_0, \vecy_0)$.
 \end{enumerate}
 By \Cref{lem:mahadev-correctness}  for standard basis openings, the distribution over $\vecd$ statistically close to the distribution obtained by measuring $\ket{\phi_2}$ in the standard basis. This, in turn, by construction is equal to the distribution obtained by measuring $\ket{\phi_1}$ in the Hadamard basis. Finally, by applying \Cref{lem:mahadev-correctness} again, this time for Hadamard basis openings, this implies that the distribution of $\vecd \cdot (1, \vecx_0 \oplus \vecx_1)$ is statistically close to the distribution obtained by measuring $\ket{\psi}$ in the Hadamard basis.

\end{proof}

\begin{lemma}[Opening in the standard basis, $b = 0$]\label{def:SCQ-correctness-standard}
For any pure single-qubit quantum state $\ket{\psi} = \alpha_0 \ket{0} + \alpha_1 \ket{1}$ and any $\TCF$ family, the distribution over the outcomes of the following two experiments are statistically indistinguishable under \Cref{construction:single}:
    \begin{itemize}
        \item \textbf{Experiment 1.} Measure $\ket{\psi}$ in the standard basis and report the outcome.
        \item \textbf{Experiment 2.} Execute $\Real(1^\lambda, \ket{\psi}, b_1 = 0)$, as described in \Cref{def:SCQ-correctness}.
    \end{itemize}
\end{lemma}

\begin{proof}
By inspection, it can be seen that the distribution of outcomes obtained from $\Real(1^\lambda, \ket{\psi}, b_1 = 1)$ here is the same as the outcome obtained from the following procedure:
 \begin{enumerate}
     \item Generate keys $(\sk_0, \pk_0)$.
     \item Apply the weak commitment once to get $\vecy_0, \ket{\phi_1} \leftarrow \Commit_W(\pk_0, \ket{\psi})$.
     \item Apply a Hadamard transform to the state to get
     $\ket{\phi_2} = H^{\otimes (n+1)} \ket{\phi_1}$.
     \item Execute $\Real_W(1^\lambda, \ket{\phi_2}, 1^{n+1})$ to obtain an outcome $\vecm = (m_0, \dots, m_{n})$.
     \item If $\Check_\TCF(\pk_0, m_0, (m_1, \dots, m_{n}), \vecy_0) =1$, output $m_0$; else, output $\bot$.
 \end{enumerate}
By \Cref{lem:mahadev-correctness}, applied in the Hadamard basis case, the outcome $\vecm$ has a distribution that is statistically close to the outcome of a Hadamard basis measurement of $\ket{\phi_2}$. By construction, this is equal to the distribution of the outcome of a standard basis measurement of $\ket{\phi_1}$. Finally, by \Cref{lem:mahadev-correctness}, applied in the standard basis case, the distribution of a standard outcome of $\ket{\phi_1}$ will pass the check $\Check_\TCF(\pk_0, m_0, (m_1, \dots, m_{n}), \vecy_0)$ with probability negligibly close to $1$, and the bit $m_0$ will be distributed close to the distribution obtained by measuring $\ket{\psi}$ in the standard basis.

\end{proof}
\begin{proof}[Proof of \Cref{def:SCQ-correctness-theorem-single}]
The theorem follows immediately from \Cref{def:SCQ-correctness-hadamard} and \Cref{def:SCQ-correctness-standard}.
\end{proof}
We now proceed with the proof of correctness for the succinct commitment scheme.

\begin{theorem}\label{def:SCQ-correctness-theorem-succinct}
   The succinct multi-qubit commitment scheme described in \Cref{def:SCQ-succinct-construction-multi} satisfies the correctness property given in \Cref{def:SCQ-succinct-correctness}.  
\end{theorem}

\begin{proof}
    By correctness of the state-preserving succinct argument of knowledge protocol \cite{LombardiMS22}, $\Ver.\Commit$ in \Cref{def:SCQ-succinct-construction-multi} accepts with probability $1-\negl(\lambda)$. The $\Commit$ algorithm in that construction uses $\Commit_\ss$ as a black box, which consists of applying the $\Commit_1$ procedure to each qubit of $\bsigma$ under the same public key $\pk_1$. Therefore, for the $c=0$ case in \Cref{def:SCQ-succinct-correctness}, correctness holds by \Cref{def:SCQ-correctness-theorem-single}. The $c = 1$ case follows from the correctness of the state-preserving argument of knowledge \cite{LombardiMS22}.
\end{proof}

\subsection{Binding} \label{sec:binding:single}
In this section we prove the following two theorems.

\begin{theorem}\label{thm:main:sound-non-succinct}
       The non-succinct commitment scheme described in \Cref{def:SCQ-construction-multi} satisfies the binding property given in \Cref{def:binding} (assuming the existence of a $\TCF$ family). 
\end{theorem}

\begin{theorem}\label{thm:main:sound-semi-succinct}
The semi-succinct multi-qubit commitment scheme described in \Cref{def:SCQ-construction-multi} (\Cref{remark:semi-succinct}) satisfies the binding property given in \Cref{def:binding} if the underlying single qubit commitment scheme is the one from \Cref{sec:construction:single} and the underlying $\TCF$ family is the one from \cite{BCMVV18} and assuming it satisfies the {\em distributional strong adaptive hardcore bit property} (see \Cref{def:dist-stat-HCB}).\footnote{We recall that \cite{BCMVV18} satisfies the distributional strong adaptive hardcore bit property under $\LWE$ (see \Cref{claim:Bra-is-stat}). }  
\end{theorem}

To prove the above two theorems we need to prove that both the non-succinct and the semi-succinct commitment schemes described in \Cref{def:SCQ-construction-multi} satisfy \Cref{eqn:binding1} and \Cref{eqn:binding2} of the binding property (\Cref{def:binding}).

\begin{remark}
We note that \Cref{eqn:binding2} only relies on the fact that the underlying $\TCF$ family is collapsing (as defined in \Cref{claim:collaps-binding}), whereas \Cref{eqn:binding1} relies on the adaptive hardcore bit property for the non-succinct scheme and on the specific properties of the $\TCF$ family from \cite{BCMVV18} (specifically, the distributional strong adaptive hardcore bit property) for the semi-succinct scheme.
\end{remark}

We start by proving that both the non-succinct and the semi-succinct commitment schemes satisfy \Cref{eqn:binding2}. 
We actually prove a stronger version of \Cref{eqn:binding2}, stated below.

\begin{lemma}\label{lemma:soundness-single1}[Stronger version of \Cref{eqn:binding2}]
For any $\BQP$ algorithm $\cC^*.\Commit$ and quantum state $\bsigma$, any purification $\ket{\varphi}$ of $\bsigma$, any $\BQP$ algorithms $\cC^*_1.\Open$ and $\cC^*_2.\Open$, any $\vecb\in\{0,1\}^\ell$, and any efficient unitaries $V_1$ and $V_2$ there exists a negligible function $\mu=\mu(\secp)$ such that 
    \begin{align*}
&\E_{\substack{(\pk, \sk) \leftarrow \Gen(1^\lambda)\\
(\vecy, \ket{\psi}) \leftarrow \cC^*.\Commit(\pk, \ket{\varphi})}}
\| U_1^\dagger U_\Out^{\dagger} \CNOT_{\mathsf{copy},\out} U_\Out U_1 V_{\mathsf{ext},1}\ket{\psi_\mathsf{ext}}
 - U_{2}^\dagger U_\Out^{\dagger} \CNOT_{\mathsf{copy},\out} U_\Out U_{2} V_{\mathsf{ext},2}\ket{\psi_\mathsf{ext}}\|_2\leq\\
 &~~~~~~~~~~~~~~~~~~~~~~~~~~~~~{\eta+\epsilon+\mu}
\end{align*}
where 
\begin{itemize}
\item $\epsilon= \E\limits_{\substack{(\pk, \sk) \leftarrow \Gen(1^\lambda)\\
(\vecy, \ket{\psi}) \leftarrow \cC^*.\Commit(\pk, \ket{\varphi})}} \|V_1\ket{\psi}-V_2\ket{\psi}\|_2$.
\item $\eta=\sum_{j=1}^{\ell}2\sqrt{\delta_j}$ for
   \begin{equation*}
   \delta_j\triangleq \E_{\substack{(\pk, \sk) \leftarrow \Gen(1^\lambda)\\
(\vecy, \ket{\psi}) \leftarrow \cC^*.\Commit(\pk, \ket{\varphi})\\(\vecz_i,\brho'_i)\gets \cC^*_i.\Open(\brho,\vecb)}}\max_{i\in\{1,2\}}\Pr[\Ver(\sk, \vecy, (j,b_j), \vecz_{i,j})=0~|~\Ver(\sk, \vecy, (k,b_k), \vecz_{i,k})=1~\forall k\in[j-1]]
   \end{equation*}
where $\vecz_i=(\vecz_{i,j})_{j=1}^\ell$.
\item $\ket{\psi_\mathsf{ext}}=\ket{0^\ell}_\mathsf{copy}\tensor\ket{0^\ell}_\mathsf{out}\tensor\ket{\vecb}_\mathsf{basis}\tensor\ket{\psi}$.
\item For every $i\in\{1,2\}$, $V_{\mathsf{ext},i}=I_\mathsf{copy}\tensor I_\mathsf{out}\tensor I_\mathsf{basis}\tensor V_i$.
\item For every $i\in\{1,2\}$, $U_i$ is the unitary  defined by applying $\cC^*_i.\Open$ to the registers $\mathsf{open}$ and $\mathsf{basis}$.
\item $U_\Out$ is the unitary defined by first applying the unitary corresponding to $\Ver(\sk,\vecy,\cdot,\cdot)$ to registers $\mathsf{open}$ and $\mathsf{basis}$, and controlled on $\Ver$ accepting, applying the unitary corresponding to $\Out(\sk,\vecy,\cdot,\cdot)$ to registers $\mathsf{open}$ and $\mathsf{basis}$, and writing the output on the register $\out$.
    \item $\CNOT_{\mathsf{copy},\out}$ applies a $\CNOT$ to registers $\mathsf{copy}$ and $\out$ (i.e., it copies register $\out$ to register $\mathsf{copy}$).  
\end{itemize}
Moreover, $\cC^*_1.\Open$ and $\cC^*_2.\Open$ can be $\BQP$ given $\sk_1,\ldots,\sk_{n+1}$ when opening in the standard basis and $\QPT$ when opening in the Hadamard basis.\footnote{This generalization is needed to obtain \Cref{corr:binding-real}.} Alternatively, they can be $\BQP$ given $\sk_0$ when opening in the Hadamard basis and  $\QPT$ when opening in the standard basis.
\end{lemma}

\begin{corollary}\label{corr:binding-real}
For any $\BQP$ algorithm $\cC^*.\Commit$ and quantum state $\bsigma$, any $\BQP$ algorithms $\cC^*_1.\Open$ and $\cC^*_2.\Open$, and any $\vecb\in\{0,1\}^\ell$,
    \[
     \Real^{\cC^{*}_1.\Open}(\secp, \vecb,\bsigma)
 \stackrel{2(\sqrt{\delta_0}+\sqrt{\delta_1})}\approx\Real^{\cC^{*}_2.\Open}(\secp, \vecb,\bsigma)
    \]
    where denoting by $I_b=\{i\in[\ell]:\vecb_i=b\}$,
    \[
    \delta_b=\E_{\substack{(\pk, \sk) \leftarrow \Gen(1^\lambda)\\
    (\vecy, \brho) \leftarrow \cC^*.\Commit(\pk, \bsigma)\\ \vecz_{i}\gets\cC^*_i.\Open(\brho,\vecb)}}\max_{i\in\{1,2\}}\Pr[\Ver\left(\sk,\vecy,(I_b,b^{|I_b|}),\vecz_{i,I_b}\right)=0]
    \]
    where $\vecz_{i,I_b}=(\vecz_{i,j})_{j\in I_b}$.

    Moreover, $\cC^*_1.\Open$ and $\cC^*_2.\Open$ can be $\BQP$ given $\sk_1,\ldots,\sk_{n+1}$ when opening in the standard basis and $\QPT$ when opening in the Hadamard basis.
\end{corollary}

\paragraph{Proof of \Cref{corr:binding-real}}
Fix any $\BQP$ algorithm $\cC^*.\Commit$ and quantum state $\bsigma$, any algorithms $\cC^*_1.\Open$ and $\cC^*_2.\Open$ as in the statement of \Cref{corr:binding-real}, and any basis $\vecb$.  For every $i\in\{1,2\}$ we slightly change $\cC^*_i.\Open$ to $\cC^{**}_i.\Open$, as follows: 
$\cC^{**}_i.\Open(\brho,(j,b))$ coherently computes $\vecz\gets \cC^*_i.\Open(\brho,\vecb)$ and outputs $\vecz_j$ if $\Ver(\sk,\vecy,(I_b,\vecb_{I_b}),\vecz)=1$, and otherwise it outputs $\bot$.\footnote{Note that $\vecb_{I_b}=b^{|I_b|}$}
Note that $\cC^{**}_i.\Open$ remains a $\BQP$ algorithm when opening in the Hadamard basis since $\Ver$ does not use $\sk$ when verifying a Hadamard basis opening, whereas it uses $\sk_1,\ldots,\sk_{n+1}$ when opening in the standard basis. Thus $\cC^{**}_1.\Open$ and $\cC^{**}_2.\Open$ satisfy the efficiency conditions of \Cref{lemma:soundness-single1}.  In addition, note that 
for every $i\in\{1,2\}$,
\[
 \Real^{\cC^{*}_i.\Open}(\secp,\vecb,\bsigma)\equiv  \Real^{\cC^{**}_i.\Open}(\secp,\vecb,\bsigma).
\]
where $(\pk,\sk)\gets\Gen(1^\secp)$ and $(\vecy,\brho)\gets \cC^*.\Commit(\pk,\bsigma)$.
By \Cref{lemma:soundness-single1} for any purification $\ket{\varphi}$ of $\bsigma$ there exists a negligible function~$\mu$ such that 
    \begin{align*}
&\E_{\substack{(\pk, \sk) \leftarrow \Gen(1^\lambda)\\
(\vecy, \ket{\psi}) \leftarrow \cC^*.\Commit(\pk, \ket{\varphi})}}\| U_1^\dagger U_\Out^{\dagger} \CNOT_{\mathsf{copy},\out} U_\Out U_1 \ket{\psi_\mathsf{ext}}
 - U_{2}^\dagger U_\Out^{\dagger} \CNOT_{\mathsf{copy},\out} U_\Out U_{2} \ket{\psi_\mathsf{ext}}\|_2\\
 &\leq ~~~~~~~~~~~~~~~~~~~~~~~~~~{\eta+\mu}
\end{align*}
where $U_i$ is the unitary defined by $\cC^{**}_i.\Open$, and $U_\Out$ and $\eta$ are as defined in \Cref{lemma:soundness-single1}.
It remains to observe that $\eta\leq 2\sqrt{\delta_0}+2\sqrt{\delta_1}$, which follows from the the definition of $\cC^{**}_i.\Open$, which asserts that $\delta_j=0$ if there exists $k\in\{1,\ldots,j-1\}$ for which $b_j=b_k$.
\qed

\paragraph{Proof of \Cref{lemma:soundness-single1}}  We prove this lemma for the semi-succinct variant of the multi-qubit commitment scheme described in \Cref{def:SCQ-construction-multi}.
The proof for the non-succinct variant is identical.
The proof is by induction on~$\ell$.

\paragraph{Base case: $\ell=1$.}
Fix  any $\BQP$ algorithm $\cC^*.\Commit$, a quantum state $\bsigma$,  algorithms $\cC^*_1.\Open$ and $\cC^*_2.\Open$, basis $b\in\{0,1\}$, and efficient unitaries $V_1$ and $V_2$, as in the lemma statement. Also fix a purification $\ket{\varphi}$ of $\bsigma$.
Suppose for the sake of contradiction that there exists a non-negligible $\xi=\xi(\secp)$ such that 
 \begin{align*}
&\E_{\substack{(\pk, \sk) \leftarrow \Gen(1^\lambda)\\
(\vecy, \ket{\psi}) \leftarrow \cC^*.\Commit(\pk, \ket{\varphi})}}
\| U_1^\dagger U_\Out^{\dagger} \CNOT_{\mathsf{copy},\out} U_\Out U_1 V_{\mathsf{ext},1}\ket{\psi_\mathsf{ext}}
 - U_{2}^\dagger U_\Out^{\dagger} \CNOT_{\mathsf{copy},\out} U_\Out U_{2} V_{\mathsf{ext},2}\ket{\psi_\mathsf{ext}}\|_2\geq\\
 &~~~~~~~~~~~~~~~~~~~~~~~~~~~~~{\eta+\epsilon+\xi}
\end{align*} 

We construct a  
 $\BQP$ adversary~$\A$ that uses the $\BQP$  committer~$\cC^*.\Commit$, its purified state $\ket{\varphi}$, and the unitaries $U_1,U_2,V_1,V_2,U_\Out$ to break the collapsing property of the underlying $\TCF$ family (\Cref{def:collapsing}). We break the collapsing property as formulated in \Cref{remark:collapsing}.  We distinguish between the case that $b=0$ and the case that $b=1$.

\paragraph{Case 1: $b=0$.}
The adversary $\A$ operates as follows:
\begin{enumerate}
    \item \textbf{Advserary:} Upon receiving a public key $\pk_0$ from the challenger, where $(\pk_0,\sk_0)\leftarrow \Gen_\TCF(1^\secp)$:
    \begin{enumerate}
        \item For every $i\in[n+1]$ generate $(\pk_i,\sk_i)\leftarrow\Gen_\TCF(1^\secp)$.        
        \item Let $\pk=(\pk_0,\pk_1,\ldots,\pk_{n+1})$.
        \item Compute $(\vecy,\ket{\psi})\leftarrow \cC^*.\Commit(\pk,\ket{\varphi})$. 
        \item Parse $\vecy=(\vecy_0,\vecy_1,\ldots,\vecy_{n+1})$ 
%

\item Let $\ket{\psi'}=U\left(\ket{+}_\mathsf{coin}\tensor \ket{\psi_\mathsf{ext}}\right)$, where 
        $$U=\ket{0}\bra{0}_{\mathsf{coin}} \otimes U_\Out U_{1}V_{\mathsf{ext},1}+ \ket{1}\bra{1}_{\mathsf{coin}} \otimes  U_\Out U_{2}V_{\mathsf{ext},2}$$ 
        Recall that $U_\Out$ first computes $\Ver$ which in \Cref{item:comp:m} computes $\vecm\in\{0,1\}^{n+1}$. $U_\Out$ stores in register $\out$ the output, which is the first bit of $\vecm$. We denote by $\mathsf{preimage}$ the registers that store the last $n$ bits of $\vecm$.

\item 
Send to the challenger the string $\vecy_0$ and the registers $\out$ and $\mathsf{preimage}$  of $\ket{\psi'}$.  
 
Notice that since $b=0$, $U_\Out$ (as possibly $U_1$ and $U_2$) use only the secret keys $(\sk_1,\ldots,\sk_{n+1})$, which $\A$ knows, and thus $\A$ can efficiently apply the unitary~$U$ to the state $\ket{+}_\mathsf{coin}\tensor \ket{\psi_{\mathsf{ext}}}$. 
        
\end{enumerate}
    \item \textbf{Challenger:} Recall that 
   the challenger applies in superposition the algorithm $\Check_\TCF$ to the state it receives w.r.t.\ public key $\pk_0$ and the image string $\vecy_0$, and measures the bit indicating whether the output of $\Check_\TCF$ is~$1$. If this is not the case it sends $\bot$.  Otherwise, it chooses a random bit $u\gets\{0,1\}$ and measures this state if and only if $u=1$.  It then sends the resulting state to the adversary.  

   Note that by the two-to-one nature of the underlying $\TCF$ family, measuring the entire state is equivalent to measuring only the first qubit of the state, i.e., register $\mathsf{out}$.  Thus, we can assume that the challenger measures only register $\mathsf{out}$ if and only if $u=1$.

In addition, note that conditioned on the challenger not outputting $\bot$, the state is projected to  $\Pi_{\Ver}\ket{\psi'}$ (up to normalization), where $\Pi_{\Ver}\ket{\psi'}$ is the state $\ket{\psi'}$ projected to the challenger accepting the state.  Consider the state $\CNOT^u_{\mathsf{copy},\mathsf{out}}\Pi_{\Ver}\ket{\psi'}$. Note that this state, with the $\mathsf{copy}$ register excluded, is indistinguishable from the state returned from the challenger conditioned on choosing the random bit~$u$.  Thus we think of the adversary as receiving this state.
    
    \item 
    \textbf{Adversary:}
    If the adversary receives $\bot$ from the challenger, then it outputs a uniformly random $u'$.
    Note that this occurs with probability at most $\delta$.
    
    Otherwise, the adversary $\A$ receives the registers $\out$ and $\mathsf{preimage}$ from the challenger (either measured or not, depending on $u$). The joint state of the adversary and challenger at this point is  $\CNOT^u_{\mathsf{copy},\mathsf{out}}\Pi_{\Ver}\ket{\psi'}$, where all registers except $\mathsf{copy}$ are held by the adversary.  The adversary does the following:
    \begin{enumerate}
\item Let 
\[
U' =\ket{0}\bra{0}_{\mathsf{Coin}} \otimes U_{1}^\dagger U_\out^\dagger + \ket{1}\bra{1}_{\mathsf{Coin}} \otimes U_{2}^\dagger U_\out^\dagger .
\]

\item Apply $U'$ to the adversary's system, resulting in the joint state 
\begin{align*}
&U' \CNOT^u_{\mathsf{copy},\mathsf{out}}\Pi_{\Ver}\ket{\psi'}=\\
&U'\CNOT^u_{\mathsf{copy},\mathsf{out}}\Pi_{\Ver}U\ket{\psi_{\mathsf{ext}}} =\\
&\ket{0}\bra{0}_{\mathsf{Coin}}\tensor U^\dagger_{1}  U_\out^\dagger \CNOT^u_{\mathsf{copy},\mathsf{out}}\Pi_\Ver U_\out U_1 V_{\mathsf{ext},1}(\ket{+}_{\mathsf{Coin}}\tensor\ket{\psi_{\mathsf{ext}}}) + \\
&\ket{1}\bra{1}_{\mathsf{Coin}}\tensor 
 U^\dagger_{2}  U_\out^\dagger \CNOT^u_{\mathsf{copy},\mathsf{out}}\Pi_\Ver U_\out U_2 V_{\mathsf{ext},2}(\ket{+}_{\mathsf{Coin}}\tensor\ket{\psi_{\mathsf{ext}}})
\end{align*}

\item Output the measurement of the $\mathsf{Coin}$ register in the Hadamard basis, denoted by~$u'$ (i.e., $u'=0$ if the measurement is $\ket{+}$  and is $u'=1$ if the measurement is $\ket{-}$).
    \end{enumerate}
\end{enumerate}
Consider the states 
\[
U^\dagger_{1}  U_\out^\dagger \CNOT^u_{\mathsf{copy},\mathsf{out}}\Pi_\Ver U_\out U_1 V_{\mathsf{ext},1}\ket{\psi_{\mathsf{ext}}}~~\mbox{ and }~~  
 U^\dagger_{2}  U_\out^\dagger \CNOT^u_{\mathsf{copy},\mathsf{out}}\Pi_\Ver U_\out U_2 V_{\mathsf{ext},2}\ket{\psi_{\mathsf{ext}}}
\]
Note that for $u=0$, these states are $(2\sqrt{\delta}+\epsilon)$-close in $\|\cdot \|_2$ distance.  This follows from the fact that by \Cref{lem:gentle-infant}, together with the assumption that the probability that $\ket{\psi_{\mathsf{ext}}}$ opens successfully is $\geq 1 - \delta$, it holds that for every $i\in\{1,2\}$:
\begin{align*}
    \E_{\substack{(\pk, \sk) \leftarrow \Gen(1^\lambda)\\
(\vecy, \ket{\psi}) \leftarrow \cC^*.\Commit(\pk, \ket{\varphi})}}\| U^\dagger_{i}  U_\out^\dagger \Pi_\Ver U_\out U_i V_{\mathsf{ext},i}\ket{\psi_{\mathsf{ext}}}
 - U^\dagger_{i}  U_\out^\dagger  U_\out U_i V_{\mathsf{ext},i}\ket{\psi_{\mathsf{ext}}}\|_2 \leq \sqrt{\delta},
\end{align*}
and from our assumption that 
\begin{align*}
   \epsilon= \E_{\substack{(\pk, \sk) \leftarrow \Gen(1^\lambda)\\
(\vecy, \ket{\psi}) \leftarrow \cC^*.\Commit(\pk, \ket{\varphi})}}\| V_{\mathsf{ext},1}\ket{\psi_{\mathsf{ext}}}
 - V_{\mathsf{ext},2}\ket{\psi_{\mathsf{ext}}}\|_2.
\end{align*}
This implies that there exists a negligible function~$\mu$ such that 

 \begin{align*}
&\E_{\substack{(\pk, \sk) \leftarrow \Gen(1^\lambda)\\
(\vecy, \ket{\psi}) \leftarrow \cC^*.\Commit(\pk, \ket{\varphi})}}
\| U_1^\dagger U_\Out^{\dagger} \CNOT_{\mathsf{copy},\out} U_\Out U_1 V_{\mathsf{ext},1}\ket{\psi_\mathsf{ext}}
 - U_{2}^\dagger U_\Out^{\dagger} \CNOT_{\mathsf{copy},\out} U_\Out U_{2} V_{\mathsf{ext},2}\ket{\psi_\mathsf{ext}}\|_2\leq\\
 &~~~~~~~~~~~~~~~~~~~~~~~~~~~~~{2\sqrt{\delta}+\epsilon+\mu}
\end{align*}

On the other hand, by our contradiction assumption, for $u=1$, these two states are $(2\sqrt{\delta}+\epsilon)$-far.  This, together with \Cref{claim:aux} below, implies that $\A$ indeed breaks the collapsing property of the underlying $\TCF$ family.
\begin{proposition}\label{claim:aux}
    For any two states $\ket{\psi_0}$ and $\ket{\psi_1}$ such that $\left\| \ket{\psi_0} - \ket{\psi_1} \right\|=\epsilon$, and for $\ket{\varphi}=\frac{1}{\sqrt{2}}\ket{0}\ket{\psi_0}+\frac{1}{\sqrt{2}}\ket{1}\ket{\psi_1}$, it holds that
    $$\Pr[H[\varphi]\rightarrow 1] = \frac{\epsilon^2}{4}.$$
\end{proposition}
\begin{proof}
We calculate
\begin{align*}
    \Pr[H[\varphi] \mapsto 1] &= \norm{(\bra{1} \otimes I) H \ket{\varphi}}^2 \\
    &= \norm{(\bra{1} \otimes I) \left(\frac{1}{\sqrt{2}} \ket{+} \ket{\psi_0} + \frac{1}{\sqrt{2}} \ket{-} \ket{\psi_1} \right)}^2 \\
    &= \left\| \frac{1}{2} \ket{\psi_0} - \frac{1}{2} \ket{\psi_1} \right\|^2 \\
    &= \frac{1}{4} \epsilon^2. \qedhere
\end{align*}
\end{proof}

\paragraph{Case 2:  $b=1$.}  We show how to use the adversary $\A$ to break the extended collapsing game (see \Cref{claim:ext-collapsing}). 
The adversary $\A$ operates as follows:
\begin{enumerate}
    \item Upon receiving public keys $(\pk_1,\ldots,\pk_{n+1})$ from the challenger, where $(\pk_i,\sk_i)\leftarrow \Gen_\TCF(1^\secp)$ for every $i\in[n+1]$, do the following:
    \begin{enumerate}
        \item Generate $(\pk_0,\sk_0)\leftarrow\Gen_\TCF(1^\secp)$.
        \item Let $\pk=(\pk_0,\pk_1,\ldots,\pk_{n+1})$.
       \item Compute $(\vecy,\ket{\psi})\leftarrow \cC^*.\Commit(\pk,\ket{\varphi})$.
      \item Parse $\vecy=(\vecy_{0},\vecy_{1},\ldots,\vecy_{n+1})$.
        \item Compute $((0,\vecx_0),(1,\vecx_1))=\Invert_{\TCF}(\sk_0,\vecy_{0})$.
        \item Let $J= \{j\in\{2,\ldots,n+1\}: x_{0,j-1}\oplus x_{1,j-1}=1\}\cup \{1\}$.

\item As in the $b=0$ case, define
          $$U=\ket{0}\bra{0}_{\mathsf{coin}} \otimes U_\Out U_{1}V_{\mathsf{ext},1}+ \ket{1}\bra{1}_{\mathsf{coin}} \otimes  U_\Out U_{2}V_{\mathsf{ext},2}$$ 
        and prepare the state $\ket{\psi'}=U(\ket{+}_\mathsf{coin}\tensor\ket{\psi_\mathsf{ext}})$.

        Note that since $b=1$ it holds that $\ket{\psi'}$ can be computed efficiently given $\sk_0$

         \item For every $j\in [J]$, denote by $\cX_j$ and $\cZ_j$ the registers in $\brho'$  corresponding to $d_j$ and $x'_j$, respectively.
     
         \item Send  $J$, $\{\vecy_{j}\}_{j\in J}$ and the  registers $\{\cX_j,\cZ_j\}_{j\in J}$ of $\ket{\psi'}$.  

         \end{enumerate}

          \item Recall that the challenger applies in superposition the algorithm $\Check$ to the state it received w.r.t.\ the image strings $\{\vecy_j\}_{j\in J}$, where the $j$'th check is w.r.t~$\pk_j$, and measures the bit indicating whether the output of $\Check$ is~$1$. If any of the outputs of $\Check$ are $0$, the challenger immediately halts and sends $\bot$ to the adversary. Otherwise, it chooses a random bit $u\gets\{0,1\}$ and applies $Z^u$ to every $\cX_j$ register.  It then sends the resulting state to the adversary.

          \item 
          If the adversary receives $\bot$, it returns a uniformly random $u'$.
          Otherwise, observe that once the adversary receives the state from the challenger, it is in possession of all the quantum registers.  At this point, they are, up to normalization, in the state $Z^u_J \Pi_{\Ver,J}\ket{\psi'}$, where $Z_J=\prod_{j\in J}Z_{\cX_{j}}$ and $\Pi_{\Ver,J}\ket{\psi'}$ is the state $\ket{\psi'}$ projected to an accepting state. 
    
      It then does the following:
\begin{enumerate}
\item 
Let 
\[
U' =\ket{0}\bra{0}_{\mathsf{Coin}} \otimes U_{1}^\dagger U_\out^\dagger + \ket{1}\bra{1}_{\mathsf{Coin}} \otimes U_{2}^\dagger U_\out^\dagger .
\]

 \item Apply $U'$ to its registers, resulting in the state
\begin{align*}
& U'Z^u_J \Pi_{\Ver,J}\ket{\psi'}=\\
&\left(\ket{0}\bra{0}_\coin\tensor  
U^\dagger_{1}  U_\out^\dagger Z^u_{J}\Pi_\Ver U_\out U_{1}V_\mathsf{ext,1} + \ket{1}\bra{1}_\coin\tensor 
U^\dagger_{2}  U_\out^\dagger Z^u_{J}\Pi_\Ver U_\out U_{2}V_\mathsf{ext,2} \right)(\ket{+}_\coin\tensor\ket{\psi_{\mathsf{ext}}})
\end{align*}

\item Output the measurement of  the first register of this state in the Hadamard basis, denoted by~$u'$.
\end{enumerate}

        \end{enumerate}

        Consider the states 
\[
U^\dagger_{1}  U_\out^\dagger Z^u_J\Pi_\Ver U_\out U_1 V_{\mathsf{ext},1}\ket{\psi_{\mathsf{ext}}}~~\mbox{ and }~~  
 U^\dagger_{2}  U_\out^\dagger Z^u_J\Pi_\Ver U_\out U_2 V_{\mathsf{ext},2}\ket{\psi_{\mathsf{ext}}}
\]
Note that similarly to the $b=0$ case, for $u=0$ these states are $(2\sqrt{\delta}+\epsilon)$-close in $\|\cdot\|_2$ distance.  
On the other hand, by our contradiction assumption, together with \Cref{lem:control-distance-z}, for $u=1$, these two states are $(2\sqrt{\delta}+\epsilon)$-far in $\|\cdot\|_2$ distance. This together with \Cref{claim:aux}, implies that indeed $\A$ breaks the collapsing property of the underlying $\TCF$ family.

\paragraph{Induction step:}      Suppose that the multi-qubit commitment scheme is sound for $\ell-1$ and we prove that it is sound for $\ell$. 
We need to prove that there exists a negligible function $\mu=\mu(\secp)$ such that 
    \begin{align*}
&\E_{\substack{(\pk, \sk) \leftarrow \Gen(1^\lambda)\\
(\vecy, \ket{\psi}) \leftarrow \cC^*.\Commit(\pk, \ket{\varphi})}}
\| U_1^\dagger U_\Out^{\dagger} \CNOT_{\mathsf{copy},\out} U_\Out U_1 V_{\mathsf{ext},1}\ket{\psi_\mathsf{ext}}
 - U_{2}^\dagger U_\Out^{\dagger} \CNOT_{\mathsf{copy},\out} U_\Out U_{2} V_{\mathsf{ext},2}\ket{\psi_\mathsf{ext}}\|_2\leq\\
 &~~~~~~~~~~~~~~~~~~~~~~~~~~~~~{\eta+\epsilon+\mu}
\end{align*}
    for $\eta=\sum_{j=1}^\ell 2\sqrt{\delta_j}$ and $\epsilon= \E\limits_{\substack{(\pk, \sk) \leftarrow \Gen(1^\lambda)\\
(\vecy, \ket{\psi}) \leftarrow \cC^*.\Commit(\pk, \ket{\varphi})}} \|V_1\ket{\psi}-V_2\ket{\psi}\|_2$.\\

 To this end, note that for every $i\in\{1,2\}$
    \begin{align*}
 &U_i^\dagger U_\Out^{\dagger} \CNOT_{\mathsf{copy},\out} U_\Out U_i V_{\mathsf{ext},i}\ket{\psi_\mathsf{ext}}=\\
&U_i^\dagger U_{\Out}^{\dagger} \CNOT_{\mathsf{copy}_{\ell},\out_{\ell}}\CNOT_{\mathsf{copy}_{[1,\ell-1]},\out_{[1,\ell-1]}} U_{\Out} U_i V_{\mathsf{ext},i}\ket{\psi_\mathsf{ext}}=\\
&U_i^\dagger U_{\Out_\ell}^{\dagger} \CNOT_{\mathsf{copy}_{\ell},\out_{\ell}}U_{\Out_\ell} U_i \underbrace{U_i^\dagger U_{\Out_{[1,\ell-1]}}^{\dagger} \CNOT_{\mathsf{copy}_{[1,\ell-1]},\out_{[1,\ell-1]}} U_{\Out_{[1,\ell-1]}} U_i V_{\mathsf{ext},i}}_{V_i'} \ket{\psi_\mathsf{ext}} 
\end{align*}
For every $i\in\{1,2\}$, denote by 
\[
\ket{\psi'_i}= V_i' \ket{\psi_\mathsf{ext}}
\]
By the induction hypothesis, there exists a negligible function $\mu=\mu(\secp)$ such that 
\begin{equation*}
\E_{\substack{(\pk, \sk) \leftarrow \Gen(1^\lambda)\\
(\vecy, \ket{\psi}) \leftarrow \cC^*.\Commit(\pk, \ket{\varphi})}}\|\ket{\psi'_1}-\ket{\psi'_2}\|_2\leq \eta'+\mu
\end{equation*}
where $\eta'=\sum_{j=1}^{\ell-1} 2\sqrt{\delta_{j}}+\epsilon$.
Denoting by $\epsilon'=\eta'$, our base case implies that there exists a negligible function $\nu=\nu(\secp)$ such that 
\begin{align*}
&\E_{\substack{(\pk, \sk) \leftarrow \Gen(1^\lambda)\\
(\vecy, \ket{\psi}) \leftarrow \cC^*.\Commit(\pk, \ket{\varphi})}}\|U_1^\dagger U_{\Out_\ell}^{\dagger} \CNOT_{\mathsf{copy}_{\ell},\out_{\ell}}U_{\out_\ell} U_1V'_1\ket{\psi_\mathsf{ext}}- U_1^\dagger U_{\Out_\ell}^{\dagger} \CNOT_{\mathsf{copy}_{\ell},\out_{\ell}}U_{\out_\ell} U_1 V_2'\ket{\psi_\mathsf{ext}}\|_2\leq\\
&~~~~~~~~~~~~~~~~~~~~~~~~~ 2\sqrt{\delta_\ell}+\eta'+\nu
\end{align*}
as desired.\\

\qed

We next prove that both the non-succinct and the semi-succinct commitment schemes from \Cref{def:SCQ-construction-multi} satisfy \Cref{eqn:binding1}.

\begin{lemma}\label{lemma:non-succinct-real-to-ideal}
  The non-succinct commitment scheme described in \Cref{def:SCQ-construction-multi} satisfies \Cref{eqn:binding1}  from \Cref{def:binding} assuming the underlying $\TCF$ family has the adaptive hardcore bit property.   
\end{lemma}

\begin{lemma}\label{lemma:semi-succinct-real-to-ideal}
  The semi-succinct commitment scheme described in \Cref{def:SCQ-construction-multi}  (\Cref{remark:semi-succinct}) satisfies \Cref{eqn:binding1} from \Cref{def:binding} assuming the underlying $\TCF$ family is the one from \cite{BCMVV18} and assuming it has the distributional strong adaptive hardcore bit property (which is the case under $\LWE$).   
\end{lemma}

\paragraph{Proof of \Cref{lemma:non-succinct-real-to-ideal,lemma:semi-succinct-real-to-ideal}}
We prove these two lemmas jointly since much of the proof is identical.  In both cases we think of the public and secret keys as being \[\pk=(\pk_1,\ldots,\pk_\ell)~~\mbox{ and }~~\sk=(\sk_1,\ldots,\sk_\ell)
\]
where in the non-succinct commitment each $(\pk_i,\sk_i)\gets\Gen_\TCF(1^\secp)$ and in the semi-succinct commitment $(\pk,\sk)\gets\Gen_\TCF(1^\secp)$, and for every $i\in[\ell]$
\[
\sk_i=\sk~~\mbox{ and }~~ \pk_i=\pk
\]
Fix any $\BQP$ cheating committer $\cC^*.\Commit$ with auxiliary quantum state~$\bsigma$ that commits to an $\ell$-qubit state. Denote by 
\[(\vecy,\brho)\gets \cC^*.\Commit(\pk,\bsigma),
\]
where $\vecy=(\vecy_1,\ldots,\vecy_\ell)$, each $\vecy_i=(\vecy_{i,0},\vecy_1,\ldots,\vecy_{i,n+1})$ and each $\vecy_{i,j}$ is in $\mathsf{R}_{\pk_j}$ which is the range of the $\TCF$ function $\Eval(\pk_j,\cdot)$. 
Fix any $\BQP$ algorithm $\cC^*.\Open$.
We start by defining the $\QPT$ extractor $\Ext^{\cC^*.\Open}(\sk, \vecy, \brho)$.  We do so 
 in two steps:
\begin{enumerate}
\item First, we define $2\ell$ ``operational observables'' $\{P_{X_i}, P_{Z_i}\}_{i\in \ell]}$  such that for every $i\in[\ell]$ and $b\in\{0,1\}$,
\[
(\pk,\vecy,\vecm_{\mathsf{ideal},i,b})\equiv(\pk,\vecy,\vecm_{i,b})
\]
where $(\pk,\sk)\gets \Gen(1^\secp,1^\ell)$,\footnote{In the semi-succinct setting $(\pk,\sk)\gets \Gen(1^\secp)$.} $(\vecy,\brho)\gets \cC^*.\Commit(\pk,\bsigma)$, $\vecm_{\mathsf{ideal},i,b}$ is obtained by measuring $\brho$ in the $P_{X_i}$ basis if $b=1$ and measuring it in the $P_{Z_i}$ basis if $b=0 $, and $\vecm_{i,b}$ is obtained by computing $\vecz\gets\cC^*.\Open(\brho,b^\ell)$  and setting $\vecm_{i,b}=\Out(\sk,\vecy,(i,b),\vecz_i)$. 
\item We then use these operational observables to extract a state $\btau$. This is done following the approach of \cite{Mah18a,Vid20-course,Bartusek22},
\end{enumerate}

\paragraph{Defining the operational observables $\{P_{X_i},P_{Z_i}\}_{i\in[\ell]}$.}
To define these operational observables formally, we add $L=\ell\cdot\left((n+1)^2+1\right)$ ancilla registers to $\brho$, which we initialize to~$0$.  We denote by
\begin{equation*}
\brho_\Ext=\brho\tensor\ket{0^{L}}\bra{0^{L}},
\end{equation*}
where the first $\ell\cdot(n+1)^2$ ancilla registers are denoted by $\mathsf{open}=(\mathsf{open}_1,\ldots,\mathsf{open}_\ell)$, and these registers store the output $(\vecz_1,\ldots,\vecz_\ell)$ generated by $\Open$, where $\vecz_i\in\{0,1\}^{(n+1)^2}$ is stored in $\mathsf{open}_i$. The last $\ell$ ancilla registers are denoted by $\out=(\out_1,\ldots,\out_\ell)$, and these registers store the output $(v_1,\ldots,v_\ell)$ generated by $\Out$, where $v_i\in\{0,1\}$ is stored in register $\out_i$. 

\begin{definition}\label{def:obs-multi}
 For any $(\sk,\vecy)$ and any $\BQP$ algorithm $\cC^*.\Open$ we define the  operational observables $(P_{X_i},P_{Z_i})_{i\in[\ell]}$ to be
\[
 P_{X_i}=U_{1}^\dagger  \Out_{i,1}^{\dagger} Z_{\out_i} \Out_{i,1} U_{1}
 \]
 and 
 \[
 P_{Z_i}=U_{0}^\dagger \Out_{i,0}^{\dagger} Z_{\out_i} \Out_{i,0} U_{0}
 \]
 where for every $i\in[\ell]$ and every $b\in\{0,1\}$,
 \begin{itemize}
 \item $U_{b}$ is the unitary corresponding to $\cC^*.\Open(\cdot,(b,\ldots,b))$. The output is recorded in registers $\mathsf{open}$.
     \item $\Out_{i,b}$ computes $\Out(\sk,\vecy,(i,b),\cdot)$ and records the output in the ancilla register $\out_i$. 
     


     \item $Z_{\out,i}$ is the Pauli $Z$ operator applied only on the  register $\out_i$.   
     \end{itemize}
\end{definition}

We next define the extractor $\Ext$ which uses the operational observables $\{P_{X_i},P_{Z_i}\}_{i\in[\ell]}$, defined above.  
For the sake of simplicity, we define $\Ext$ to operate on pure states.  The definition easily generalizes to mixed states by linearity.  
\paragraph{$\Ext^{\cC^*.\Open}(\sk, \vecy, {\ket\varphi})$ operates as follows:}
\begin{enumerate}
\item Consider the operational observables $\{P_{X_i},P_{Z_i}\}_{i\in[\ell]}$ corresponding to $(\sk,\vecy)$.  
\item  Prepare the state 
$$
\frac{1}{2^\ell}\sum_{\vecr,\vecs\in\{0 1\}^\ell}\ket{\vecr,\vecs}_\coin\tensor\ket{0^\ell}_\cA\tensor\ket\varphi_\cB.
$$
\item Denote by 
\[\vecX^\vecr=X_\ell^{r_\ell}\ldots X_1^{r_1}~\mbox{ and }\vecZ^\vecs=Z_\ell^{s_\ell}\ldots Z_1^{s_1}.
\]
Similarly, denote by
\[
P_{\vecX}^\vecr=P_{X_\ell}^{r_\ell}\ldots P_{X_1}^{r_1}~\mbox{ and }P_\vecZ^\vecs=P_{Z_\ell}^{s_\ell}\ldots P_{Z_1}^{s_1}.
\]

\item Controlled on the values $\vecr, \vecs$ of the $\coin$ register, apply $\vecZ^\vecs \vecX^\vecr $ to the $\cA$ register and apply $P_\vecX^\vecr P_\vecZ^\vecs$ to the $\cB$ register to obtain the state
$$
\frac{1}{2^\ell}
\sum_{\vecr,\vecs\in\{0,1\}^\ell}\ket{\vecr,\vecs}_\coin\tensor \vecZ^\vecs \vecX^\vecr\ket{0^\ell}_\cA\tensor P_\vecX^\vecr P_\vecZ^\vecs\ket{\varphi}_\cB
$$
\item Apply Hadamard gates $H^{\otimes 2\ell}$ to the $\coin$ register in the  to obtain the state 
$$
\frac{1}{4^\ell}
\sum_{\vecr,\vecs,\vecr',\vecs'\in\{0,1\}^\ell}(-1)^{\vecr\cdot\vecr'+\vecs\cdot\vecs'}\ket{\vecr',\vecs'}_\coin\tensor \vecZ^\vecs \vecX^r\ket{0^\ell}_\cA\tensor P_\vecX^\vecr P_\vecZ^\vecs\ket{\varphi}_\cB 
$$
where 
\[\vecr\cdot\vecr'=\sum_{i=1}^\ell r_i\cdot r'_i~\mathsf{mod}~2~~\mbox{ and }~~\vecs\cdot\vecs'=\sum_{i=1}^\ell s_i\cdot s'_i~\mathsf{mod}~2.
\]
\item Apply $\vecX^{\vecs'}\vecZ^{\vecr'}$ to the $\cA$ register.  Note that 
\begin{align*}
&\vecX^{\vecs'}\vecZ^{\vecr'} \vecZ^\vecs \vecX^\vecr\ket{0^\ell}=\\
&\vecX^{\vecs'}\vecZ^{\vecs} \vecZ^{\vecr'} \vecX^\vecr\ket{0^\ell}=\\
&(-1)^{\vecr\cdot\vecr'}\vecX^{\vecs'}\vecZ^{\vecs} \vecX^\vecr \vecZ^{\vecr'} \ket{0^\ell}=\\
&(-1)^{\vecr\cdot\vecr'}\vecX^{\vecs'}\vecZ^{\vecs} \vecX^\vecr \ket{0^\ell}=\\
&(-1)^{\vecr\cdot \vecr'+\vecs\cdot \vecs'}\vecZ^{\vecs} \vecX^{\vecs'} \vecX^\vecr \ket{0^\ell}=\\
&(-1)^{\vecr\cdot \vecr'+\vecs\cdot \vecs'}\vecZ^{\vecs} \vecX^\vecr  \vecX^{\vecs'}\ket{0^\ell}=\\
&(-1)^{\vecr\cdot \vecr'+\vecs\cdot \vecs'}\vecZ^{\vecs} \vecX^\vecr  \ket{\vecs'}.
\end{align*}
Therefore the state obtained is 
$$
\frac{1}{4^\ell}
\sum_{\vecr,\vecs,\vecr',\vecs'\in\{0,1\}^\ell}\ket{\vecr',\vecs'}_\coin\tensor \vecZ^\vecs \vecX^\vecr\ket{\vecs'}_\cA\tensor P_\vecX^\vecr P_\vecZ^\vecs\ket{\varphi}_\cB 
$$
which is equal to the state
$$
\left(\frac{1}{\sqrt{2}^\ell}(\ket{0}+\ket{1})^{\tensor\ell}\right)\tensor \frac{1}{(2\sqrt{2})^\ell}\sum_{\vecr,\vecs,\vecs'\in\{0,1\}^\ell}\ket{\vecs'}\tensor \vecZ^\vecs \vecX^\vecr\ket{\vecs'}_\cA\tensor P_\vecX^\vecr P_\vecZ^\vecs\ket{\varphi}_\cB. 
$$
\item Discard the first $\ell$ registers to obtain the state
$$
\frac{1}{({2\sqrt{2}})^\ell}\sum_{\vecr,\vecs,\vecs'\in\{0,1\}^\ell}\ket{\vecs'}\tensor \vecZ^\vecs \vecX^\vecr\ket{\vecs'}_\cA \tensor P_\vecX^\vecr P_\vecZ^\vecs\ket{\varphi}_\cB.
$$
\item Output the state $\btau_{\cA, \cB}$ that is the reduced state of the above on registers $\cA, \cB$.

\end{enumerate}

We next prove that 
\begin{equation}\label{eqn:real-ideal-3sqrt}
 \Real^{\cC^*.\Open}(\secp, \vecb,\bsigma) \stackrel{10\sqrt{\delta}}\approx \Ideal^{\Ext,\cC^*.\Open}(\secp, \vecb,\bsigma).
\end{equation}
To this end, for a given $\vecb \in\{0,1\}^\ell$, denote by
\[I=\{i\in [\ell]:\vecb_i=0\}~\mbox{ and }~J=\{j\in[\ell]:\vecb_j=1\},
\]
so that $I$ and $J$ partition $[\ell]$.

Next we define a new opening algorithm $\cC^*.\Open_{[\ell]}$.  We first give a  ``buggy'' definition of $\cC^*.\Open_{[\ell]}$, and then show how to fix it in \Cref{remark:C*ell}. $\cC^*.\Open_{[\ell]}$ on input $(\brho,\vecb)$ does the following: 
\begin{enumerate}
\item 
 Compute $\brho_1=U_{0}^\dagger \CNOT_{\mathsf{copy}_{I},\mathsf{open}_{I}} U_0[\brho]$, 
    where $\mathsf{open}_{I}$ is the register that contains the openings $\{\vecz_i\}_{i\in I}$, and $\CNOT_{\mathsf{copy}_{I},\mathsf{open}_{I}}$ copies the content of this register to a fresh register denoted by $\mathsf{copy}_{I}$.  
    \item Measure the registers $\mathsf{copy}_{I}$ of $\brho_1$ to obtain $\{\vecz_i\}_{i\in I}$ and denote the resulting state by $\brho_2$. 
    \item 
 Compute $\brho_3=U_{1}^\dagger \CNOT_{\mathsf{copy}_{J},\mathsf{open}_{J}} U_1[\brho_2]$, 
    where $\mathsf{open}_{J}$ is the register that contains the openings $\{\vecz_j\}_{j\in J}$, and $\CNOT_{\mathsf{copy}_{J},\mathsf{open}_{J}}$ copies the content of this register to a fresh register denoted by $\mathsf{copy}_{J}$.  
     \item Measure the registers $\mathsf{copy}_J$ of $\brho_3$ to obtain $\{\vecz_j\}_{j\in J}$ and denote the resulting state by $\brho_4$. 
    \item  Output $((\vecz_1,\ldots,\vecz_\ell),\brho_4)$. 
\end{enumerate} 
\begin{remark}\label{remark:C*ell}
    We remark that as written, $\cC^*.\Open_{[\ell]}(\brho,\vecb)$ may be rejected with high probability.  The reason is that, while the standard basis openings of $\cC^*.\Open_{[\ell]}$ and $\cC^*.\Open$ are identical, $\cC^*.\Open_{[\ell]}$ can completely fail to open in the Hadamard basis, since after computing the standard basis openings its state becomes $U_{0}^\dagger \CNOT_{\mathsf{copy}_{I},\mathsf{open}_{I}} U_0[\brho]$, with the $\mathsf{copy}_I$ registers measured. This is a disturbed state and it is no longer clear that computing the Hardamard basis opening on it will give an accepting opening.  
    
    To ensure that $\cC^*.\Open_{[\ell]}(\brho,\vecb)$ is accepted with the same probability as $\cC^*.\Open(\brho,\vecb)$, up to negligible factors, we need to ensure that the state after computing the standard basis openings remains undisturbed, or at least that this disturbance is undetected by the algorithms that compute the Hadamard basis opening and verify whether this opening is valid.     To achieve this we slightly modify $\cC^*.\Open_{[\ell]}$ and allow it to compute the standard basis opening using $(\sk_1,\ldots,\sk_{n+1})$.  We note that such opening algorithms are allowed in \Cref{corr:binding-real} (which we will later use in our analysis).  
    
    Specifically, $\cC^*.\Open_{[\ell]}$, rather than placing the output of $U_0$ in the $\mathsf{open}_I$ registers, which when measured may disturb the state, we use $(\sk_1,\ldots,\sk_{n+1})$ to apply the following post-opening unitary to each $\mathsf{open}_i$ register, to ensure that when measured the disturbance will not be noticed. Recall that $\mathsf{open}_i$ contains a vector $\vecz=(\vecz_1,\ldots,\vecz_{n+1})\in\{0,1\}^{(n+1)^2}$ where each $\vecz_j\in\{0,1\}^{n+1}$. The post-opening unitary does the following:
        \begin{enumerate}
           \item Coherently compute for every $j\in[n+1]$ the bit $m_j=\vecz_j\cdot (1,\vecx'_{j,0}\oplus \vecx'_{j,1})$, where $\vecx'_{j,0}$ and $\vecx'_{j,1}$ are the two preimages of $\vecy_{i,j}$ that are computed using $\sk_j$.
            \item Let $\vecm=(m_1,\ldots,m_{n+1})\in\{0,1\}^{n+1}$.
            
            Note that if $\vecz$ is a successful opening (i.e., it is accepted) then $\vecm$ is a preimage of $\vecy_{i,0}$, and whether a preimage is measured or not is undetectable without knowing $\sk_0$, due to the collapsing property of the underlying $\TCF$ family.
            \item On an ancila register, compute a super-position over all $\vecz'=(\vecz'_1,\ldots,\vecz'_{n+1})\in\{0,1\}^{(n+1)^2}$ such that for every $j\in[n+1]$ $m_j=\vecz'_j\cdot (1,\vecx'_{j,0}\oplus \vecx'_{j,1})$.  
            \item Swap register $\mathsf{open}_i$ with the ancila register above, so that now $\vecz'=(\vecz'_1,\ldots,\vecz'_{n+1})$ is in register $\mathsf{open}_i$.
        \end{enumerate}
  Now we can argue that the residual state after computing the standard basis opening seems undisturbed for anyone who does not know $\sk_0$ due to the collapsing property of the underlying $\TCF$ family, and computing the Hadamard opening and verification of it does not use $\sk_0$ (and is done publicly given only $\pk)$. 
\end{remark}

Note that since
\[
    \delta= \E_{\substack{(\pk,\sk) \leftarrow \Gen(1^\lambda, 1^\ell) \\ (\vecy, \brho) \leftarrow \cC^*.\Commit(\pk, \bsigma)}}\max_{\vecb'\in\{\vecb,{\bf 0},{\bf 1\}}}\Pr[\Ver(\sk,\vecy,\vecb',\cC^*.\Open(\brho,\vecb'))=0].
\]
it holds that \begin{equation}\label{eqn:delta[ell]}
    \E_{\substack{(\pk,\sk) \leftarrow \Gen(1^\lambda, 1^\ell) \\ (\vecy, \brho) \leftarrow \cC^*.\Commit(\pk, \bsigma)}}\max_{\vecb'\in\{\vecb,{\bf 0},{\bf 1\}}}\Pr[\Ver(\sk,\vecy,\vecb',\cC^*.\Open_{[\ell]}(\brho,\vecb'))=0]\leq 2\delta.
    \end{equation}
This is the case since the probability that $\cC^*.\Open_{[\ell]}(\brho,\vecb)$ is rejected is bounded by the sum of the probabilities that $\cC^*.\Open(\brho,0^\ell)$ is rejected and  $\cC^*.\Open(\brho,1^\ell)$ is rejected. 

By \Cref{corr:binding-real}, we conclude that for every $\vecb\in\{0,1\}^\ell$,
\begin{equation}\label{eqn:real-to-real}
\Real^{\cC^*.\Open_{[\ell]}}(\secp, \vecb,\bsigma) \stackrel{6\sqrt{\delta}}\approx \Real^{\cC^*.\Open}(\secp, \vecb,\bsigma)
\end{equation}
This implies that to prover \Cref{eqn:real-ideal-3sqrt} it suffices to prove
\begin{equation}\label{eqn:real-ind-ideal-semi}
 \Real^{\cC^*.\Open_{[\ell]}}(\secp, \vecb,\bsigma) \stackrel{4\sqrt{\delta}}\approx \Ideal^{\Ext,\cC^*.\Open}(\secp, \vecb,\bsigma)
\end{equation}

To this end, we first  compute the distribution of measurement outcomes on the extracted state. While in general the input to the extractor is a mixed state $\brho$, we will perform the calculations for a general pure state $\ket{\varphi}$ instead. The results we obtain will hold for any pure state $\ket{\varphi}$. Thus, they will extend by convexity to the post-commitment state $\brho$ as well, since we can always write $\brho = \sum_k p_k \ket{\varphi_k}\bra{\varphi_k}$ for some collection of pure states $\{\ket{\varphi_k}\}$.

As a first step in the computation, we remark that for every $i,j\in[\ell]$, it holds that  $P_{Z_i}$ and $P_{Z_j}$ commute and $P_{X_i}$ and $P_{X_j}$ commute. This follows from the fact that we defined all the $P_{Z_i}$ with respect to the same unitary $U_0$ and defined all the $P_{X_i}$ with respect to the same unitary $U_1$.  
Thus, for an input state $\ket{\varphi}$, the output of the extractor can be written as
\[
\frac{1}{{(2\sqrt{2})^\ell}}
\sum_{\vecr,\vecs,\vecs'\in\{0,1\}^\ell} \ket{\vecs'}_\coin \tensor \vecZ^\vecs \vecX^{\vecr}\ket{\vecs'}_\cA\tensor  P_{\vecX_{I}}^{\vecr_I} P_{\vecX_{J}}^{\vecr_{J}} P_{\vecZ_{J}}^{\vecs_{J}} P_{\vecZ_I}^{\vecs_I}\ket{\varphi}_\cB. 
\]
\paragraph{Measuring the $I$ registers of~$\cA$ in the standard basis.} 
Now, we imagine measuring the $I$ registers of $\cA$ in the standard basis; we denote these registers by $\cA_I$. When we measure them we obtain an outcome which we will denote $\veca_I$. The unnormalized post-measurement state is obtained by applying the projector $I \otimes \ket{\veca_I}\bra{\veca_I}_{\cA_I} $ to the state, where the factor of identity acts on all registers other than $\cA_I$. To calculate what happens, let us examine what happens when we apply the projector $\ket{\veca_I}\bra{\veca_I}$ to $\vecZ^{\vecs_I}_I \vecX^{\vecr_I}_I \ket{\vecs'_I}$.
Note that 
\begin{align*}
&\bra{\veca_I} \vecZ_I^{\vecs_I} \vecX_I^{\vecr_I}\ket{\vecs'_I}= \\
&\prod_{i\in I} \bra{a_i} Z_i^{s_i} X_i^{r_i}\ket{s'_i} =\\
&\prod_{i\in I} \bra{a_i} Z_i^{s_i}\ket{s'_i\oplus r_i} =\\
&\prod_{i\in I} \bra{a_i} (-1)^{s_i\cdot a_i}\ket{s'_i\oplus r_i}=\\
&\prod_{i\in I} (-1)^{s_i\cdot a_i}\bra{a_i} \ket{s'_i\oplus r_i},\\
\end{align*}
where for every $i\in I$, $\bra{a_i}\ket{s'_i\oplus r_i}$ is $1$ if $a_i\oplus r_i = s'_i$, and $0$ otherwise. This means that if we obtain an outcome $\veca_I$, then we force the $\vecs'_I$ register to be $\veca_I \oplus \vecr_I$. This means that the sum over $\vecs'$ collapses to a sum over $\vecs'_J$, since $J$ is the complement of $I$.
Thus, we obtain the unnormalized post-measurement state
$$
\frac{1}{({2\sqrt{2}})^\ell}\sum_{\vecr,\vecs\in\{0,1\}^\ell,\vecs'_J\in\{0,1\}^{|J|}}(-1)^{\vecs_I\cdot \veca_I}\ket{\veca_I\oplus \vecr_I}_{\coin_I}\ket {\vecs'_J}_{\coin_J}\tensor \ket{\veca_I}_{\cA_I}\tensor Z_J^{\vecs_J}X_J^{\vecr_J}\ket{\vecs'_J}_{\cA_J}\tensor  P_{\vecX_{I}}^{\vecr_I} P_{\vecX_{J}}^{\vecr_{J}} P_{\vecZ_{J}}^{\vecs_{J}} P_{\vecZ_I}^{\vecs_I}\ket{\varphi}_\cB
$$
Note that 
$$\frac{1}{2^{|I|}}\sum_{\vecs_I\in\{0,1\}^{|I|}}(-1)^{\vecs_I\cdot \veca_I}P_{Z_I}^{\vecs_I}=
\frac{1}{2^{|I|}}\prod_{i\in I}\sum_{s_i\in\{0,1\}}(-1)^{s_i\cdot a_i}P_{Z_i}^{s_i}=
\prod_{i\in I}\frac{I+(-1)^{a_i}P_{Z_i}}{2}
$$
and thus the state we obtain after the projection is equal to 
$$
 \frac{1}{2^{|J|}\cdot {2}^{\ell/2}}\sum_{\vecr\in\{0,1\}^\ell,\vecs_J,\vecs'_J\in\{0,1\}^{|J|}}\ket{\veca_I\oplus \vecr_I}\tensor \ket{\vecs'_J}\tensor \ket{\veca_I}_{\cA_I}\tensor  Z_J^{\vecs_J}X_J^{\vecr_J}\ket{\vecs'_J}_{\cA_J}\tensor
 P_{\vecX_I}^{\vecr_I}P_{\vecX_J}^{\vecr_J} P_{\vecZ_J}^{\vecs_J}\prod_{i\in I}\left(\frac{I+(-1)^{a_i}P_{Z_i}}{2}\right)
\ket{\varphi}_\cB. 
$$
Denoting by $\Pi_{P_{\vecZ_I,\veca_I}}=\prod_{i\in I}\frac{I+(-1)^{a_i}P_{Z_i}}{2}$, the above projected state is equal to 
\begin{align}
\ket{\Psi_{\veca_I}} = 
  \frac{1}{2^{|J|}\cdot {2}^{\ell/2}}\sum_{\vecr\in\{0,1\}^\ell,\vecs_J,\vecs'_J\in\{0,1\}^{|J|}} &\ket{\veca_I\oplus \vecr_I}_{\coin_I} \tensor \ket{\vecs'_J}_{\coin_J}  \nonumber \\
&\tensor \ket{\veca_I}_{\cA_I}\tensor  Z_J^{\vecs_J}X_J^{\vecr_J}\ket{\vecs'_J}_{\cA_J}\tensor
 P_{\vecX_I}^{\vecr_I}P_{\vecX_J}^{\vecr_J} P_{\vecZ_J}^{\vecs_J}\Pi_{P_{\vecZ_I,\veca_I}}
\ket{\varphi}_\cB. \label{eq:ext-post-standard}
\end{align}


The square norm of this unnormalized state is the probability that the measurement returns outcome $\veca_I$. We now calculate this:

\begin{align}
    \Pr[\veca_I] &= \frac{1}{2^{2|J|} \cdot 2^\ell} \sum_{\vecr \in \{0,1\}^\ell, \vecs'_J \in \{0,1\}^{|J|}} \| \sum_{\vecs_J \in \{0,1\}^{|J|}} Z^{\vecs_J}_J X^{\vecr_J}_J \ket{\vecs'_J}_{{\cal{A}}_J} \otimes P^{\vecr_I}_{\vecX_I} P^{\vecr_J}_{\vecX_J} P^{\vecs_J}_{\vecZ_J}  \Pi_{P_{\vecZ_I, \veca_I}} \ket{\varphi}_{\cB}\|^2 \nonumber \\
    &=\frac{1}{2^{2|J|} \cdot 2^\ell} \sum_{\vecr \in \{0,1\}^\ell, \vecs'_J \in \{0,1\}^{|J|}} \| \sum_{\vecs_J \in \{0,1\}^{|J|}} (-1)^{\vecs_J \cdot (\vecs'_J + \vecr_J)} \ket{\vecs'_J + \vecr_J}_{{\cal{A}}_J} \otimes P^{\vecr_I}_{\vecX_I} P^{\vecr_J}_{\vecX_J} P^{\vecs_J}_{\vecZ_J}  \Pi_{P_{\vecZ_I, \veca_I}} \ket{\varphi}_{\cB}\|^2 \nonumber \\
    &= \frac{1}{2^{2|J|} \cdot 2^\ell} \sum_{ \vecr \in \{0,1\}^\ell, \vecs'_J \in \{0,1\}^{|J|}} \sum_{\vecs_J, \vecs''_J \in \{0,1\}^{|J|}} (-1)^{(\vecs_J + \vecs''_J) \cdot (\vecs'_J + \vecr_J)} \bra{\varphi} \Pi_{P_{\vecZ_I, \veca_I}} P^{\vecs''_J}_{\vecZ_J} P^{\vecs_J}_{\vecZ_J} \Pi_{P_{{\vecZ_I, \veca_I}}} \ket{\varphi}_{\cB} \nonumber \\
    &= \frac{1}{2^{2|J|} } \sum_{ \vecs'_J \in \{0,1\}^{|J|}} \sum_{\vecs_J \in \{0,1\}^{|J|}}  \bra{\varphi} \Pi_{P_{\vecZ_I, \veca_I}}  \Pi_{P_{{\vecZ_I, \veca_I}}} \ket{\varphi}_{\cB} \nonumber \\
    &= \bra{\varphi} \Pi_{P_{\vecZ_I, \veca_I}}  \Pi_{P_{{\vecZ_I, \veca_I}}} \ket{\varphi}_{\cB}. \label{eq:prai}
\end{align}
Thus, we have shown that the outcome distribution from the extracted state is identical to the outcome distribution from measuring the original state $\ket{\varphi}$.
\paragraph{Measuring the $J$ registers of~$\cA$ in the Hadamard basis.} 

Now, we imagine taking the standard basis post-measurement state $\ket{\Psi_{\veca_I}}$, and then further measuring the $J$ registers of $\cA$ in the Hadamard basis. We denote these registers by $\cA_J$ and the outcome by $\veca_J$. To obtain the unnormalized post-measurement state after this measurement, we apply the projector $H^{\tensor |J|}\ket{\veca_J}\bra{\veca_J}H^{\tensor |J|}$ to the $J$ registers of $\cA$.
Note that 
\begin{align*}
   &\bra{\veca_J}H^{\tensor |J|} Z_J^{\vecs_J}X_J^{\vecr_J}\ket{\vecs'_J}=\\
   &\prod_{j\in J} \bra{a_j}H Z_j^{s_j}X_j^{r_j}\ket{s'_j}=\\
    &\prod_{j\in J} \bra{a_j}H Z_j^{s_j}\ket{s'_j\oplus r_j}=\\
     &\prod_{j\in J} \bra{a_j}H (-1)^{s_j\cdot (s'_j\oplus r_j)}\ket{s'_j\oplus r_j}=\\
     &\prod_{j\in J}\frac{(-1)^{s_j\cdot(s'_j\oplus r_j)}}{\sqrt{2}}\bra{a_j}(\ket{0}+(-1)^{s'_j\oplus r_j}
\ket{1})=\\
&\frac{1}{2^{|J|/2}} \prod_{j\in J}(-1)^{(s_j\oplus a_j)\cdot(s'_j\oplus r_j)}\triangleq \beta_J^{\vecs'_J}
\end{align*}
Thus we obtain the state
$$
\frac{1}{2^{|J|}\cdot {2}^{\ell/2}}\sum_{\vecr\in\{0,1\}^\ell,\vecs_J,\vecs'_J\in\{0,1\}^{|J|}}\ket{\veca_I\oplus \vecr_I}\tensor \beta_J^{\vecs'_J} \ket{\vecs'_J}\tensor \ket{\veca_I}_{\cA_I}\tensor  H^{\otimes |J|}\ket{\veca_J}_{\cA_J}\tensor 
 P_{\vecX_I}^{\vecr_I}P_{\vecX_J}^{\vecr_J} P_{\vecZ_J}^{\vecs_J}\Pi_{P_{\vecZ_I,\veca_I}}
\ket{\varphi}_\cB. 
$$
Next, we observe that
\begin{align}
    \sum_{\vecs'_J} \beta_J^{\vecs'_J} \ket{\vecs'_J} &= \frac{1}{\sqrt{2^{|J|}}} \sum_{\vecs'_J} (-1)^{(\vecs_J \oplus \veca_J)\cdot (\vecs'_J \oplus \vecr_J)} \ket{\vecs'_J} \\
    &= (-1)^{(\vecs_J \oplus \veca_J) \cdot \vecr_J} H^{\otimes |J|} \ket{\vecs_J \oplus \veca_J} \label{eq:beta-sprime-sum}
\end{align}
Thus, applying \Cref{eq:beta-sprime-sum} to simplify the sum over $\vecs'_J$, we can write this as
\begin{align*} 
\ket{\Psi_{\veca_I, \veca_J}} = \frac{1}{2^{\ell/2}2^{|J|}} \sum_{\vecr \in \{0,1\}^\ell, \vecs_J \in \{0,1\}^{|J|}} (-1)^{(\vecs_J \oplus \veca_J)\cdot \vecr_J} &\ket{\veca_I \oplus \vecr_I}_{\coin_I} \otimes H^{\otimes |J|} \ket{\vecs_J \oplus \veca_J}_{\coin_J} \\
&\otimes \ket{\veca_I}_{\cA_I} \otimes H^{\otimes |J|} \ket{\veca_J}_{\cA_J} \otimes P^{\vecr_I}_{\vecX_I} P^{\vecr_J}_{\vecX_J} P^{\vecs_J}_{\vecZ_j} \Pi_{P_{\vecZ_I, \veca_I}} \ket{\varphi}_\cB. \end{align*}
Note that 
\[
\frac{1}{2^{|J|}}\sum_{\vecr_J\in\{0,1\}^{|J|}}(-1)^{(\vecs_J\oplus \veca_J)\cdot \vecr_J}P_{\vecX_J}^{\vecr_J}=\prod_{j\in J}\frac{I+(-1)^{s_j\oplus a_j}P_{X_j}}{2}.
\]
Let us define \[\Pi_{P_{\vecX_j},\vecs_J\oplus \veca_J}= \prod_{j \in J}
\frac{I+(-1)^{s_j\oplus a_j}P_{X_j}}{2}.
\]
Then we can rewrite $\ket{\Psi_{\veca_I, \veca_J}}$ as
\begin{align*}
\ket{\Psi_{\veca_I, \veca_J}} = \frac{1}{{2}^{\ell/2}}\sum_{\vecr_I \in\{0,1\}^{|I|},\vecs_J \in\{0,1\}^{|J|}} \ket{\veca_I\oplus \vecr_I}_{\coin_I} &\tensor H^{\otimes |J|} \ket{\vecs_J\oplus \veca_J}_{\coin_J}
\tensor \ket{\veca_I}_{\cA_I}\tensor  H^{\otimes |J|}\ket{\veca_J}_{\cA_J} \\
&\tensor 
 P_{\vecX_I}^{\vecr_I}\Pi_{P_{\vecX_J},\veca_J\oplus \vecs_J}P_{\vecZ_J}^{\vecs_J}\Pi_{P_{\vecZ_I,\veca_I}}
\ket{\varphi}_\cB. 
\end{align*}
The square norm of this unnormalized state is the probability that the measurement returns outcome $\veca_J$. We now calculate this:
\begin{align}
    \Pr[\veca_I, \veca_J] &= \frac{2^{|I|}}{2^\ell} \left\|  \sum_{\vecs_J \in \{0,1\}^{|J|}}   H^{\otimes |J|}\ket{\vecs_J\oplus \veca_J}_{\coin_J}\otimes  \Pi_{P_{\vecX_J},\veca_J\oplus \vecs_J}P_{\vecZ_J}^{\vecs_J}\Pi_{P_{\vecZ_I,\veca_I}}
\ket{\varphi}_\cB\right\|^2 \nonumber \\
&=\frac{1}{2^{|J|}} \sum_{\vecs_J \in \{0,1\}^{|J|}} \| \Pi_{P_{\vecX_j, \veca_J \oplus \vecs_J}} P_{\vecZ_J}^{\vecs_J}\Pi_{P_{\vecZ_I,\veca_I}}
\ket{\varphi}_\cB \|^2 \nonumber \\
&= \E_{\vecs_J \in \{0,1\}^{|J|}} \| \Pi_{P_{\vecX_j, \veca_J \oplus \vecs_J}} P_{\vecZ_J}^{\vecs_J}\Pi_{P_{\vecZ_I,\veca_I}}
\ket{\varphi}_\cB \|^2.  \label{eq:praj}
\end{align}
This equation can be interpreted operationally as follows: the probability of obtaining an outcome $(\veca_I, \veca_J)$ by measuring the extracted state is \emph{equal} to the probability of obtaining this outcome by the following procedure acting on $\ket{\varphi}$:
\begin{enumerate}
    \item First, measure the observables $P_{Z_i}$ for every $i \in I$ on $\ket{\varphi}$, obtaining an outcome $\veca_I$.
    \item Next, sample a string $\vecs_J \in \{0,1\}^{|J|}$ uniformly at random and apply $P_{\vecZ_J}^{\vecs_J}$ to the state.
    \item Next, measure the observables $P_{X_j}$ for every $j \in J$ on the state, obtaining an outcome $\vecu_{J}$. 
    \item Set $\veca_J = \vecu_J \oplus \vecs_J$ and return $(\veca_I, \veca_J)$. 
\end{enumerate}

\paragraph{The proof of \Cref{eqn:real-ind-ideal-semi}} 
We first define a new distribution, which we denote by $\widehat{\Real}^{\cC^*.\Open_{[\ell]}}(\secp, \vecb,\bsigma)$. This distribution is identical to $\Real^{\cC^*.\Open_{[\ell]}}(\secp, \vecb,\bsigma)$ except that it does not run the $\Ver$ algorithm (i.e., it does not run \Cref{item:real:ver} of the definition of $\Real$ in \Cref{def:binding}), and simply sets $\vecm=\Out(\sk,\vecy, \vecb,\vecz)$.  We note that by \Cref{lem:gentle-mixed},
\begin{equation}\label{eqn:real-ind-hat-real}
 \widehat{\Real}^{\cC^*.\Open_{[\ell]}}(\secp, \vecb,\bsigma)\stackrel{\sqrt{2\delta}}\approx 
 \Real^{\cC^*.\Open_{[\ell]}}(\secp, \vecb,\bsigma) 
\end{equation}
where recall $2\delta$ is the probability that $\Ver$ rejects $\cC^*.\Open_{[\ell]}(\secp, \vecb,\bsigma)$ (see \Cref{eqn:delta[ell]}).
Therefore to prove \Cref{eqn:real-ind-ideal-semi} it suffices to prove that
\begin{equation}\label{eqn:hat-real-ind-ideal}
\widehat{\Real}^{\cC^*.\Open_{[\ell]}}(\secp, \vecb,\bsigma) \stackrel{2\sqrt{\delta}}\approx \Ideal^{\Ext,\cC^*.\Open}(\secp, \vecb,\bsigma)
\end{equation}
To this end, we first claim that
\[
(\pk,\vecy,\vecb,\vecm_{\Sim,I})\equiv (\pk,\vecy,\vecb,\vecm_{\widehat{\Real},I})
\]
where $(\pk,\vecy,\vecb,\vecm_{\Sim,I})$ is distributed by generating 
\[(\pk,\vecy,\vecb,\vecm)\gets {\Ideal}^{\Ext,\cC^*.\Open}(\secp, \vecb,\bsigma)\]
and outputting $(\pk,\vecy,\vecb,\vecm_I)$, and 
$(\pk,\vecy,\vecb,\vecm_{\widehat{\Real},I})$ is distributed by generating 
\[(\pk,\vecy,\vecb,\vecm)\gets\widehat{\Real}^{\cC^*.\Open_{[\ell]}}(\secp, \vecb,\bsigma)
\] and outputting $(\pk,\vecy,\vecb,\vecm_I)$.
To see why this is true, recall that \Cref{eq:prai} implies that for a given $\pk, \vecy, \vecb$, the outcome $\vecm_{\Sim, I}$, which is equal to $\veca_I$ in the notation used in that equation, is distributed according to the outcome of measuring $P_{Z_i}$ on the qubits $i\in I$ qubits of the post-commitment state. Moreover, $P_{Z_i}$ is defined so that it exactly matches the action of $\widehat{\Real}$ since both do not run $\Ver$. 

\begin{remark}\label{remark:P-and-P'} We note that the observable $P_{Z_i}$ was defined with respect to the opening algorithm $\cC^*.\Open$ and we are considering $\widehat{\Real}$ with respect to the opening algorithm $\cC^*.\Open_{[\ell]}$.  The observable $P_{Z_i}$ corresponding to $\cC^*.\Open$, when viewed as a unitary, is different from observable corresponding to $\cC^*.\Open_{[\ell]}$, denoted by $P'_{Z_i}$, when viewed as a unitary.
In particular, recall that 
\[
 P_{Z_i}=U_{0}^\dagger \Out_{i,0}^{\dagger} Z_{\out_i} \Out_{i,0} U_{0}
 \]
 whereas 
 \[
 P'_{Z_i}=U_{0}^\dagger U_{\mathsf{post}}^\dagger \Out_{i,0}^{\dagger} Z_{\out_i} \Out_{i,0} U_{\mathsf{post}} U_{0},
 \]
 where $U_{\mathsf{post}}$ is the unitary that does some post-processing to the $\mathsf{open}_i$ register to ensure that measuring it will not disturb the state in a detectable way. Despite the fact that $P_{Z_i}$ and $P'_{Z_i}$ are different unitaries, on the subspace where the ancila registers are initialized to $\ket{0}$, they are identical operators.  In particular, $P'_{Z_i}$ preserves the subspace where the ancila registers are initialized to $\ket{0}$.
 \end{remark}


To avoid cluttering of notation, from now on we denote $\vecm_{\widehat{\Real},I}$ and $\vecm_{\Sim,I}$ by $\vecm_I$.  Denote by 
\[\brho'_{I}= \Pi_{P_{\vecZ_{I}},\vecm_{I}}[\brho]\]
where $\brho$ is post-commitment state and $\vecm_{I}$ is distributed as $\vecm_{\Sim,I}$.
We note that in the experiment $\widehat{\Real}^{\cC^*.\Open_{[\ell]}}(\secp,\vecb,\bsigma)$, the post state after measuring $\vecm_I$ is  $\brho'_{I}$.  This follows from \Cref{remark:P-and-P'}.
We prove that 
\begin{equation}\label{eqn:final}
(\pk,\vecy,\vecb,\vecm_{I},\vecm_{\widehat{\Real},J})\stackrel{2\sqrt{\delta}}\approx (\pk,\vecy,\vecb, \vecm_{I},\vecm_{\Sim,J})
\end{equation}
where $\vecm_{\widehat{\Real},J}$ is obtained as follows:
\begin{enumerate}
\item Compute $\vecz_J\gets \cC^*.\Open_{[\ell]}(\brho'_I,(J,\vecb_J))$.
\item For every $j\in J$ let $\vecm_{\widehat{\Real},j}=\Out(\sk,\vecy,(j,1),\vecz_j)$
\item Output  $\vecm_{\widehat{\Real},J}=\{\vecm_{\widehat{\Real},j}\}_{j\in J}$. 
\end{enumerate}
To describe how $\vecm_{\Sim,J}$ is obtained, we take the procedure obtained immediately below \Cref{eq:praj}, and apply the definitions of the operational observable $P_X$, to obtain the following:
\begin{enumerate}
    \item Sample at random $\vecs_J\gets \{0,1\}^{|J|}$.
    \item Compute $\vecz_J\gets \cC^*.\Open_{[\ell]}(P_{Z_J}^{\vecs_J}[\brho'_I],(J,\vecb_J))$.

    \item For every $j\in J$ let  $u_j =\Out(\sk,\vecy,(j,1),\vecz_j)$.
    \item  For every $j\in J$ let $\vecm_{\Sim,j}=u_j\oplus \vecs_j$.
    \item Output $\vecm_{\Sim,J}=(\vecm_{\Sim,j})_{j\in J}$.
\end{enumerate}
We prove \Cref{eqn:final} separately for the non-succinct and the semi-succinct versions. For the non-succinct version we rely on the adaptive hardcore bit property and for the semi-succinct version we rely on the distributional strong adaptive hardcore bit property.  In both cases, we assume without loss of generality that $\Open_{[\ell]}$ opens in the Hadamard basis honestly, by measuring the relevant qubits in the standard basis. For every $j\in J$, we denote by $\cO_{j}$ the $n+1$ registers that are measured to obtain the opening of the $j$'th committed qubit in the Hadamard basis.

\paragraph{Proof of \Cref{eqn:final} in the non-succinct setting.}

Let $\Pi_\Ver[\brho'_I]$ denote the state $\brho'_I$ projected to 
\[\Ver(\sk,\vecy,(J,1^{|J|}),\Open(\brho'_I,(J,1^{|J|})))=1.
\]
By \Cref{lem:gentle-mixed},
\[
\Pi_\Ver[\brho'_I]\stackrel{\sqrt\delta}\equiv\brho'_I.
\]
Therefore to prove \Cref{eqn:final} it suffices to prove that 
\begin{equation}\label{eqn:final'}
(\pk,\vecy,\vecb,\vecm_{I},\vecm^*_{{\Real},J})\approx (\pk,\vecy,\vecb, \vecm_{I},\vecm^*_{\Sim,J})
\end{equation}
 where $\vecm^*_{{\Real},J}$ is distributed as $\vecm_{\widehat{\Real},J}$ except that $\brho'_I$ is replaced with $\Pi_\Ver[\brho'_I]$.  Similarly, $\vecm^*_{\Sim,J}$ is distributed as $\vecm_{\Sim,J}$  except that $\brho'_I$  is replaced with $\Pi_\Ver[\brho'_I]$.\\


To prove \Cref{eqn:final'} it suffices to prove that
\begin{equation}\label{eqn:ideal-real-H-wo-about}
\left(\pk,\vecy,\vecb,\vecm_{I},\{\vecm_{j,0}\}_{j\in J}\right)\approx\left(\pk,\vecy,\vecb,\vecm_{I},\{\vecm_{j,1}\oplus 1\}_{j\in J}\right)
\end{equation}
where for every $j\in J$ and $u\in\{0,1\}$, \[(\vecz_{j,u},\brho'_{j,u})=\cC^*.\Open(P^u_{Z_j}\Pi_\Ver[\brho'_I],(j,1))~~\mbox{ and }~~\vecm_{j,u}=\Out(\sk,\vecy,(j,1),\vecz_{j,u}).
\]

We prove that for every $j\in J$,
\begin{align}\label{eqn:hyb:succ}
\left(\pk,\sk_{-(j,0)},\vecy,\vecb,\vecm_I,\vecm_{j,0},\brho'_{j,0}\right)\approx \left(\pk,\sk_{-(j,0)},\vecy,\vecb,\vecm_I,\vecm_{j,1}\oplus 1,\brho'_{j,1}\right),
\end{align} 
where $\sk_{-(j,0)}$ denotes all the secret keys except $\sk_{j,0}$.  Namely,
\[
\sk_{-(j,0)}\triangleq \left(\sk_{[\ell],1},\ldots,\sk_{[\ell],n+1},\sk_{[\ell]\setminus\{j\},0}\right).
\]
We next argue that \Cref{eqn:hyb:succ} implies \Cref{eqn:ideal-real-H-wo-about}.
To this end, we first notice that $P_{Z_j}$ and $\cC^*.\Open(\cdot,(j,1))$ only touch the registers corresponding to the $j$'th committed qubit. This follows from our assumption that $\cC^*.\Open$ behaves honestly when opening in the Hadamard basis.  This in turn implies that  for every $u\in\{0,1\}$ it holds that $\brho'_{j,u}$ and $\Pi_\Ver[\brho'_I]$ are distributed identically on the registers that do not correspond to the $j$'th committed qubit. 
Thus, \Cref{eqn:hyb:succ} implies that
\begin{align}\label{eqn:hyb:succ2}
\left(\pk,\sk_{-(j,0)},\vecy,\vecb,\vecm_I,\vecm_{j,0},\Pi_\Ver[\brho'_I]_{\{\cO_{j}\}_{j\in J\setminus\{j\}}}\right)\approx \left(\pk,\sk_{-(j,0)},\vecy,\vecb,\vecm_I,\vecm_{j,1}\oplus 1,\Pi_\Ver[\brho'_I]_{\{\cO_{j}\}_{j\in J\setminus\{j\}}}\right).
\end{align} 
We next note that $\vecm_{j,u}$ is a $\QPT$ function of $\Pi_\Ver[{\brho'_I}]_{\cO_j}$ and $\sk_{j}$.
This, together with a hybrid argument implies that indeed
\Cref{eqn:hyb:succ2} implies \Cref{eqn:ideal-real-H-wo-about}, as desired.

Thus, we focus on proving \Cref{eqn:hyb:succ}.
Fix $j\in J$ and consider the mixed state
$$
\brho_{\mathsf{mix,j}}=\frac12 \Pi_{\Ver}[\brho'_I] +\frac12 P_{Z_{j}}\Pi_{\Ver}[\brho'_I] 
$$
Note that this state can be generated efficiently, with probability $1-\delta$, from $\brho$ given $(\sk_{[\ell],1},\ldots,\sk_{[\ell],n+1})$. In addition, note that $\brho_{\mathsf{mix,j}}$
is identical to the state $\Pi_{\Ver}[\brho'_I]$ after measuring it in the $P_{Z_{j}}$ basis. 
Recall that we assume that $\cC^*.\Open$ behaves honestly on the Hadamard basis.  Thus, the $(n+1)^2$ qubits of this projected state $\brho_{\mathsf{mix,j}}$ corresponding to the $j$'th committed qubit are in superposition over $\ket{d_1,\vecx'_1}\ldots\ket{d_{n+1},\vecx'_{n+1}}$ such that for every $i\in[n+1]$, $\vecy_{j,i}=\Eval(\pk_{j,i},d_i,\vecx'_i)$.

Let
$$(\vecz^*,\brho^*)\gets \cC^*.\Open\left(\brho_{\mathsf{mix,j}},(j,1)\right).
$$
By the adaptive hardcore bit property (w.r.t.~$\pk_{j,0}$), letting $m^*=\Out(\sk,\vecy,(j,1),\vecz^*)$, 
\begin{equation}\label{eqn:m*}
(\pk,\sk_{-(j,0)},\vecy,\vecb,\vecm_I,m^*,\brho^*)\approx (\pk,\sk_{-(j,0)},\vecy,\vecb,\vecm_I,U,\brho^*)
\end{equation}
where $U$ is the uniform distribution over $\{0,1\}$.
Note that~$m^*$ is a random variable that 
with probability~$\frac12$ is distributed identically to $\vecm_{j,0}$ and with probability~$\frac12$ is distributed identically to $\vecm_{j,1}$.
We next argue that \Cref{eqn:m*} implies that 
\begin{align*}
&(\pk,\sk_{-(j,0)},\vecy,\vecb,\vecm_I,\vecm_{j,0},\brho'_{j,0})\approx \\&(\pk,\sk_{-(j,0)},\vecy,\vecb,\vecm_I,\vecm_{j,1}\oplus 1,\brho'_{j,1}),
\end{align*}
as desired. 
To this end, suppose for contradiction that there exists a $\BQP$ algorithm $\A$ and a non-negligible $\epsilon>0$ such that 
\begin{align*}
&\Pr[\A(\pk,\sk_{-(j,0)},\vecy,\vecb,\vecm_I,\vecm_{j,1}\oplus 1,\brho'_{j,1})=1]-\\
&\Pr[\A(\pk,\sk_{-(j,0)},\vecy,\vecb,\vecm_I,\vecm_{j,0},\brho'_{j,0})=1]\geq \epsilon.
\end{align*}
Denote by 
\[
p_u=\Pr[\A(\pk,\sk_{-(j,0)},\vecy,\vecb,\vecm_I,\vecm_{j,u},\brho'_{j,u})=1]
\]
and denote by 
$$
p'_1=\Pr[\A(\pk,\sk_{-(j,0)},\vecy,\vecb,\vecm_I,\vecm_{j,1}\oplus 1,\brho'_{j,1})=1]
$$
Note that  
$$
\Pr[\A(\pk,\sk_{-(j,0)},\vecy,\vecb,\vecm_I,\vecm^*,\brho^*)=1]=\frac12 p_0+\frac12 p_1.
$$
On the other hand 
$$
\Pr[\A(\pk,\sk_{-(j,0)},\vecy,\vecb,\vecm_I,U,\brho_{j,1}')=1]=\frac12 p'_1+\frac12 p_1,
$$
which by the collapsing property implies that there exists a negligible function~$\mu$ such that
$$
\Pr[\A(\pk,\sk_{-(j,0)},\vecy,\vecb, \vecm_I,U, \brho^*)=1]=\frac12 p'_1+\frac12 p_1 \pm \mu
$$
This contradicts \Cref{eqn:m*} since
$$
\left(\frac12 p'_1+\frac12 p_1\pm\mu\right)-\left(\frac12 p_0+\frac12 p_1\right)=\frac12 (p'_1-p_0)-\mu\geq \frac{\epsilon}{2} \pm\mu,
$$
which is non-negligible.

\paragraph{The proof of \Cref{eqn:final}  for the semi-succinct version.}

We prove \Cref{eqn:final} holds if the underlying $\TCF $ family satisfies the strong adaptive hardcore bit property (see \Cref{def:TCF}) and the distributional strong adaptive hardcore bit property (\Cref{def:dist-stat-HCB}).
Let $\Pi_\Ver[\brho'_I]$ denote the state $\brho'_I$ projected to 
\[\Ver(\sk,\vecy,(J,0^{|J|}),\cC^*.\Open(\brho'_I,(J,0^{|J|})))=1.\footnote{Note that this is different from the non-succinct case where we used $\Pi_\Ver[\brho'_I]$ to denote the state $\brho'_I$ projected to 
$\Ver(\sk,\vecy,(J,1^{|J|}),\cC^*\Open(\brho'_I,(J,1^{|J|})))=1$.}
\]
By \Cref{lem:gentle-mixed},  
\[
\Pi_\Ver[\brho'_I]\stackrel{\sqrt{\delta}}\equiv\brho'_I.
\]
Therefore to prove \Cref{eqn:final} it suffices to prove that 
\begin{equation}\label{eqn:final''}
(\pk,\vecy,\vecb,\vecm_{I},\vecm^*_{{\Real},J})\approx (\pk,\vecy,\vecb, \vecm_{I},\vecm^*_{\Sim,J})
\end{equation}
 where $\vecm^*_{{\Real},J}$ is distributed as $\vecm_{\widehat{\Real},J}$ except that $\brho'_I$ is replaced with $\Pi_\Ver[\brho'_I]$.  Similarly, $\vecm^*_{\Sim,J}$ is distributed as $\vecm_{\Sim,J}$  except that $\brho'_I$  is replaced with $\Pi_\Ver[\brho'_I]$.\\
We prove \Cref{eqn:final''} via a hybrid argument. Namely, denoting by $k=|J|$, we prove that for every $\alpha\in[k]$,
\begin{align}\label{eqn:induction}
\left(\pk,\vecy,\vecb,\vecm_I,\vecm^{(\alpha-1)}_J\right)\approx
\left(\pk,\vecy,\vecb,\vecm_I,\vecm^{(\alpha)}_J\right)
\end{align}
where $\vecm^{(\beta)}_J$ is distributed exactly as $\vecm_{\Sim,J}$ accept that rather than choosing $\vecs_J\gets \{0,1\}^{|J|}$ at random, we only choose the first $\beta$ coordinates randomly and the rest we set to zero.  Namely, we choose $\vecs_1,\ldots,\vecs_{\beta}\gets\{0,1\}$ and set $\vecs_{\beta+1}=\ldots=\vecs_k=0$.
\Cref{eqn:induction} implies that 
\begin{align*}
\left(\pk,\vecy,\vecb,\vecm_I,\vecm^{(0)}_J\right)\approx
\left(\pk,\vecy,\vecb,\vecm_I,\vecm^{(k)}_J\right)
\end{align*}
thus proving \Cref{eqn:ideal-real-H-wo-about} since 
\[
\left(\pk,\vecy,\vecb,\vecm_I,\vecm^{(0)}_J\right)\equiv \left(\pk,\vecy,\vecb,\vecm_I,\vecm^*_{\widehat{\Real},J}\right)~~\mbox{ and }~~
\left(\pk,\vecy,\vecb,\vecm_I,\vecm^{(k)}_J\right)\equiv \left(\pk,\vecy,\vecb,\vecm_I,\vecm^*_{\Sim,J}\right).
\]
Fix any $\alpha\in [k]$ and we prove \Cref{eqn:induction} for this $\alpha$.   Denote by 
$J=\{j_1,\ldots,j_k\}\subseteq[\ell]$. For every $\vecs_{[\alpha-1]}=(s_1,\ldots,s_{\alpha-1})\in\{0,1\}^{\alpha-1}$ consider the states
\begin{equation}\label{eqn:rho*}\brho_{\vecs_{[\alpha-1]}}=\prod_{i=1}^{\alpha-1} P^{s_i}_{Z_{j_i}}\Pi_\Ver[\brho'_I]~~\mbox{ and }~~\brho^{*}_{\vecs_{[\alpha-1]}}=\frac12\brho_{\vecs_{[\alpha-1]}}+\frac12 P_{Z_{j_\alpha}}[\brho_{\vecs_{[\alpha-1]}}].
\end{equation}

The main ingredient in the proof of \Cref{eqn:induction} is following proposition.
\begin{proposition}\label{claim:blah}
\begin{align}\label{eqn:blah}
&(\pk, \vecy,\vecb,\vecm_I,(\vecx_{j_i,0}\oplus\vecx_{j_i,1})_{i\in[k]},\vecd^*_{J\setminus\{j_{\alpha}\}}, \vecd^*_{j_{\alpha}}\cdot(1, \vecx_{j_\alpha,0}\oplus\vecx_{j_\alpha,1}))\approx\\
&(\pk, \vecy,\vecb,\vecm_I,(\vecx_{j_i,0}\oplus\vecx_{j_i,1})_{i\in[k]},\vecd^*_{J\setminus\{j_{\alpha}\}},U)\nonumber
\end{align}
where $\vecd^*_J$ is distributed by  measuring the registers $\cO_{j_1},\ldots,\cO_{j_k}$ of $\brho^*_{\vecs_{[\alpha-1]}}$ in the standard basis.
\end{proposition}

In what follows, we  use \Cref{claim:blah} to prove \Cref{eqn:induction} and then we prove  \Cref{claim:blah}.  To this end, for every  $u\in\{0,1\}$, denote by 
 \[\vecd^{(u)}_J=\left(\vecd^{(u)}_{j_1},\ldots,\vecd^{(u)}_{j_k}\right)
 \] 
 the values obtained by measuring registers $\cO_{j_1},\ldots,\cO_{j_k}$ of $P^{(u)}_{Z_{j_{\alpha}}}[\brho_{\vecs_{[\alpha-1]}}]$ in the standard basis.  For every $i\in[k]$, denote by 
\[m^{(u)}_{j_i}=\vecd^{(u)}_{j_i}\cdot(1, \vecx_{j_i,0}\oplus \vecx_{j_i,0})~~\mbox{ and }~~m^{*}_{j_i}=\vecd^{*}_{j_i}\cdot(1, \vecx_{j_i,0}\oplus \vecx_{j_i,0})
\]
and let 
\[
\vecm^{(u)}=\left(m^{(u)}_{j_1},\ldots,m^{(u)}_{j_k}\right)~~\mbox{ and }~~\vecm^{*}=\left(m^*_{j_1},\ldots,m^*_{j_k}\right).
\]
Fix any $\BQP$ adversary $\A$.  For every $u\in\{0,1\}$ denote by 
$$
p_u=\Pr[\A\left(\pk, \vecy,\vecb,\vecm_I, \vecm^{(u)}_J\right) =1],$$
and denote by
$$
p'_1=\Pr[\A\left(\pk, \vecy,\vecb,\vecm_I, \left\{\vecm^{(1)}_{J\setminus\{\alpha\}},m^{(1)}_{\alpha}\oplus 1\right\} \right)=1].
$$
To prove \Cref{eqn:induction} we need to prove that
\begin{equation}\label{eqn:ind:negl}
|p'_1-p_0|\leq \negl(\secp).
\end{equation}
Note that 
$$
\Pr\left[\A\left(\pk, \vecy,\vecb,\vecm_I, \vecm^*_J\right)=1\right]=\frac12 p_0+\frac12 p_1
$$
which by \Cref{claim:blah} implies that 
\[
\Pr[\A\left(\pk,\vecy,\vecb,\vecm_I, \left\{\vecm^*_{J\setminus\{\alpha\}},U\right\}\right)=1]=\frac12 p_0+\frac12 p_1\pm\negl(\secp)
\]
In addition note that
$$
\Pr[\A\left(\pk, \vecy,\vecb,\vecm_I, \left\{\vecm^{(1)}_{J\setminus \{\alpha\}}, U\right\}\right) =1] =\frac12 p'_1+\frac12 p_1.
$$
It remains to note that 
\[
\left(\pk,\vecy,\vecb,\vecm_I, \vecm^*_{J\setminus\{\alpha\}}\right)\equiv \left(\pk, \vecy,\vecb,\vecm_I, \vecm^{(1)}_{J\setminus \{\alpha\}}\right),
\]
which together with the equation above, implies that \Cref{eqn:ind:negl} indeed holds, 
thus proving \Cref{eqn:induction}, as desired.\\

The rest of the proof is dedicated to proving \Cref{claim:blah}, which we prove assuming the underlying $\TCF$ family has the strong adaptive hardcore bit property (see \Cref{def:TCF}) and the distributional strong adaptive hardcore bit property (\Cref{def:dist-stat-HCB}).

\paragraph{Proof of \Cref{claim:blah}.}

We start by defining $\QPT$ algorithms $\A$  and $C$ for the distributional strong adaptive hardcore bit property.
\begin{description}
    \item {\bf Algorithm $\A$.}
    It takes as on input $\pk_{0}$ generated by $\Gen_\TCF(1^\secp)$, and does the following:
\begin{enumerate}
    \item For every $i\in[n+1]$ generate $(\pk_{i},\sk_i)\gets \Gen_\TCF(1^\secp)$.
    \item Set $\pk=(\pk_{0},\pk_{1},\ldots,\pk_{n+1})$.
    \item Compute $(\vecy,\brho)\gets \cC^*.\Commit(\pk,\bsigma)$.
    
    \item Use $\sk_1$ to generate 
        $\brho'_I=\Pi_{P_{Z_I},\vecm_I}(\brho)$.
        \item Use $\sk_1,\ldots,\sk_{n+1}$ to compute $\Pi_\Ver[\brho'_I]$, which is the state $\brho'_I$ projected to \[\Ver\left(\left(\sk_1,\ldots,\sk_{n+1}\right),\vecy,\left(J,0^{|J|}\right),\cC^*.\Open\left(\brho'_I,\left(J,0^{|J|}\right)\right)\right)=1.\]
        This step fails with probability $\delta$, in which case $\A$ outputs~$\bot$.
        \item Use $\sk_1$ to compute  $\brho_{\vecs_{[\alpha-1]}}=\prod_{i=1}^{\alpha-1} P^{s_i}_{Z_{j_i}}\Pi_\Ver[\brho'_I]$.
        \item Use $\sk_1$ to compute $\brho^{*}_{\vecs_{[\alpha-1]}}=\frac12\brho_{\vecs_{[\alpha-1]}}+\frac12 P_{Z_{\alpha}}[\brho_{\vecs_{[\alpha-1]}}]$.
        
    \item Parse $\vecy=(\vecy_1,\ldots,\vecy_\ell)$ and parse $\vecy_{j_\alpha}=(\vecy_{j_\alpha,0},\vecy_{j_\alpha,1},\ldots,\vecy_{j_\alpha,n+1})$
    \item Measure the $P_{Z_{j_\alpha}}$ observable of the state $\brho^{*}_{\vecs_{[\alpha-1]}}$ to obtain a preimage $(b,\vecx_{j_\alpha,b})$ of $\vecy_{j_\alpha,0}$.  
    
    Denote the post-measurement state by $\brho^{**}$.  We assume without loss of generality that this state includes $(\pk,(\sk_1,\ldots,\sk_{n+1}),\vecy,\vecm_I)$.
    \item Output $(\vecy_{j_\alpha,0},b,\vecx_{j_\alpha,b},\brho^{**})$.
    
    We rename the register $\cO_{j_\alpha}$ of $\brho^{**}$ by $\cO_1$ and rename all the other registers by $\cO_2$.
    
\end{enumerate} 
\item {\bf Algorithm~$C$.}
It takes as input the secret vector $\sk_{0,\mathsf{pre}}$, corresponding to $\sk_0$, and the $\cO_2$ registers  of $\brho^{**}$, denoted by $\brho_{\cO_2}^{**}$, and does the following:
\begin{enumerate}
    \item For every $i\in[k]\setminus\{\alpha\}$ use $\sk_1,\ldots,\sk_{n+1}$ to coherently compute a preimage of $\vecy_{i,0}$ as follows:
    
    \begin{itemize}
        \item Coherently compute $\cC^*.\Open(\brho^{**},(j_i,0))$.
        \item Coherently run $\Ver((\sk_1,\ldots,\sk_{n+1}),\vecy_{j_i},(j_i,0),\cdot)$ to to coherently generate a preimage. 
    \end{itemize} 
    Note that since the state $\Pi_\Ver[\brho'_I]$ is the state $\brho'_I$ projected to the state where $\cC^*.\Open$ is always accepted when the opening is in the standard basis, we indeed get a coherent preimage of $\vecy_{j_i,0}$ with probability~$1$. 
    
    \item Use $\sk_{0,\mathsf{pre}}$ to compute (and measure) $\vecx_{j_i,0}\oplus \vecx_{j_i,1}$.
    
    Note that this is a deterministic quantity and hence does not disturb the state.

\item Measure all the registers in $\cO_2$ corresponding to $\{\cO_{j_i}\}_{i\in [k]\setminus\{\alpha\}}$ to obtain $\vecd^*_{J\setminus \{j_\alpha\}}$.
\item Output $\aux=\left(\pk,\vecy, \vecm_I,(x_{j_i,0}\oplus x_{j_i,1})_{i\in[k]\setminus\{\alpha\}}, \vecd^*_{J\setminus \{j_\alpha\}}\right)$.

Denote the residual state by ${\brho_\aux^{**}}$.
\end{enumerate}
\end{description}
To finish the proof of \Cref{claim:blah} it remains to prove that $\vecd^*_{j_\alpha}$ satisfies the desired min-entropy requirement.
This is captured in the proposition below.   

\begin{proposition}\label{claim:min-entropy}
With overwhelming probability over  $\aux\gets C(\sk_{0,\mathsf{pre}},{\brho_{\cO_2}^{**}})$ it holds that for every $\vecd'\in\{0,1\}^k$, 
\begin{equation*}  
\Pr[\Good(\vecx_{j_\alpha,0},\vecx_{j_\alpha,1},\vecd^*_{j_\alpha})=\vecd']=\negl(\secp).
\end{equation*}
where $\vecd^*_{j_\alpha}$ is obtained by measuring the $\cO_1$ registers of ${\brho_\aux^{**}}$ in the standard basis.
\end{proposition}
The proof of \Cref{claim:min-entropy} relies on the adaptive hardcore bit property of the underlying $\TCF$ family.  It also makes use the following fact, which was used previously, and which follows immediately from the definition of $P_{Z_j}$.\footnote{Recall that $P_{Z_j}$ does not check if $\Ver(\sk,\vecy,(j,0),\cdot)=1$.}
\begin{fact}\label{claim:rho*}
    For every $\alpha\in[k]$ and every $\vecs_{[\alpha-1]}\in\{0,1\}^{\alpha-1}$, the state $\brho_{\vecs_{[\alpha-1]}}$ and the state $\brho^*_{\vecs_{[\alpha-1]}}$ can be efficiently constructed from $(\pk,\vecy,\brho)$ and $\sk_1$.  
\end{fact}

 \paragraph{Proof of \Cref{claim:min-entropy}.}

Suppose for the sake of contradiction that there is a non-negligible~$\epsilon=\epsilon(\secp)$ such that with probability $\epsilon$ over  $(\aux,{\brho_\aux^{**}})\gets C(\sk_{0,\mathsf{pre}},{\brho_{\cO_2}^{**}})$ it holds that there exists a vector $\vecd'=\vecd'_\aux\in\{0,1\}^k$ such that 
\begin{equation}\label{eqn:claim-ent-cont}
\Pr[\Good(\vecx_{j_\alpha,0},\vecx_{j_\alpha,1},\vecd^*_{j_\alpha})=\vecd']\geq\epsilon
\end{equation}   
where $\vecd^*_{j_\alpha}$ is obtained by measuring the $\cO_1$ registers of ${\brho_\aux^{**}}$ in the standard basis.
Denote the set of all $\aux$ that satisfy \Cref{eqn:claim-ent-cont} by $\mathsf{BAD}$.
In the rest of this proof, for the sake of ease of notation, we denote by 
\[\vecd^*\triangleq\vecd^*_{j_\alpha}~~\mbox {and }~~(\vecx_0,\vecx_1)\triangleq (\vecx_{j_\alpha,0},\vecx_{j_\alpha,1}).
\]
We next denote by 
\[\vecd^*=(d^*_0,d^*_1,\ldots,d^*_n)
\]
and we argue that there exists a subset 
\begin{equation}\label{eqn:S}
S=\left\{\beta_1,\ldots,\beta_m\right\}\subseteq[n]
\end{equation}
of size $m=k^{0.1}$ (where $k=|\vecd'|$) for which there exists an (all powerful) algorithm $\B$ such that for every $\aux\in \mathsf{BAD}$ and for every $i\in [m]$
\begin{equation}\label{eqn:disturb}
\Pr[\B\left(\aux,\left(d^*_{1},\ldots,d^*_{j'_{\beta_{i}}-1}\right)\right)=d^*_{j'_{\beta_i}}] \geq 1-\frac{1}{\sqrt{k}}
\end{equation}
where the probability is over $\vecd^*$ obtained by measuring the $\cO_1$ registers of ${\brho_\aux^{**}}$ 
 in the standard basis.
The existence of such a set~$S$ follows from
\Cref{eqn:claim-ent-cont}, since otherwise, by the fact that each coordinate of $\vecd'$ is of the form $\veca[i]\cdot\vecd[i]$, where $\veca_i$ is a non-zero vector and $\vecd[1],\ldots,\vecd[k]$ forms a partition of the coordinates of $\vecd$,  it holds that for every $\vecd'$, 
\[
\Pr[\Good(\vecx_0,\vecx_1,\vecd^*)=\vecd']\leq \left(1-\frac{1}{\sqrt{k}}\right)^{k-m}=\negl(\secp),\]
contradicting \Cref{eqn:claim-ent-cont}.


We construct $\BQP$ algorithm $\A$ that on input $(\pk_2,\ldots,\pk_{n+1})$ generates with non-negligible probability 
images $\left\{\vecy^*_i\right\}_{i\in S}$ corresponding to $\{\pk_{i+1}\}_{i\in S}$, preimages $\{(d^*_i,\vecx^*_{i})\}_{i\in S}$, and equations $\{\vecz_i,\vecz_i\cdot (1, \vecx_{i,0}\oplus\vecx_{i,1})\}_{i\in S}$, such that each $\vecz_i\in \Good_{\Invert(\sk_i,\vecy^*_{i})}$ and $(\vecx_{i,0},\vecx_{i,1})=\Invert(\sk_i,\vecy_i)$, thus breaking the adaptive hardcore bit property.  Algorithm $\A(\pk_2,\ldots,\pk_{n+1})$ does the following: 
    \begin{enumerate}
        \item Sample $(\pk_0,\sk_0),(\pk_1,\sk_1)\gets \Gen(1^\secp)$.
        \item Let $\pk=(\pk_0,\pk_1,\pk_2,\ldots,\pk_{n+1})$.
        \item Generate $(\vecy,\brho)\gets\cC^*.\Commit(\pk,\bsigma)$.   
         \item Parse $\vecy=\left(\vecy_1,\ldots,\vecy_\ell\right)$ and parse $\vecy_{j_{\alpha}}=\left(\vecy^*_{0},\vecy^*_{1},\ldots,\vecy^*_{n+1}\right)$.
        \item Compute $((0,\vecx_0),(1,\vecx_1))\gets \Invert\left(\sk_0,\vecy^*_{0}\right)$.
        
        \item Use $\sk_1$ to generate 
        $\brho'_I=\Pi_{P_{Z_I},\vecm_I}(\brho)$.
        \item Use $\sk_1$ to compute         $\brho_{\vecs_{[\alpha-1]}}=\prod_{i=1}^{\alpha-1} P^{s_i}_{Z_{j_i}}[\brho'_I]$.
        \item Use $\sk_1$ to compute $\brho^{*}_{\vecs_{[\alpha-1]}}=\frac12\brho_{\vecs_{[\alpha-1]}}+\frac12 P_{Z_{\alpha}}[\brho_{\vecs_{[\alpha-1]}}]$.
        
        \item Measure registers $\cO_{j_\alpha}$ of $\brho^{*}_{\vecs_{[\alpha-1]}}$ in the standard basis to obtain  $\vecd^*=(d^*_0,d^*_1,\ldots,d^*_n)$.
        
      \item For every $i\in S$ measure the registers corresponding to $\vecx^*_i$.  We note that with probability $1-\delta$ it holds that for every $i\in S$, 
        \[\vecy^*_{i+1}=\Eval(\pk_{i+1},d^*_i,\vecx^*_{i}).
        \]
        
        Denote the resulting state by $\brho'$.

         
         \item Compute $(\vecz,\brho'')\gets \cC^*.\Open(\pk,(j_{\alpha},0),\brho')$.

        \item Parse $\vecz=(\vecz_1,\ldots,\vecz_{n+1})\in\{0,1\}^{(n+1)^2}$.
        \item Compute $(0,\vecx'_{1,0}),(1,\vecx'_{1,1}))\gets\Invert\left(\sk_1,\vecy^*_{1}\right)$.
        
        \item Let $\gamma=\vecz_1\cdot(1,\vecx'_{1,0}\oplus \vecx'_{1,1})$.
        \item Output $\left\{\pk_{i+1},\vecy^*_{i+1},d^*_{i},\vecx^*_{i},  \vecz_{i+1}, \vecx_{\gamma,i}\right\}_{i\in S}$.
    \end{enumerate}
    We next argue that with non-negligible probability it holds that for every $i\in S$:
    \[
  \vecy^*_{i+1}=\Eval(\pk_{i+1},(d^*_{i},\vecx^*_{i}))~~\wedge~~ \vecz_{i+1}\cdot(1,\vecx'_{i,0}\oplus \vecx'_{i,1})=\vecx_{\gamma,i}~~\wedge~~\vecz_{i+1}\in \Good_{\vecx'_{i,0},\vecx'_{i,1}},
    \]
    where $((0,\vecx'_{i,0}),(1,\vecx'_{i,1}))=\Invert(\sk_{i+1},\vecy^*_{i+1})$,
    contradicting the adaptive hardcore bit property.
    
The fact that the underlying $\cC^*.\Open$ algorithm is accepted by $\Ver$ with probability $\geq 1-\delta$ implies that 
\[
\Pr[\forall i\in S:~\vecy^*_{i+1}=\Eval(\pk_{i+1},(d^*_{i},\vecx^*_{i}))]\geq 1-\delta.
\]
Moreover, if we did not measure $\{(d^*_i,\vecx^*_i)\}_{i\in[n]}$ it would also imply that,
\[
\Pr[\forall i\in S:~\vecz_{i+1}\cdot(1,\vecx'_{i,0}\oplus \vecx'_{i,1})=\vecx_{\gamma,i}~~\wedge~~\vecz_{i+1}\in \Good_{\vecx'_{i,0},\vecx'_{i,1}}]\geq 1-\delta
\]
By the collapsing property applied to $\{\pk_i\}_{i\in\{2,\ldots,n+1\}\setminus \{i+1:~ i\in S\}}$, even if we measured  $\{(d^*_i,\vecx^*_i)\}_{i\in[n]\setminus S}$,
\[
\Pr[\forall i\in S:~\vecz_{i+1}\cdot(1,\vecx'_{i,0}\oplus \vecx'_{i,1})=\vecx_{\gamma,i}~~\wedge~~\vecz_i\in \Good_{\vecx'_{i,0},\vecx'_{i,1}}]\geq 1-2\delta-\negl(\secp)
\]
We note however, that by \Cref{eqn:disturb}, for every $i\in S$ measuring $(d^*_i,\vecx^*_i)$ disturbs the state by at most $k^{-1/4}$.
Hence 
\begin{align*}
&\Pr[\forall i\in S:~\vecy^*_{i+1}=\Eval(\pk_{i+1},(\vecd^*_{i},\vecx^*_{i}))~~\wedge~~\vecz_{i+1}\cdot(1,\vecx'_{i,0}\oplus \vecx'_{i,1})=\vecx_{\gamma, i}~~\wedge~~\vecz_{i+1}\in \Good_{\vecx'_{i,0},\vecx'_{i,1}}]\geq\\
&1-2\delta -\frac{k^{0.1}}{k^{1/4}}-\negl(\secp),
\end{align*}
contradicting the adaptive hardcore bit property.

 \qed

\qed

\begin{remark}
We next argue that relying on some form of the adaptive hardcore bit is necessary.  Specifically, if there exists a $\BQP$ algorithm $\A$ and a non-negligible function $\epsilon=\epsilon(\secp)$ such that 
\[
\Pr[\A(\pk_1,\ldots,\pk_{n+1})=(b_i, \vecx_i,\vecd_i,m_i)_{i=1}^{n+1}~:~\forall i\in[n+1]~~ \vecd_i\cdot(1,\vecx_{i,0}\oplus \vecx_{i,1})=m_i]\geq \epsilon(\secp)
\]
 where $(\vecx_{i,0},\vecx_{i,1})=\Invert(\sk_i,\Eval(\pk_i,(b_i,\vecx_i))$,
then one can use this adversary~$\A$ to attack the scheme, as follows: 
\begin{enumerate}
    \item Given $(\pk_0,\pk_1,\ldots,\pk_{n+1})$, compute 
    \[(b_i, \vecx_i,\vecd_i,m_i)_{i=1}^{n+1}=\A(\pk_1,\ldots,\pk_{n+1}).
    \]
    \item Let $\vecm=(m_1,\ldots,m_{n+1})$ and set $\vecy_0=\Eval(\pk_0,\vecm)$.
    \item For every $i\in[n+1]$, set $\vecy_i=\Eval(\pk_i,(b_i,\vecx_i))$.
    \item Output $(\vecy_0,\vecy_1,\ldots,\vecy_{n+1})$ as the commitment.
    \item If asked to open in the Hadamard basis output $((b_1,\vecx_1),\ldots,(b_{n+1},\vecx_{n+1}))$.
    \item If asked to open in the Standard basis output $(\vecd_1,\ldots,\vecd_{n+1})$.
\end{enumerate} 
The adversary is accepted with probability $\epsilon$ and the openings are distinguishable from a qubit, since the standard basis opening is deterministic and the Hadamard opening is biased, both in a detectable way. 
\end{remark}

\subsection{Binding for the Succinct Commitment Scheme}

In this section we prove the following theorem.
\begin{theorem}\label{def:SCQ-succinct-soundness-theorem}
   The succinct multi-qubit commitment scheme described in \Cref{sec:succinct-multi-qubit-com} satisfies the binding condition defined in \Cref{def:binding:succ}.
\end{theorem}

\begin{proof}
To prove soundness we need to prove that \Cref{eqn:binding1:succinct,eqn:binding2:succinct} hold. To this end, fix any $\BQP$ algorithm $\cC^*.\Commit$, a quantum state $\bsigma$, a polynomial $\ell=\ell(\secp)$, a $\BQP$ prover $P^*$ for $\Ver.\Commit$.
We start by defining a $\QPT$ algorithm $\cC^{*}_\mathsf{ss}.\Commit$ for the underlying semi-succinct commitment scheme.  This algorithm is associated with a parameter $\epsilon_0$, and sometime to be explicit, we denote it by $\cC^{*}_\mathsf{ss}.\Commit_{\epsilon_0}$. It  takes as input $(\pk_1,\bsigma)$ and commits to an  $\ell$-qubit state, as follows:
   \begin{enumerate}
   \item Sample $\hk\gets\Gen_{\mathsf{H}}(1^\secp)$.
   \item Set $\pk=(\pk_1,\hk)$.
       \item Compute $(\rt,\brho)\gets\cC^*.\Commit(\pk,\bsigma)$.
       \item Use the state-preservation extractor~$\cE$ (from \Cref{def:state-preserving-aok}) for the $\NP$ language $\cL^*$ (defined in \Cref{eqn:def:L*}) to generate
    \[(\mathbb{T}_\Sim,\vecy,\brho_{\mathsf{post},\Sim})\gets \cE^{P^*,\brho}\left((\hk,\rt),1^\secp,\epsilon_0\right)\]
       \item If $\Eval(\hk,\vecy)\neq\rt$ then output $\bot$. 
       \item Else, output $(\vecy,\brho_{\mathsf{post},\Sim})$.
   \end{enumerate}
    By \Cref{def:state-preserving-aok,cor:LMS},
    \begin{equation}\label{eqn:vecy-success}
        [\Pr[\Eval_\mathsf{H}(\hk,\vecy)=\rt]\geq 1-\delta_0-2\epsilon_0-\negl(\secp).
    \end{equation}

We use this $\QPT$ algorithm $\cC^*_\mathsf{ss}.\Commit_{\epsilon_0}$  to prove the soundness of the succinct scheme.  We start with proving \Cref{eqn:binding2:succinct}.
To this end, fix a subset $J\subseteq[\ell]$ and a basis $\vecb_J\in\{0,1\}^{|J|}$, and two $\BQP$ algorithms $\cC^*_1.\Open$ and $\cC^*_2.\Open$.  We need to prove that  
\begin{equation}\label{eqn:cont:succ:real}
\Real^{\cC^*.\Commit,P^*,\cC^{*}_1.\Open}\left(\secp,(J,\vecb_J),\bsigma\right)
 \stackrel{\eta}{\approx}\Real^{\cC^*.\Commit,P^*,\cC^{*}_2.\Open}\left(\secp,(J,\vecb_J),\bsigma\right)
\end{equation}
where 
$\eta=O\left(\sqrt{\delta_0+\delta}\right)$, where $\delta_0$ is defined in \Cref{eqn:delta0} and $\delta$ is defined in \Cref{eqn:delta:succinct}.

 
We next define two $\BQP$ opening algorithms $\cC^{*}_{\mathsf{ss},1}.\Open$ and $\cC^{*}_{\mathsf{ss},2}.\Open$ for the underlying semi-succinct commitment scheme, corresponding to   $\cC^{*}_1.\Open$ and $\cC^{*}_2.\Open$, respectively.
For every $i\in\{1,2\}$, the opening algorithm $\cC^{*}_{\mathsf{ss},i}.\Open(\brho_{\mathsf{post},\Sim}, \vecb_J)$ does the following:\footnote{We assume without loss of generality that the state $\brho_{\mathsf{post},\Sim}$ includes $\pk$ and $\vecy$.}
\begin{enumerate}
    \item Run $(\vecy_{i,J},\veco_{i},\vecz_{i,\Sim},\brho'_\mathsf{i,Sim})\gets \cC^*_i.\Open(\brho_{\mathsf{post},\Sim},(J,\vecb_J))$.
    \item If $\Ver_\mathsf{H}(\hk,\rt,J,\vecy_{i,J},\veco_{i})=0$ then output $\bot$.

    \item Else, output $(\vecz_{i,\Sim},\brho'_\mathsf{i,Sim})$.
\end{enumerate}
By \Cref{def:state-preserving-aok},
  \begin{equation}\label{eqn:T-vs-T_Sim}
        (\mathbb{T}_\Sim,\brho_{\mathsf{post},\Sim},\sk)\stackrel{\epsilon_0}\approx (\mathbb{T},\brho_{\mathsf{post}},\sk)
    \end{equation}
which implies that for every $i\in\{1,2\}$,
\begin{equation}\label{eqn:z-vs-z_sim}
\Pr[\Ver(\sk,\rt,(J,\vecb_J),\vecy_{i,J},\veco_{i},\vecz_{i,\Sim})=0]\leq \delta+\epsilon_0.
\end{equation}
For every $j\in J$ let  
\[\vecm_{i,\Sim,j}=\Out_1(\sk_1,\vecy_j,b_j,\vecz_{i,\Sim,j})~~\mbox{ and }~~\vecm_{i,j}=\Out_1(\sk_1,\vecy_j,b_j,\vecz_{i,j}),
\]
where $(\vecz_i,\brho'_\mathsf{i})\gets \cC^*_i.\Open(\brho_\mathsf{post},(J,\vecb_J))$.
Let $\vecm_{i,\Sim,J}=(\vecm_{i,\Sim,j})_{j\in J}$ and $\vecm_{i,J}=(\vecm_{i,j})_{j\in J}$.  \Cref{eqn:T-vs-T_Sim} implies that for every $i\in\{1,2\}$,
\begin{equation}\label{eqn:z-vs-z_sim'}
(\pk,\rt,(J,\vecb_J),\vecm_{i,J})\stackrel{\epsilon_0}\approx (\pk,\rt,(J,\vecb_J),\vecm_{i,\Sim,J}).
\end{equation}
By the binding of the underlying semi-succinct commitment scheme,
\begin{equation}\label{eqn:real**1-real**2}
\Real^{\cC^{*}_\mathsf{ss}.\Commit,\cC^{*}_{\mathsf{ss},1}.\Open}(\secp, \vecb_J,\sigma)
 \stackrel{\eta^*}{{\approx}}\Real^{\cC^{*}_\mathsf{ss}.\Commit,\cC^{*}_{\mathsf{ss},2}.\Open}(\secp, \vecb_J,\sigma).
 \end{equation}
 where $\eta^*=O(\sqrt{\delta^*})$  is defined in \Cref{def:binding} and 
  \begin{equation}\label{delta'}
    \delta^*= \E_{\substack{(\pk_1,\sk_1) \leftarrow \Gen_\mathsf{1}(1^\lambda) \\ (\vecy, \brho) \leftarrow \cC^{*}_\mathsf{ss}.\Commit(\pk_1, \bsigma)}}\max_{\substack{i\in\{1,2\},\\ \vecb'\in\{\vecb_{|J|},{\bf 0},{\bf 1\}}}}\Pr[\Ver_\mathsf{ss}(\sk,\vecy,\vecb',\cC^{*}_{\mathsf{ss},i}.\Open(\brho,\vecb')=0].
    \end{equation}
 We note that $\delta^*\leq \delta_0+3\epsilon_0+\delta+\negl(\secp)$.  This follows from \Cref{eqn:vecy-success,{eqn:z-vs-z_sim}}, together with the collision resistance property of the underlying hash family.

We thus conclude that 
\begin{align*}
&\Real^{\cC^*.\Commit,P^*,\cC^{*}_1.\Open}(\secp,(J,\vecb_J),\bsigma)=\\
&(\pk,\rt,(J,\vecb_J),\vecm_{1,J})\stackrel{\epsilon_0}\approx \\
&(\pk,\rt,(J,\vecb_J),\vecm_{1,\Sim,J})\stackrel{\eta^*}\approx \\
&(\pk,\rt,(J,\vecb_J),\vecm_{2,\Sim,J})\stackrel{\epsilon_0}\approx \\ 
&(\pk,\rt,(J,\vecb_J),\vecm_{2,J})=\\
&\Real^{\cC^*.\Commit,P^*,\cC^{*}_2.\Open}(\secp,(J,\vecb_J),\bsigma),
\end{align*}
where the second and forth equations follow from \Cref{eqn:z-vs-z_sim'} and the third equation follows \Cref{{eqn:real**1-real**2}}. Setting $\epsilon_0=\delta_0$ we conclude that \Cref{eqn:binding2:succinct} holds.\\

 It remains to prove \Cref{eqn:binding1:succinct}.  To this end, we use the $\BQP$ extractor $\Ext_\mathsf{ss}$ corresponding to the underlying semi-succinct commitment scheme, as well as the extractor $\cE$ corresponding to the underlying state-preserving argument-of-knowledge system, to construct the extractor $\Ext$ for the succinct commitment scheme.  $\Ext^{P^*, P^*_\Test}(\sk,\rt,\brho,1^{\lceil1/\epsilon\rceil})$   does the following:
 \begin{enumerate}
 \item\label{item:def-C} Let $C\in \mathbb{N}$ be a constant that is larger than the constant from the definition of $\eta$ in \Cref{eqn:binding1} and in \Cref{eqn:binding2:succinct}.  Namely, $C$ is chosen so that in \Cref{eqn:binding1} $\eta\leq C\cdot \sqrt{\delta}$, and in \Cref{eqn:binding2:succinct} $\eta\leq C\cdot \sqrt{\delta_0+\delta}$.
 \item Set $\epsilon_0=\left(\frac{\epsilon}{8C}\right)^2$.
   \item Use the state-preservation extractor~$\cE$ (from \Cref{def:state-preserving-aok}) for the $\NP$ language $\cL^*$ (defined in \Cref{eqn:def:L*}) to generate
    \begin{equation}\label{eqn:dist-y}
        (\mathbb{T}_\Sim,\vecy,\brho_{\mathsf{post},\Sim})\gets \cE^{P^*,\brho}\left((\hk,\rt),1^\secp,\epsilon_0\right)
    \end{equation}
    
    \item If $\Eval_{\mathsf{H}}(\hk,\vecy)\neq \rt$ then set $\vecy=\bot$.
   \item \label{item:C**} Use $P^*_\Test$ to define $\cC^{*}_\mathsf{ss}.\Open$, which is associated with a parameter $\epsilon_0$, and on input $(\brho_{\mathsf{post},\Sim},(j,b))$,\footnote{We assume without loss of generality that the state $\brho_{\mathsf{post},\Sim}$ includes $(\pk,\vecy,\rt)$.} operates as follows:\footnote{ $\cC^{*}_\mathsf{ss}.\Open$ is defined somewhat analogously to $\cC^*.\Open_{[\ell]}$ as defined in the proof of \Cref{lemma:non-succinct-real-to-ideal,lemma:semi-succinct-real-to-ideal}.}

    \begin{enumerate}

  \item 
   Denote by $U_{b}$ the unitary that does the following computation (coherently, using  ancilla registers): 
    \begin{enumerate}
        \item Compute the first message of  $P^*_\Test(\pk,\rt,\brho_{\mathsf{post},\Sim})$ upon receiving the bit $b\in\{0,1\}$ from $V_\Test$, to obtain $\rt'$ and a post state $\brho'_\Sim$.
    
   \item  Use the state-preservation extractor~$\cE$ (from \Cref{def:state-preserving-aok}) for the $\NP$ language $\cL^*$ (defined in \Cref{eqn:def:L*}) to generate
    \begin{equation}\label{eqn:ext-z}
    (\mathbb{T}_\Sim,\vecz,\brho'_{\mathsf{post},\Sim})\gets \cE^{P^*_\Test,\brho'_{\Sim}}\left((\hk,\rt'),1^\secp,\epsilon_0\right)
    \end{equation}
    \item Denote the ancila registers where $\vecz$ is stored by $(\mathsf{open}_1,\ldots,\mathsf{open}_\ell)$.
    \item If $b=0$ (corresponding to a standard basis measurement), apply a post-processing unitary to each $\mathsf{open}_i$ register, to ensure that measuring this register would not disturb the state in a detectable way. This is done as in \Cref{remark:C*ell}. 
    
    Specifically, Denoting $\sk_1=(\sk_{1,0},\sk_{1,1},\ldots,\sk_{1,n+1})$, the unitary $U_0$ uses $(\sk_{1,1},\ldots,\sk_{1,n+1})$ to apply the following post-processing unitary to each $\mathsf{open}_i$ register, to ensure that when measured the disturbance will not be noticed to a $\QPT$ algorithm which is not given $(\sk_{1,1},\ldots,\sk_{1,n+1})$. Recall that $\mathsf{open}_i$ contains a vector $\vecz=(\vecz_1,\ldots,\vecz_{n+1})\in\{0,1\}^{(n+1)^2}$ where each $\vecz_j\in\{0,1\}^{n+1}$. The post-processing unitary does the following:
        \begin{enumerate}
           \item Coherently compute for every $j\in[n+1]$ the bit $m_j=\vecz_j\cdot (1,\vecx'_{j,0}\oplus \vecx'_{j,1})$, where $\vecx'_{j,0}$ and $\vecx'_{j,1}$ are the two preimages of $\vecy_{i,j}$ that are computed using $\sk_j$.
            \item Let $\vecm=(m_1,\ldots,m_{n+1})\in\{0,1\}^{n+1}$.
                        Note that if $\vecz$ is a successful opening (i.e., it is accepted) then $\vecm$ is a preimage of $\vecy_{i,0}$, and whether a preimage is measured or not is undetectable without knowing $\sk_0$, due to the collapsing property of the underlying $\TCF$ family.
            \item On an ancila register, compute a super-position over all $\vecz'=(\vecz'_1,\ldots,\vecz'_{n+1})\in\{0,1\}^{(n+1)^2}$ such that for every $j\in[n+1]$ $m_j=\vecz'_j\cdot (1,\vecx'_{j,0}\oplus \vecx'_{j,1})$.  
            \item Swap register $\mathsf{open}_i$ with the ancila register above, so that now $\vecz'=(\vecz'_1,\ldots,\vecz'_{n+1})$ is in register $\mathsf{open}_i$.
        \end{enumerate}
    \end{enumerate}
    \item Compute $ \brho'=U^\dagger_{b} \CNOT_{\mathsf{open}_j,\mathsf{copy}_j} U_{b}[\brho_{\mathsf{post},\Sim}]$
    \item Measure register $\mathsf{copy}_j$ in the standard basis to obtain $\vecz_j$.
    
   \item Output $\vecz_j$.
    \end{enumerate}
   So far, we defined $\cC^*_\mathsf{ss}.\Open$ on a single coordinate $(j,b)$.  We define $\cC^*_\mathsf{ss}.\Open$ on a set of coordinates $(J,\vecb_J)$ to first apply $\cC^*_\mathsf{ss}.\Open$ on all the coordinates $j\in J$ such that $\vecb_j=0$ (in order) and then apply it on all the coordinates $j\in J$ such that $\vecb_j=1$ (in order).\footnote{This ordering is done for simplicity, as it allows us to rely on the analysis of $\cC^*.\Open_{[\ell]}$ in the proof of \Cref{lemma:non-succinct-real-to-ideal,lemma:semi-succinct-real-to-ideal}.  In particular, we do not need to rely on the fact that measuring the Hadamard basis opening is not detectable when opening, and verifying the opening, in the standard basis.} 
    \item Output $\btau_{\cA,\cB}\gets \Ext_\mathsf{ss}^{\cC^{*}_\mathsf{ss}.\Open }(\sk,\vecy,\brho_{\mathsf{post},\Sim})$.
      \end{enumerate}

     We need to argue that for every $\BQP$ algorithm $\cC^{*}.\Open$,
      \begin{equation} \label{eqn:Real-ideal-succ}
          \Real^{\cC^*.\Commit,P^*,\cC^{*}.\Open}(\secp,(J,\vecb_J),\bsigma)\stackrel{\zeta}\approx\Ideal^{\Ext,\cC^*.\Commit,P^*,P^*_\Test}(\secp, (J,\vecb_J),\bsigma,\epsilon),
     \end{equation}
     for 
     \[\zeta=O\left(\sqrt{\delta_0+\delta'_0+\delta}\right)+\epsilon.
     \]
To this end, we rely on the binding property of the underlying semi-succinct scheme (and in particular \Cref{eqn:binding1}), which implies that
\begin{equation}\label{eqn:ss-real-ideal}
\Real^{\cC^{*}_\mathsf{ss}.\Commit,\cC^{*}_\mathsf{ss}.\Open}(\secp,(J,\vecb_J),\bsigma)\stackrel{\eta^*}\approx\Ideal^{\Ext_{\mathsf{ss}},\cC^{*}_\mathsf{ss}.\Commit,\cC^*_\mathsf{ss}.\Open}(\secp, (J,\vecb_J),\bsigma)
\end{equation}
where $\eta^*\leq C\cdot \sqrt{\delta^*}$ 
    and
    \begin{equation}\label{eqn:delta-ss}
    \delta^*=\E_{\substack{(\pk,\sk) \leftarrow \Gen(1^\lambda) \\ (\vecy, \brho) \leftarrow \cC^*_\mathsf{ss}.\Commit(\pk, \bsigma)}}\max_{\vecb'\in\{\vecb_J,{\bf 0}^{|J|},{\bf 1}^{|J|}\}}\Pr[\Ver_\mathsf{ss}(\sk,\vecy,(J,\vecb'),\cC^*_\mathsf{ss}.\Open(\brho,\vecb'))=0].
    \end{equation}
By the definition of $\cC^*_\mathsf{ss}.\Open$, and as explained in \Cref{remark:C*ell} (and similarly to \Cref{eqn:delta[ell]}), 
\begin{equation}\label{eqn:delta*}
\delta^*\leq \epsilon^*_0+\epsilon^*_1+\negl(\secp)
\end{equation}
where
\[
\epsilon^*_b=\Pr[\Ver_\mathsf{ss}(\sk_1,\vecy,(J,b^{|J|}),\vecz_b)=0],
\]
and where $\vecy$ is distributed as in \Cref{eqn:dist-y}, and  $\vecz_b$ is distributed as in \Cref{eqn:ext-z} when computed coherently by $U_b$. 
We next argue that 
\begin{equation}\label{eqn:bound-eps}
    \epsilon^*_0+\epsilon^*_1\leq \Pr[\vecy=\bot]+\Pr[\vecz_0=\bot]+\Pr[\vecz_1=\bot]+2\delta'_{0}+8\epsilon_0+\negl(\secp).
\end{equation}
The reason \Cref{eqn:bound-eps} holds is that after extracting $\vecy$ the residual state is $\epsilon_0$-indistinguishable from the state obtained without extraction.  After further extracting $\vecz_b$ the residual state is $2\epsilon_0$-indistinguishable from the state obtained without extraction. By \Cref{def:state-preserving-aok,cor:LMS}, this implies that the probability that in the third argument-of-knowledge, the extractor outputs a valid witness $(\vecy,\vecz_b)$, corresponding to the instance $(\sk_1,\hk,\rt,\rt',b)$, is at most $2\epsilon_0+\delta'_{0,b}+2\epsilon_0=\delta'_{0,b}+4\epsilon_0$, up to negligible factors, where $\delta'_{0,b}$ is the probability that $P^*_\Test$ is rejected given that the fist message sent by $V_\Test$ is~$b\in\{0,1\}$. This, together with the collision resistance property of the underlying hash family, and with the fact that $\delta'_{0,0}+\delta'_{0,1}=2\delta'_0$, implies that  \Cref{eqn:bound-eps} indeed holds.

Note that
\begin{equation}\label{eqn:ybot}
\Pr[\vecy=\bot]\leq \delta_0+2\epsilon_0+\negl(\secp)
\end{equation}
This follows from the following calculation:
\begin{align*}
&\Pr[\vecy=\bot]= \\
&\Pr[\vecy=\bot~\wedge~\mathbb{T}_\Sim~\mbox{ is rejecting}]+\Pr[\vecy=\bot~\wedge~\mathbb{T}_\Sim~\mbox{ is accepting}]\leq \\
&\Pr[\mathbb{T}_\Sim~\mbox{ is rejecting}]+\Pr[\vecy=\bot~\wedge~\mathbb{T}_\Sim~\wedge~\mbox{ is accepting}]\leq \\
&\delta_0+\epsilon_0+\epsilon_0+\negl(\secp)
\end{align*}
where the latter equation follows from the definition of $\delta_0$ and from \Cref{def:state-preserving-aok}.  Similarly, \begin{equation}\label{eqn:zbot}
\Pr[\vecz_b=\bot]\leq \delta'_{0,b}+3\epsilon_0+\negl(\secp)
\end{equation}
This follows from the following calculation:
\begin{align*}
&\Pr[\vecz_b=\bot]= \\
&\Pr[\vecz_b=\bot~\wedge~\mathbb{T}_\Sim~\mbox{ is rejecting}]+\Pr[\vecz_b=\bot~\wedge~\mathbb{T}_\Sim~\mbox{ is accepting}]\leq \\
&\Pr[\mathbb{T}_\Sim~\mbox{ is rejecting}]+\Pr[\vecz_b=\bot~\wedge~\mathbb{T}_\Sim\mbox{ is accepting}]\leq \\
&\delta'_{0,b}+2\epsilon_0+\epsilon_0+\negl(\secp)
\end{align*}
This, together with \Cref{eqn:delta*,eqn:bound-eps}, implies that 
\[
\delta^*\leq  (\delta_0+2\epsilon_0)+ (2\delta'_0+6\epsilon_0)+ 2\delta'_0+8\epsilon_0=\delta_0+ 4\delta'_0+16\epsilon_0.
\]
We conclude that
\[
\eta^*\leq O(\sqrt{\delta_0+\delta'_0})+ C\cdot \sqrt{16\epsilon_0}\leq O(\sqrt{\delta_0+\delta'_0})+\frac{\epsilon}{2}.
\]

%
In order to use \Cref{eqn:ss-real-ideal}, with $\eta^*$ as above, we define 
    \[
\mathsf{Succ}\mbox{-}\Ideal^{\Ext_{\mathsf{ss}},\cC^{*}_\mathsf{ss}.\Commit_{\epsilon_0},\cC^*_\mathsf{ss}.\Open}(\secp, (J,\vecb_J),\bsigma)
    \]
    to be the distribution obtained by sampling
  \[
(\pk_1,\vecy,(J,\vecb_J),\vecm)\gets \Ideal^{\Ext_{\mathsf{ss}},\cC^{*}_\mathsf{ss}.\Commit_{\epsilon_0},\cC^*_\mathsf{ss}.\Open}(\secp, (J,\vecb_J),\bsigma),
     \]
    sampling $\hk\gets \Gen_\mathsf{H}(1^\secp)$, computing $\rt=\Eval_\mathsf{H}(\hk,\vecy)$, and outputing
    \[
    ((\pk_1,\hk),\rt,(J,\vecb_J),\vecm_J).
    \]
    Similarly, we define
    \[
\mathsf{Succ}\mbox{-}\Real^{\cC^{*}_\mathsf{ss}.\Commit,\cC^{*}_\mathsf{ss}.\Open}(\secp,(J,\vecb_J),\bsigma)
    \]
  to be the distribution obtained by sampling  \[
(\pk_1,\vecy,(J,\vecb_J),\vecm)\gets \Real^{\Ext_{\mathsf{ss}},\cC^{*}_\mathsf{ss}.\Commit_{\epsilon_0},\cC^*_\mathsf{ss}.\Open}(\secp, (J,\vecb_J),\bsigma),
     \]
     sampling $\hk\gets \Gen_\mathsf{H}(1^\secp)$, computing $\rt=\Eval_\mathsf{H}(\hk,\vecy)$, and outputting
     \[
      ((\pk_1,\hk),\rt,(J,\vecb_J),\vecm_J).
     \]
Note that by the definition of the extractor $\Ext$ it holds that
\[
\Ideal^{\Ext,\cC^*.\Commit,P^*,P^*_\Test}(\secp, (J,\vecb_J),\bsigma,\epsilon)\equiv\mathsf{Succ}\mbox{-}\Ideal^{\Ext_{\mathsf{ss}},\cC^{*}_\mathsf{ss}.\Commit_{\epsilon_0},\cC^*_\mathsf{ss}.\Open}(\secp, (J,\vecb_J),\bsigma).
\]
This is the case since 
\[\Ideal^{\Ext,\cC^*.\Commit,P^*,P^*_\Test}(\secp, (J,\vecb_J),\bsigma,\epsilon)
\]
extracts the state
\[\btau_{\cA,\cB}\gets \Ext_\mathsf{ss}^{\cC^{*}_\mathsf{ss}.\Open }(\sk,\vecy,\brho_{\mathsf{post},\Sim}),
\]
where $(\vecy,\brho_{\mathsf{post},\Sim})\gets \cC^*_\mathsf{ss}.\Commit_{\epsilon_0}(\pk_1,\bsigma)$.\\

Therefore, to prove \Cref{eqn:Real-ideal-succ} it suffices to prove that 
 \begin{equation} \label{eqn:Real-ideal-succ2}
          \Real^{\cC^*.\Commit,P^*,\cC^{*}.\Open}(\secp,(J,\vecb_J),\bsigma)\stackrel{\zeta^*}\approx\mathsf{Succ}\mbox{-}\Real^{\cC^{*}_\mathsf{ss}.\Commit,\cC^{*}_\mathsf{ss}.\Open}(\secp,(J,\vecb_J),\bsigma),
     \end{equation}
where $\zeta^*=O(\sqrt{\delta_0+\delta'_o+\delta})+\frac{\epsilon}{2}$.\\

To this end,  we use $P^*_\Test$ to define a $\QPT$ algorithm  $\cC^{**}.\Open$.  We mention that $\cC^{**}.\Open$ bears similarity to $\cC^*_\mathsf{ss}.\Open$ (defined in \Cref{item:C**} of the definition of $\Ext$), with the difference being that the latter was defined for the semi-succinct commitment, whereas $\cC^{**}.\Open$  is defined for the succinct commitment. In particular,  recall that for the succinct commitment, an opening to the $j$'th qubit consists of a tuple $(\vecy_j,\veco_j,\vecz_j)$.  $\cC^{**}.\Open$ uses $P^*_\Test$ to generate this opening as follows:
      \begin{enumerate}
      \item Use the state-preservation extractor~$\cE$ (from \Cref{def:state-preserving-aok}) for the $\NP$ language $\cL^*$ (defined in \Cref{eqn:def:L*}) to generate
     \[(\mathbb{T}_\Sim,\vecy,\brho_{\mathsf{post},\Sim})\gets \cE^{P^*_\Test,\brho}\left((\hk,\rt),1^\secp,\epsilon_0\right).\]
     Let $(\vecy_j,\veco_j)=\Open_{\mathsf{H}}(\hk,\vecy,j)$.
      \item   Let $U_{b}$ be the unitary as defined in the definition of the extractor $\Ext$ above.  
   
      \item Compute $ \brho'=U^\dagger_{b} \CNOT_{\mathsf{open}_j,\mathsf{copy}_j} U_{b}[\brho_{\mathsf{post}\Sim}]$.


      \item Measure register $\mathsf{copy}_j$ in the standard basis to obtain $\vecz_j$.
       \item Output $(\vecy_j,\veco_j,\vecz_j)$.
\end{enumerate}
      
    \begin{proposition}\label{claim:Real-ideal-succ}
      \[\Real^{\cC^*.\Commit,P^*,\cC^{*}.\Open}(\secp,(J,\vecb_J),\bsigma)\stackrel{\zeta_1}\approx \Real^{\cC^*.\Commit,P^*,\cC^{**}.\Open}(\secp,(J,\vecb_J),\bsigma)
      \]
      where $\zeta_1=O\left( \sqrt{\delta_0+\delta'_0+\delta}\right)+\frac{\epsilon}{2}$.
      \end{proposition}
We note that \Cref{claim:Real-ideal-succ} completes the proof of \Cref{eqn:Real-ideal-succ2} since by the definition of $\cC^{**}.\Open$ and $\cC^{*}_\mathsf{ss}.\Open$
\[
\Real^{\cC^*.\Commit,P^*,\cC^{**}.\Open}(\secp,(J,\vecb_J),\bsigma)\equiv\mathsf{Succ}\mbox{-}\Real^{\cC^{*}_\mathsf{ss}.\Commit_{\epsilon_0},\cC^{*}_\mathsf{ss}.\Open}(\secp,(J,\vecb_J),\bsigma).
\]

\paragraph{Proof of \Cref{claim:Real-ideal-succ}.}
      \Cref{eqn:binding2:succinct} (which we proved above) implies that it suffices to prove the following:
\begin{equation}\label{eqn:lem:succ:pf}
\Pr[\Ver(\sk,\rt,(J,b_J),(\vecy_J,\veco_J,\vecz_J))=0]\leq O({\delta_0+\delta'_0})+\left(\frac{\epsilon}{2C}\right)^2+\negl(\secp)
\end{equation}
where $(\vecy_J,\veco_J,\vecz_J)=\cC^{**}.\Open(\brho_\mathsf{post},(J,\vecb_J))$ and where $C\in\mathbb{N}$ is defined in \Cref{item:def-C} of the definition of $\Ext$. 
We first note that by \Cref{def:state-preserving-aok,cor:LMS},
\begin{equation}\label{eqn:lem:succ:pf1}
\Pr[\Ver_\mathsf{H}(\hk,\rt,J,\vecy_J,\veco_J)=0]\leq \delta_0+2\epsilon_0+\negl(\secp).\end{equation}
Moreover, the residual state, denoted by $\brho_{\mathsf{post},\Sim}$ satisfies that 
\begin{equation}\label{eqn:rho-sim}
(\brho_{\mathsf{post},\Sim},\sk)\stackrel{\epsilon_0}\approx(\brho_{\mathsf{post}},\sk)
\end{equation}
which implies that $P^*_\Test(\pk,\rt,\brho_{\mathsf{post},\Sim})$, upon receiving $b\in\{0,1\}$ from $V_\Test$  is rejected with probability at most $\delta'_{0,b}+\epsilon_0+\negl(\secp)$. By \Cref{def:state-preserving-aok,cor:LMS}, this implies that the tuple
$(\mathbb{T}_\Sim,\vecz_b,\brho'_{\mathsf{post},\Sim})$
generated in \Cref{eqn:ext-z} satisfies that 
\[
\Pr[\Eval_\mathsf{H}(\hk,\vecz_b)\neq \rt'_b]\leq \delta'_{0,b}+3\epsilon_0+\negl(\secp).
\]
By the union bound, we conclude that for every $b\in\{0,1\}$,
\begin{equation}\label{eqn:yveez}
\Pr[\Eval_{\mathsf{H}}(\hk,\vecy)\neq \rt~~\vee~~\Eval_\mathsf{H}(\hk,\vecz_b)\neq \rt'_b]\leq \delta_0+\delta'_{0,b}+5\epsilon_0+\negl(\secp).
\end{equation}
By \Cref{def:state-preserving-aok}, for every $b\in\{0,1\}$ it holds that the state $\brho'_{\mathsf{post},\Sim}$, generated in \Cref{eqn:ext-z} as part of $U_b$, is $\epsilon_0$-indistinguishable from the state of $P^*_\Test(\pk,\rt,\brho_{\mathsf{post},\Sim},b)$ after executing the first state-preserving argument-of-knowledge.  This, together with \Cref{eqn:rho-sim}, implies that the state $\brho'_{\mathsf{post},\Sim}$ is $2\epsilon_0$-indistinguishable from the state of $P^*_\Test(\pk,\rt,\brho_{\mathsf{post}})$  after executing the first state-preserving argument-of-knowledge.  Since $P^*_\Test(\pk,\rt,\brho_{\mathsf{post}},b)$ is accepted in both its state-preserving argument-of-knowledge protocols with probability at least $1-\delta'_{0,b}$, it holds that it is accepted in the second state-preserving argument-of-knowledge protocol (w.r.t.\ the language $\cL^{**}$) when it starts with the state $\brho'_{\mathsf{post},\Sim}$ with probability at least $1-\delta'_{0,b}-2\epsilon_0$.
This, together with \Cref{def:state-preserving-aok,,cor:LMS}, implies that 
\[
\Pr[((\sk_1,\hk,\rt,\rt',b),(\vecy,\vecz))\in \cR_{\cL^{**}}]
\geq  1-\delta'_{0,b}-4\epsilon_0,
\]
which together with \Cref{eqn:yveez} and the collision resistant property of the underlying hash family implies that 
\[\Pr[\Ver(\sk,\rt,(J,b_J),(\vecy_J,\veco_J,\vecz_J))=0]\leq O(\delta_0+\delta'_0)+9\epsilon_0 +\negl(\secp)
\]
Thus it remains to note that 
$9\epsilon_0\leq \left(\frac{\epsilon}{2C}\right)^2$, as desired.

\end{proof}




\section{Applications} \label{sec:applications}

\subsection{Succinct Interactive Arguments for $\QMA$}\label{sec:succint-QMA} 

In this section we construct a succinct interactive argument for $\QMA$. To this end, we construct a {\em semi-succinct} interactive argument for $\QMA$, where only the verifier's messages are short but the messages from the prover may be long.
We then rely on a black-box transformation from \cite{Bartusek22} which shows a generic transformation for converting any post-quantum computationally sound semi-succinct interactive argument for $\QMA$ into a fully succinct one.

We refer to \cite[Definition 3.1]{Bartusek22} for a formal definition of a succinct argument system. That definition requires the soundness error to be negligible; when this is not the case for our protocols, we say so explicitly. 

\paragraph{Ingredients}

Our semi-succinct interactive argument consists of the following three ingredients:
\begin{itemize}
 \item A pseudorandom generator $\mathsf{PRG}:\{0,1\}^\secp\rightarrow \{0,1\}^\ell$, where $\ell=\ell(\secp)$ is a polynomial specified in \Cref{lmm:FHM} below.    \item A semi-succinct (qubit-by-qubit) commitment scheme $(\Gen, \Commit, \Open, \Ver, \Out)$, as defined in \Cref{sec:XZ-commitments} and constructed in \Cref{sec:constructions}.
  
    \item The information-theoretic QMA verification protocol of Fitzsimons, Hajdušek, and  Morimae~\cite{FHM18}. As in \cite{Bartusek22}, we use an ``instance-independent'' version due to~\cite{TCC:ACGH20} and assume the soundness gap is $1-\negl(\secp)$, where the latter can can be achieved by standard $\QMA$ amplification.
\begin{lemma}[\cite{FHM18,TCC:ACGH20,Bartusek22}]\label{lmm:FHM}
For all languages $\cL = (\cL_{\mathsf{yes}}, \cL_{\mathsf{no}}) \in \QMA$ there exists a polynomial $k(\secp)$, a function $\ell(\secp)$ that is polynomial in the time $T(\secp)$ required to verify instances of size $\secp$, a $\QPT$ algorithm $P_\mathsf{FHM}$, and a $\PPT$ algorithm $V_\mathsf{FHM}$ such that the following holds.
\begin{itemize}
    \item $P_\mathsf{FHM}(\vecx,\ket{\psi}) \to \ket{\pi}$: on input an instance $\vecx \in \{0,1\}^\secp$ and a quantum state $\ket{\psi}$, $P_\mathsf{FHM}$ outputs an $\ell(\secp)$-qubit state $\ket{\pi}$.
    \item \textbf{Completeness.} For all $\vecx\in\cL_{\mathsf{yes}}$ and $\ket{\phi}\in\cR_\cL(x)$ it holds that for a random $\vech \gets \{0,1\}^{\ell(\secp)}$
    $$
    \Pr[V_\mathsf{FHM}(\vecx, \vecv) = \acc : \ket{\pi} \gets P_\mathsf{FHM}\left(\vecx, \ket{\phi}^{\otimes k(\secp)}\right)] \geq 1 - \negl(\secp)
    $$
    where $\vecv$ is the result of measuring $\ket{\pi}$ in basis $\vech$.
    \item \textbf{Soundness.} For all $\vecx\in\cL_{\mathsf{no}}$ and all $\ell$-qubit states $\ket{\pi^*}$ it holds that for a  random $\vech \gets \{0,1\}^{\ell(\secp)}$,
    $$
    \Pr[V_\mathsf{FHM}(\vecx,\vecv^*) = \acc] \leq \negl(\secp)
    $$
    where $\vecv^*$ is the result of measuring $\ket{\pi^*}$ in basis $\vech$.
   
\end{itemize}

\end{lemma}

\end{itemize}

\paragraph{The semi-succinct interactive argument for $\QMA$}  In the following protocol $P$ and $V$ are given an instance $\vecx$ and $P$ is given $k$ copies of the $\QMA$ witness $\ket{\psi}$.

\begin{description}
    \item $V\rightarrow P$:  Generate $(\pk,\sk)\gets \Gen(1^\secp)$, and send $\pk$.
    \item $P\rightarrow V$: 
    \begin{enumerate}
        \item Compute $\ket{\bpi}=P_{\mathsf{FHM}}\left(\vecx,\ket{\bpsi}^{\tensor{k}}\right)$.
        \item Compute $(\vecy,\brho)\gets \Commit(\pk, \ket{\bpi})$.

        Denote by $\ell$ the number of qubits in $\ket{\bpi}$, and denote by $\vecy=(\vecy_1,\ldots,\vecy_\ell)$, where $\vecy_i$ is a commitment to the $i$'th qubit of $\bsigma$.
        \item Send $\vecy$.
    \end{enumerate}
    \item $V\rightarrow P$: 
   Send a random bit $b\in\{0,1\}$.
   \item If $b=0$:\footnote{This should be thought of as a ``test round.''}  
   \begin{enumerate}
       \item $V\rightarrow P$: Send a random bit $h\gets\{0,1\}$.  
       \item $P\rightarrow V$: Send $\vecz\gets \Open(\brho,h^\ell)$.
       \item $V\rightarrow P$: Compute $v=\Ver(\sk,\vecy,h^\ell,\vecz)$ and accept if $v=1$ and otherwise, reject. 
       
   \end{enumerate} 

    \item If $b=1$:  
    \begin{enumerate}
    \item $V\rightarrow P$: Send a random seed $\vecs\gets \{0,1\}^\secp$. 
    \item $P\rightarrow V$: Compute $\vecb=\mathsf{PRG}(\vecs)\in\{0,1\}^\ell$ and send the openings $(\vecz_1,\ldots,\vecz_\ell)\gets \Open(\brho,\vecb)$.
    \item $V$ does the following:  
    \begin{enumerate}
    \item Compute $\vecb=\mathsf{PRG}(\vecs)$.
        \item For every $i\in [\ell]$ compute $u_i=\Ver(\sk,\vecy_i,b_i,\vecz_i)$ and $v_i=\Out(\sk,\vecy_i,b_i,\vecz_i)$.
        \item If there exists $i\in[\ell]$ such that $u_i=0$ then reject.
        \item Else, accept if and only if $V_{\mathsf{FHM}}$ would accept $(\vecx,(b_1,\ldots,b_\ell),(v_1,\ldots,v_\ell))$.
    \end{enumerate}
    
    \end{enumerate}
\end{description}

\begin{theorem}\label{them:semi-succinct-QMA}
    The above scheme is a computationally sound semi-succinct interactive argument for $\QMA$, with completeness $1- \negl(\lambda)$ and soundness $1 - 1/\lambda^2$.
\end{theorem}

\begin{proof}
The completeness property  is straightforward and hence we focus on proving the soundness property, which will follow from the binding property of the commitment. Fix a $\QMA$ promise problem $\cL=(\cL_\mathsf{yes},\cL_\mathsf{no})$. Fix $P^*$, an input~$\vecx^*$ and an auxiliary state $\bsigma$, such that $P^*(\vecx^*,\bsigma)$ is accepted with probability $1-\delta$, for $\delta\leq \frac{1}{\secp^2}$.  We argue that it must be the case that $\vecx^*\notin \cL_\mathsf{no}$. To this end, we use $P^*$ to construct $P^*_\mathsf{FHM}$ that is accepted with high probability in the protocol $(P_\mathsf{FHM},V_\mathsf{FHM})$ on input $\vecx^*$.  The algorithm $P^*_\mathsf{FHM}(\vecx^*,\bsigma)$ proceeds as follows:
\begin{enumerate}
    \item Generate $(\pk,\sk)\gets \Gen(1^\secp)$.
    \item Generate $(\vecy,\brho)\gets P^*(\pk,\vecx^*,\bsigma)$.
    \item Use the extractor $\Ext$ from the binding property of the commitment scheme to extract a state $\btau\gets\Ext^{P^*}(\sk,\vecy,\brho)$
\item Send $\btau$.
\end{enumerate}
We next argue that $V_\mathsf{FHM}$ accepts $\btau$ with high probability on a random basis. To this end, it suffices to argue that it accepts $\btau$ with high probability on a  pseudorandom basis, since otherwise one can distinguish a pseudorandom string from a truly random one, thus breaking the underlying $\mathsf{PRG}$.  
Denote by 
\[
\mathsf{Good}=\{\vecs\in \{0,1\}^\secp:~P^* \mbox{ is accepted w.p.}\geq 1-\secp\delta \mbox{  when $V$ sends $\vecs$}\}
\]
Note that 
\begin{equation}\label{eqn:sgood}
    p\triangleq \Pr[\vecs\in \mathsf{Good}]\geq  1-\frac{2}{\secp}
\end{equation}
which follows from the following Markov argument:
\begin{align*}
&1-2\delta\leq\Pr[P^* \mbox{ is accepted}~|~b=1] = \\
&\Pr[P^* \mbox{ is accepted}~|~b=1~\wedge~\vecs\in \mathsf{Good}]\cdot \Pr[\vecs\in \mathsf{Good}~|~b=1]+\\
&\Pr[P^* \mbox{ is accepted}~|~b=1~\wedge~\vecs\notin \mathsf{Good}]\cdot \Pr[\vecs\notin \mathsf{Good}~|~b=1]\leq\\
&p+(1-\secp \delta)(1-p)=\\
&1-\secp\delta(1-p)
\end{align*}
which implies that $-2\delta\leq -\secp\delta(1-p)$ and in turn that $\secp(1-p)\leq 2$, thus implying \Cref{eqn:sgood}.
By the binding property of the underlying commitment scheme, for any basis $\vecb=\mathsf{PRG}(\vecs)$ such that $\vecs\in\mathsf{Good}$, it holds that 
\[
(\pk,\vecy,\vecb,\vecm_\mathsf{Real})\stackrel{O(\sqrt{\secp\delta})}\approx
(\pk,\vecy,\vecb,\vecm_\mathsf{Ideal})
\]
where $\vecm_\mathsf{Ideal}$ is the result of measuring $\btau\gets\Ext^{P^*}(\sk,\vecy,\brho)$ in basis $\vecb$, and $\vecm_\mathsf{Real}$ is the output corresponding to the opening of $P^*$.  The fact that the measurements 
$\vecm_\mathsf{Real}$ are accepted by $V_\mathsf{FHM}$ with probability $\geq 1-\secp\delta$ (for any basis $\mathsf{PRG}(\vecs)$ such that $\vecs\in\mathsf{Good}$) implies that $\vecm_\mathsf{Ideal}$ is accepted by $V_\mathsf{FHM}$ with probability $\geq 1-\secp\delta-O(\sqrt{\secp\delta})$  (for any such basis).  This, together with \Cref{eqn:sgood} implies that $\btau$ is accepted by $V_\mathsf{FHM}$ on a pseudorandom basis with probability
\begin{align*}
&\Pr[V_\mathsf{FHM} \mbox{ accepts $\btau$ on basis }\mathsf{PRG}(\vecs)]\geq \\
&\Pr[V_\mathsf{FHM} \mbox{ accepts $\btau$ on basis }\mathsf{PRG}(\vecs)~|~\vecs\in\mathsf{Good}]\cdot \Pr[\vecs\in\mathsf{Good}]\geq\\
&\Pr[V_\mathsf{FHM} \mbox{ accepts $\btau$ on basis }\mathsf{PRG}(\vecs)~|~\vecs\in\mathsf{Good}]\cdot \left(1-\frac{2}{\secp}\right)\geq\\
&\left(1-\secp\delta -O(\sqrt{\secp\delta})\right)\cdot \left(1-\frac{2}{\secp}\right)
\end{align*}
This, together with \Cref{lmm:FHM} and our assumption that $\delta\leq \frac{1}{\secp^2}$, implies that $\vecx^*\notin \cL_{\mathsf{No}}$, as desired. 

\end{proof}

\begin{theorem}\label{thm:succinct-qma}
    There exists a computationally sound succinct interactive argument for $\QMA$, with completeness $1 - \negl(\lambda)$ and soundness $\negl(\lambda)$.
\end{theorem}

\begin{proof}
   To achieve full succinctness, we apply \cite[Theorem 9.3]{Bartusek22}, which transforms a semi-succinct protocol into a fully succinct one. As input, this theorem expects a semi-succinct argument with soundness $\negl(\lambda)$, whereas the argument system given by \Cref{them:semi-succinct-QMA} only has soundness $1 - 1/\lambda^2$. However, it is easy to see that this can be made $\negl(\lambda)$ by sequentially repeating the protocol $\poly(\lambda)$ times. This does not harm succinctness since it blows up the communication by a factor of $\poly(\lambda)$, which is permitted. (Alternatively, following the techniques of \cite{Alagic_2020} as used in \cite{Bartusek22}, parallel repetition can be employed to avoid an increase in round complexity.)
\end{proof}

\subsection{Succinct Interactive Arguments from $X/Z$ Quantum PCPs}\label{sec:app:QPCP}

In this section we show how to convert any $X/Z$ quantum PCP for a language $\cL$ into an succinct interactive argument $(P,V)$ for $\cL$.  As in \Cref{sec:succint-QMA} we construct a semi-succinct interactive argument, and then use the black-box transformation from \cite{Bartusek22} to convert it into a fully succinct one. 

\paragraph{Ingredients}
Our semi-succinct interactive argument consists of the following ingredients.

\begin{itemize}
   \item A semi-succinct (qubit-by-qubit) commitment scheme $(\Gen, \Commit, \Open, \Ver, \Out)$, as defined in \Cref{sec:XZ-commitments} and constructed in \Cref{sec:constructions}.
   \item An $X/Z$ quantum PCP for the language $\cL$, with verifier $V_\mathsf{QPCP}$.
\end{itemize}

\paragraph{The semi-succinct interactive argument for $\cL$}  In the following protocol $(P,V)$ are given an instance $\vecx$ and $P$ is also given an $X/Z$ quantum PCP $\ket{\pi}$.

\begin{description}
    \item $V\rightarrow P$:  Generate $(\pk,\sk)\gets \Gen(1^\secp)$, and send $\pk$.
    \item $P\rightarrow V$: Compute $(\vecy,\brho)\gets \Commit(\pk, \ket{\bpi})$ and send $\vecy$

        Denote by $\ell$ the number of qubits in $\ket{\bpi}$, and denote by $\vecy=(\vecy_1,\ldots,\vecy_\ell)$, where $\vecy_i$ is a commitment to the $i$'th qubit of $\bsigma$.

    \item $V\rightarrow P$: 
   Send a random bit $b\in\{0,1\}$.
   \item If $b=0$:\footnote{This should be thought of as a ``test round.''}  
   \begin{enumerate}
       \item $V\rightarrow P$: Send a random bit $h\gets\{0,1\}$.  
       \item $P\rightarrow V$: Send $\vecz\gets \Open(\brho,h^\ell)$.
       \item $V\rightarrow P$: Compute $v=\Ver(\sk,\vecy,h^\ell,\vecz)$ and accept if $v=1$ and otherwise, reject.

   \end{enumerate} 
   
    \item If $b=1$:  
    \begin{enumerate}
    \item $V\rightarrow P$: Send a sample $(i_1,\ldots,i_c,b_1,\ldots,b_c)\gets V_{\mathsf{QPCP}}(\vecx,1^\secp)$.
    \item $P\rightarrow V$: Send the openings $(\vecz_1,\ldots,\vecz_c)\gets \Open(\brho,(i_1,b_1),\ldots,(i_c,b_c))$.
    \item $V$ does the following:  
    \begin{enumerate}
    
        \item For every $j\in [c]$ compute $v_j=\Ver(\sk,\vecy_{i_j},b_j,\vecz_j)$ and $u_j=\Out(\sk,\vecy_{i_j},b_j,\vecz_j)$.
        \item If there exists $j\in[c]$ such that $v_j=0$ then reject.
        \item Else, accept if and only if $V_{\mathsf{QPCP}}$ would accept $(\vecx,(i_1,\ldots,i_c), (b_1,\ldots,b_\ell),(u_1,\ldots,u_c))$.
    \end{enumerate}
    
    \end{enumerate}
\end{description}

\begin{theorem}\label{them:semi-succinct-QPCP}
    The above scheme is a semi-succinct interactive argument for $\cL$ with completeness $1 -\negl(\lambda)$ and soundness $1 - 1/\lambda^2$.
\end{theorem}

\paragraph{Proof of \Cref{them:semi-succinct-QPCP}.}
The completeness property is straightforward and hence we focus on proving the binding property. Fix a $\BQP$ cheating prover~$P^*$, an input~$\vecx^*$ and an auxiliary state $\bsigma$, such that $P^*(\vecx^*,\bsigma)$ is accepted with probability $1-\delta$, for $\delta\leq \frac{1}{\secp^2}$.  We use $P^*$ to extract an $X/Z$ quantum PCP~$\bpi$ for $x^*\in \cL$ that is accepted with high probability, thus implying that indeed $x^*\in \cL$ as desired.  This is done as follows:
\begin{enumerate}
    \item Generate $(\pk,\sk)\gets \Gen(1^\secp)$.
    \item Generate $(\vecy,\brho)\gets P^*(\pk,\vecx^*,\bsigma)$.
    \item Use the extractor $\Ext$ from the binding property of the commitment scheme to extract a state $\bpi\gets\Ext^{P^*}(\sk,\vecy,\brho)$
\item Output $\bpi$.
\end{enumerate}
The fact that $P^*$ is accepted with probability $1-\delta$ implies that for every $h\in\{0,1\}$ it opens in an accepted way on $h^\ell$ with probability at least $1-4\delta$.  
Denote by $\mathsf{Good}$ the event that $V_{\mathsf{QPCP}}$ samples $(i_1,\ldots,i_c,b_1,\ldots,b_c)$ such that $P^*$ is accepted with probability $\geq 1-\secp\delta$ when $V$ sends $(i_1,\ldots,i_c,b_1,\ldots,b_c)$.
Note that 
\begin{equation}\label{eqn:sgood:PCP}
    p\triangleq \Pr[\mathsf{Good}]\geq  1-\frac{2}{\secp}
\end{equation}
which follows from the following Markov argument:
\begin{align*}
&1-2\delta\leq \Pr[P^* \mbox{ is accepted}~|~b=1] = \\
&\Pr[P^* \mbox{ is accepted}~|~b=1~\wedge~\mathsf{Good}]\cdot \Pr[\mathsf{Good}~|~b=1]+\\
&\Pr[P^* \mbox{ is accepted}~|~b=1~\wedge~\neg \mathsf{Good}]\cdot \Pr[\neg\mathsf{Good}~|~b=1]\leq\\
&p+(1-\secp \delta)(1-p)=\\
&1-\secp\delta(1-p)
\end{align*}
which implies that $-2\delta\leq -\secp\delta(1-p)$ and in turn that $\secp(1-p)\leq 2$, thus implying \Cref{eqn:sgood:PCP}. In what follows we say that $(i_1,\ldots,i_c,b_1,\ldots,b_c)\in\mathsf{Good}$ if $P^*$ is accepted when $V$ sends $(i_1,\ldots,i_c,b_1,\ldots,b_c)$ with probability $\geq 1-\secp \delta$.
By the binding property of the underlying commitment scheme, for any $(i_1,\ldots,i_c,b_1,\ldots,b_c)\in\mathsf{Good}$, it holds that 
\[
(\pk,\vecy,(i_1,\ldots,i_c,b_1,\ldots,b_c),\vecm_\mathsf{Real})\stackrel{O(\sqrt{\secp\delta})}\approx
(\pk,\vecy,(i_1,\ldots,i_c,b_1,\ldots,b_c),\vecm_\mathsf{Ideal})
\]
where $\vecm_\mathsf{Ideal}$ is the result of measuring $\bpi\gets\Ext^{P^*} (\sk,\vecy,\brho)$ in locations $(i_1,\ldots,i_c)$ and basis $(b_1,\ldots,b_c)$, and $\vecm_\mathsf{Real}$ is the output corresponding to the opening of $P^*$.  The fact that the measurements 
$\vecm_\mathsf{Real}$ are accepted by $V_\mathsf{QPCP}$ with probability $\geq 1-\secp\delta$ (for any $(i_1,\ldots,i_c,b_1,\ldots,b_c)\in\mathsf{Good}$) implies that $\vecm_\mathsf{Ideal}$ is accepted by $V_\mathsf{QPCP}$ with probability $\geq 1-\secp\delta-O(\sqrt{\secp\delta})$  (for any such basis).  This, together with \Cref{eqn:sgood:PCP} implies that 
\begin{align*}
&\Pr[V_\mathsf{QPCP} \mbox{ accepts $\bpi$}]\geq \\
&\Pr[V_\mathsf{QPCP} \mbox{ accepts $\bpi$}~|~\mathsf{Good}]\cdot \Pr[\mathsf{Good}]\geq\\
&\Pr[V_\mathsf{QPCP} \mbox{ accepts $\bpi$ }~|~\mathsf{Good}]\cdot \left(1-\frac{2}{\secp}\right)\geq\\
&\left(1-\secp\delta -O(\sqrt{\secp\delta})\right)\cdot \left(1-\frac{2}{\secp}\right)
\end{align*}
This, together with our assumption that $\delta\leq \frac{1}{\secp^2}$, implies that indeed $\bpi$ is an $X/Z$ PCP that is accepted with high probability, and thus $\vecx\in \cL$, as desired.

\begin{theorem}\label{thm:succinct-QPCP}
    There exists a succinct interactive argument for $\cL$ with completeness $1 -\negl(\lambda)$ and soundness $\negl(\lambda)$.
\end{theorem}
\begin{proof}
    We apply \cite[Theorem 9.3]{Bartusek22} to the sequential repetition of the argument system from \Cref{them:semi-succinct-QPCP}, exactly as in the proof of \Cref{thm:succinct-qma}.
\end{proof}

\qed

\ifsubmission
\else
\ifanon
\else 

\section*{Acknowledgements}
This work was done in part while SG, AN, and AV were participants at the Simons Institute 2023 Summer Cluster on Quantum Computing, and we thank the Simons Institute and the organizers for the opportunity. Yael Kalai is supported by DARPA under Agreement No. HR00112020023.  Any opinions, findings and conclusions or recommendations expressed in this material are those of the author(s) and do not necessarily reflect the views of the United States Government or DARPA. Anand Natarajan is supported by NSF CAREER Grant CCF-2339948. Agi Villanyi acknowledges support by the Doc Bedard fellowship from the Laboratory for Physical Sciences through the Center for Quantum Engineering and the National Science Foundation Graduate Research Fellowship under Grant No. 2141064. Sam Gunn is supported by a Google PhD Fellowship and the U.S. Department of Energy, Office of Science, National Quantum Information Science Research Centers, Quantum Systems Accelerator. 

\fi
\fi

\newpage

\bibliographystyle{alpha}
\bibliography{crypto.bib,quantum_bib}


\section{Weak commitments to Quantum States (WCQ)}\label{def:WCQ-section}

In this section, we recall the Measurement Protocol from Mahadev \cite{Mah18a}, which was formalized by \cite{Bartusek22}.  We refer to such a protocol as a  \textit{weak commitment to quantum states} (WCQ) protocol, and define it formally below.

\begin{definition}[Weak Commitment to Quantum States (WCQ)]\label{def:WCQ-syntax} An $\ell$-qubit WCQ protocol is specified by the five algorithms $(\GenM,\CommitM,\OpenM,\TestM,\OutM)$:
\begin{enumerate}
    \item $\GenM$ is a $\ppt$ algorithm that takes as input the security parameter $\secp$ (in unary) and a string  $h \in \{0,1\}^\ell$, and outputs  a pair $(\pk,\sk) \gets  \GenM(1^\secp,h)$, where $\pk$ is referred to as the {\em public key} and $\sk$ is referred to as the {\em secret key}.
    \item $\CommitM$ is a $\BQP$ algorithm that takes as input a public key $\pk$ and a quantum state~$\bsigma$ and outputs a pair $(\vecy,\brho) \gets \CommitM(\pk,\bsigma)$, where $\vecy$ is a classical string, referred to as the {\em commitment string},  and $\brho$ is a quantum state.
    \item $\OpenM$ is a $\BQP$ algorithm that takes as input a bit $c\in\{0,1\}$ and a quantum state $\brho$ and outputs a classical string  $z \gets \OpenM(\brho,c)$, referred to as the  {\em opening string}.
    \item $\TestM$ is a polynomial time algorithm that takes as input a public key $\pk$ and a pair $(\vecy,z)$, where $\vecy$ is a commitment string and $z$ is an opening string, and it outputs $\{\acc,\rej\} \gets \TestM(\pk,(\vecy,z))$.
    \item $\OutM$ is a polynomial time algorithm that takes as input a secret key $\sk$ and a pair $(\vecy,z)$, where $\vecy$ is a commitment string and $z$ is an opening string, and it outputs a classical string $m\in\{0,1\}^\ell$.
\end{enumerate}

The commitment protocol associated with the tuple $(\GenM,\CommitM,\OpenM,\TestM,\OutM)$ is a two party protocol between a $\BQP$ committer $\cC$ which takes as input a quantum state $\bsigma$, and a $\BPP$ verifier $\cV$ which takes as input a classical string $h\in\{0,1\}^\ell$. Both parties also take as input the unary security parameter~$\secp$. The protocol consists of two phases $\mathsf{COMMIT}$ and $\mathsf{OPEN}$, proceeding as follows:
\begin{itemize}
    \item $\mathsf{COMMIT}$ phase:
    \begin{enumerate} 
    \item $[\cC \leftarrow \cV]$: $\cV$ samples $(\pk,\sk) \gets  \GenM(1^\secp,h)$ and sends the public key $\pk$ to $\cC$.
    \item $[\cC \rightarrow \cV]$: $\cC$ computes $(\vecy,\brho) \gets \CommitM(\pk,\bsigma)$ and sends the commitment string~$\vecy$ to the verifier.
    \end{enumerate}
    \item $\mathsf{OPEN}$ phase:
    \begin{enumerate}
    \item $[\cC \leftarrow \cV]$: $\cV$ samples a random challenge bit $c \gets \{0,1\}$ and sends $\ct$ to $\cC$. 
    \item $[\cC \rightarrow \cV]$: $\cC$ sends $z \gets \OpenM(\brho,c)$ to $\cV$. 
    \item If $c = 0$, $\cV$ outputs $\{\acc,\rej\} \gets \TestM(\pk,y,z)$. If $c = 1$, $\cV$ outputs $m \gets \OutM(\sk,y,z)$. 
    \end{enumerate}
\end{itemize}
\end{definition}

A WCQ protocol acts over registers $\P, \Y, \Z, \W$ where $\P$ contains the public component of the output of $\GenM$, $\Y$ contains the output of $\CommitM$, $\Z$ contains the output of $\OpenM$, and  $\W$ are additional work registers. Additionally, the commitment protocol satisfies the following properties for \textit{correctness} and \textit{binding}. 
\begin{definition}[WCQ correctness]\label{def:WCQ-correctness}
Let $\RealW(1^\lambda, \bsigma, h)$ be the distribution resulting from running $(\pk, \sk) \leftarrow \GenM(1^\lambda, h)$, $(\vecy, \brho) \leftarrow \CommitM(\pk, \bsigma)$, $z \leftarrow \OpenM(\brho, 1)$, and outputting $m \leftarrow \OutM(\sk, y, z)$. Let $\bsigma(h)$ denote the distribution resulting from measuring each qubit $i$ of a quantum state $\bsigma$ in the basis specified by $h_i$ for $i \in [\ell]$. A WCQ protocol is correct if, for all $\ell$-qubit quantum states $\bsigma$ and for every $h \in \{0,1\}^\ell$, the following two properties are satisfied:  


\begin{enumerate}
    \item (Test Round Completeness):
       \begin{equation}
        \mathrm{Pr}\left[\!\begin{aligned}
        & (\pk, \sk) \leftarrow \GenM(1^\lambda, h); \\
        \acc \leftarrow \TestM(\pk, y, z): \quad & (\vecy, \brho) \leftarrow \CommitM(\pk, \bsigma); \\
        & z \leftarrow \OpenM(\brho, 0)] \\
        \end{aligned}\right]
        = 1 - \negl({\lambda})
        \end{equation}
     \item (Measurement Round Completeness):
        \begin{equation}
        \left\{\!
        \begin{aligned}
        & (\pk, \sk) \leftarrow \GenM(1^\lambda, h); \\
        m \leftarrow \OutM(\sk, y, z): \quad & (\vecy, \brho) \leftarrow \CommitM(\pk, \bsigma); \\
        & z \leftarrow \OpenM(\brho, 1)\\
        \end{aligned}\right\}
        \approx_c \bsigma(h)
        \end{equation}
        \end{enumerate}
\end{definition}


\begin{definition}[WCQ Binding]\cite{Bartusek22}\label{def:WCQ-binding} A WCQ protocol is \text{\em{binding}} if there exists a $\PPT$ classical algorithm $\SimGen$ and a $\QPT$ oracle machine $\MExt$ such that, for any cheating $\BQP$ committer $\cC^*$ with quantum state $\bsigma$ that satisfies that for every $h\in\{0,1\}^\ell$:  
 \begin{equation}
    \mathrm{Pr}\left[\!\begin{aligned}
    & (\pk, \sk) \leftarrow \GenM(1^\lambda, h); \\
    \acc \leftarrow \TestM(\pk, y, z): \quad & (\vecy, \brho) \leftarrow \cC^*.\CommitM(\pk, \bsigma); \\
    & z \leftarrow \cC^*.\OpenM(\brho, 0)] \\
    \end{aligned}\right]
    = 1 - \negl({\lambda}),
 \end{equation}
 it holds that for every $h\in\{0,1\}^\ell$,
$$\Sim^{\cC^*}(1^\lambda,h)\approx_c \RealW^{\cC^*}(1^\lambda, h)$$ 
where

\begin{itemize}
    \item $\Sim^{\cC^*}(1^\lambda,h)$  is the output distribution of the following procedure: 
    \begin{enumerate}
    \item Sample $(\pk,\sk)\gets \SimGen(1^\secp)$.
    \item Execute the commitment round to obtain $(\vecy, \brho) \leftarrow \cC^*.\CommitM(\bsigma)$.
       \item Execute $\btau\gets \MExt^{\cC^*}(\pk,\sk,y,\brho)$.
        \item Measure $\btau$ in the basis specified by $h$, where $h_i = 0$ corresponds to the standard basis and $h_i = 1$ corresponds to the Hadamard, and output these measurement values.
    \end{enumerate}
    \item $\RealW^{\cC^*}(1^\lambda, h)$ 
    is the output distribution of the following procedure:
    \begin{enumerate}
        \item Sample $(\pk, \sk) \leftarrow  \GenM(1^\lambda, h)$.
        \item Emulate the commitment round to obtain $(\vecy, \brho) \leftarrow \cC^*.\CommitM(\bsigma)$.
        \item Emulate the opening phase round corresponding to $c=1$ to obtain  $z \leftarrow \cC^*.\OpenM(\brho, 1)$.
        \item Compute $m \leftarrow \OutM(\sk, y, z)$ and output $m$.
    \end{enumerate}
\end{itemize}
\end{definition}
\section{A $\TCF$ construction with the distributional strong adaptive hardcore bit property}\label{app:zvik}

We argue that the $\TCF$ from \cite{BCMVV18}, defined  below, satisfies the distributional strong adaptive hardcore bit property (under $\LWE$).  

\paragraph{The $\TCF$ family from  \cite{BCMVV18}}\label{sec:TCF-construction}

The $\TCF$ family from \cite{BCMVV18} is a lattice based construction and makes use of the following theorem from \cite{MicciancioP11}.

\begin{theorem}[Theorem 5.1 in \cite{MicciancioP11}]\label{thm:MP11}
Let $n, m \geq 1$ and $q \geq 2$ be such that $m = \Omega(n \log q)$. There is
an efficient randomized algorithm $\mathsf{TrapGen}_\mathsf{MP}(1^n,1^m,q)$ that returns a matrix $\vecA\in \mathbb{Z}_q^{m\times n}$ together with a trapdoor $\vect_\vecA\in\mathbb{Z}_q^m$ such that the distribution of $\vecA$ is negligibly (in $n$) close to the uniform distribution. Moreover, there is
an efficient algorithm $\Invert_\mathsf{MP}$ that, on input $(\vecA,\vect, \vecA\cdot \vecs + \vece)$, where $\norm{\vece} \leq  \frac{q}{C\sqrt{n\log q}}$ and where $C$ is a universal constant, returns $\vecs$ and $\vece$ with overwhelming probability over $(\vecA, \vect_\vecA)\gets\mathsf{TrapGen}_\mathsf{MP}(1^n,1^m,q)$. 
\end{theorem}

The $\TCF$ family $(\Gen,\Eval,\Invert,\Check,\Good)$ from \cite{BCMVV18} is defined as follows:

    \begin{itemize}
        \item $\Gen(1^\secp)$ is associated with the following:
        \begin{itemize}
        \item Prime $q\leq 2^\secp$ of size super-polynomial in $\secp$.
        \item $n=n(\secp)$ and $m=m(\secp)$, both polynomially bounded functions of $\secp$, such that $m\geq n\log q$ and $n\geq \secp$.
        \item Two error distributions $\chi,\chi'$ over $\mathbb{Z}_q$, that are associated with bounds $B, B'\in \mathbb{N}$ such that:
        \begin{enumerate}
        \item  $\frac{B}{B'}=\negl(\secp)$.
        \item $B'\leq \frac{q}{2C\sqrt{n\cdot m\cdot \log q}}$, where $C$ is the  universal constant from \Cref{thm:MP11}.
        \item  $\Pr_{\vece\gets \chi^m}[\norm{\vece}> B]=\negl(\secp)$.
        \item $\Pr_{\vece'\gets (\chi')^m}[\norm{\vece'}> B']=\negl(\secp)$. 
            \item  $\vece'\equiv \vece'+\vece$, for $\vece\gets\chi^m$ and $\vece'\gets(\chi')^m$.
          
        \end{enumerate}
        \end{itemize}
     It does the following:
            \begin{enumerate}
            \item Generate $(\vecA,\vect_\vecA)\gets \mathsf{TrapGen}_\MP(1^n,1^m,q)$.
            
            \item Choose a random bit string $\vecs\gets \{0,1\}^n$ and a random error vector $\vece\gets\chi^m$.
            \item Let $\vecu=\vecA\cdot \vecs+\vece$.
            \item Output $\pk=(\vecA,\vecu)$ and $\sk=(\vecA,\vecu,\vect_\vecA)$.
        \end{enumerate}
       
        \item  $\Eval(\pk,b,\cdot)$ is a function with domain $\mathbb{Z}_q^n$ and range $\mathbb{Z}_q^m$.
$\Eval(\pk,b,\vecx)$ parses $\pk=(\vecA,\vecu)$, samples $\vece'\gets(\chi')^m$, and outputs $\vecy=\vecA\vecx+b\vecu+\vece'$ (where all the operations are done modulo~$q$).

Equivalently, we think of $\Eval(\pk,b,\cdot)$ as a function with domain $\{0,1\}^{w}$ for $w\triangleq n\cdot\lceil{\log q}\rceil$, where each element $\vecx\in\mathbb{Z}_q^n$ is matched to its bit decomposition.
Namely, denote by \[J:\mathbb{Z}^n_q\rightarrow\{0,1\}^{w}\] 
the bit decomposition function where each element in $\mathbb{Z}_q$ is converted to its bit decomposition in $\{0,1\}^{\lceil{\log q}\rceil}$. We think $\Eval$ as taking as input an element $\vecz\in\{0,1\}^w$, computing $\vecx=J^{-1}(\vecz)\in\mathbb{Z}_q^n$,\footnote{ $J^{-1}:\{0,1\}^w\rightarrow \mathbb{Z}_q^n$ is the function that breaks its input into $n$ blocks of length $\lceil{\log q}\rceil$ each, and replaces each such block $(b_1,\ldots,b_{\lceil{\log q}\rceil})$ with the element $(\sum_{i=1}^{\lceil{\log q}\rceil}b_i\cdot 2^{i-1})\mod q$, which is an element in $\mathbb{Z}_q$.} and then applying $\Eval(\pk,b,\vecx)$.

\begin{remark}
  We note the change in notation:  In the definition of a $\TCF$ family, we denoted the input length by $n$, and here we denote it by~$w$.  
\end{remark}
\item $\Invert(\sk,\vecy)$ does the following:
\begin{enumerate}
\item Parse $\sk=(\vecA,\vecu,\vect_\vecA)$.
    \item Compute $\Invert_\mathsf{MP}(\vecA,\vect_\vecA,\vecy)=\vecx$
    \item If $\norm{\vecy-\vecA\cdot\vecx}\leq 2\sqrt{m}\cdot B'$ then output $((0,J(\vecx)),(1,J(\vecx-\vecs)))$.
\end{enumerate}
\item $\Check(\pk,b,J(\vecx),\vecy)$ outputs $1$ if and only if $\norm{\vecy-\vecA\vecx-b\vecu}\leq 2\sqrt{m}\cdot B'$.

\item $\Good(J(\vecx_0),J(\vecx_1),\vecd)$ outputs $\vecd'\in\{0,1\}^n$ such that 
\[\vecd\cdot (1,J(\vecx_0)\oplus J(\vecx_1))=\vecd'\cdot \vecs
\]
where the inner product in both sides of the equation above is done modulo~$2$.
The vector $\vecd'$ is computed as follows, using the fact that 
$\vecx_1=\vecx_0-\vecs$ (where subtraction is modulo~$q$) and the fact that $\vecs\in\{0,1\}^n$.
\begin{enumerate}
    \item Partition $\vecd\in\{0,1\}^{w+1}$ into its first bit, denoted by $d_0$, and the following $n$ blocks, each of size $\lceil{\log q\rceil}$, denoted by $\vecd[1],\ldots,\vecd[n]\in \{0,1\}^{\lceil{\log q\rceil}}$.
    \item  For every $b\in\{0,1\}$,  partition $J(\vecx_b)$ into blocks $J(\vecx_{b,1}),\ldots, J(\vecx_{b,n})$, each of size $\lceil{\log q\rceil}$. 
    \item For every $i\in [n]$ let 
    \[\vecd_i'=\vecd[i]\cdot(J(\vecx_{0,i}) \oplus J(\vecx_{0,i}-1)),
    \]
    where $\cdot$ denotes inner product mod~$2$ and where $\vecx_{0,i}-1$ is done mod~$q$
    \item Output the string $\vecd'\in\{0,1\}^n$.
\end{enumerate} 
Note that 
\[
\vecd\cdot (1,J(\vecx_0)\oplus J(\vecx_1))= d_0+\vecd'\cdot \vecs\mod 2. 
\]

 \end{itemize} 

\begin{proposition}\label{claim:Bra-is-stat}
 The $\TCF$ family from \cite{BCMVV18} satisfies the distributional strong adaptive hardcore bit property assuming the post-quantum  hardness of $\LWE$.  
\end{proposition}

The proof of  \Cref{claim:Bra-is-stat} makes use of the following lemma from \cite{BCMVV18}.

\begin{lemma}\label{lemma:zviketal}\cite{BCMVV18}
Let $q$ be a prime, $k,n \geq 1$ integers, and $\vecC \gets \mathbb{Z}^{k\times n}_q$
a uniformly random matrix. With probability
at least $1 - q^k\cdot 2^{-\frac{n}{8}}$
over the choice of $\vecC$ the following holds for the fixed  $\vecC$. For all $\vecv \in \mathbb{Z}^k_q$ and any distinct vectors $\vecd'_1,\vecd'_2\in \{0,1\}^n\setminus\{0^n\}$, the distribution of $(\vecd'_1\cdot \vecs,\vecd'_2\cdot \vecs)$, where both inner produces are done$\mod 2$ and where $\vecs\gets \{0,1\}^n$  is uniform 
conditioned on $\vecC\vecs = \vecv$, is within
statistical distance $O(q^{\frac{3k}{2}}\cdot 2^{\frac{-n}{40}})$ of the uniform distribution over $\{0,1\}^2$.
\end{lemma}

\begin{remark}
    We note that \cite{BCMVV18} proved this lemma for a single vector $\vecd'$ (as opposed to two distinct ones).  Their proof carries over to this setting as well, and we include it in \Cref{app:zvik} for completeness.
\end{remark}

\paragraph{Proof of \Cref{claim:Bra-is-stat}.}
We define $\sk_\mathsf{pre}$, corresponding to $\sk=(\vecA,\vecu,\vect_\vecA)$, to be $\sk_\mathsf{pre}=\Invert_\mathsf{MP}(\vecA,\vect_\vecA,\vecu)$. Thus $\sk_\mathsf{pre}=\vecs$, where $\vecu=\vecA\cdot\vecs+\vece$ and $\vece$ is a low norm vector.

 Fix any $\QPT$ circuit $C$ and any $\QPT$ algorithm $\A$ that takes as input~$\pk=(\vecA,\vecA\cdot\vecs+\vece)$ and a quantum state $\bsigma$, and outputs a tuple $(\vecy, b,\vecx, \brho)$, such that $\Check(\pk,b,\vecx,\vecy)=1$  and $\brho$ is a state that has registers $\cO_1$ and $\cO_2$, where $\cO_1$ contains $w+1$ qubits, and with overwhelming probability over $\aux\gets C(\brho_{\cO_2},\vecs)$ it holds that the residual state $\brho_\aux$ satisfies that for every $\vecd'\in\{0,1\}^{n}$ the probability of measuring registers $\cO_1$ of $\brho_\aux$ in the standard basis and obtaining $\vecd\in\{0,1\}^{w+1}$ such that $\Good(\vecd,\vecx_{0},\vecx_{1})=\vecd'$ is negligible in $\secp$.
 
We need to prove that 
\begin{equation}\label{eqn:need-to-prove}
(\pk, \vecy, J(\vecx_{0})\oplus J(\vecx_{1}), \vecd\cdot (1,J(\vecx_{0})\oplus J(\vecx_{1})),\aux) \approx (\pk, \vecy,J(\vecx_{0})\oplus J(\vecx_{1}), U,\aux)
\end{equation}
where $(\pk,\sk)\gets\Gen(1^\secp)$, $(\vecy, b,\vecx, \brho, C)\gets \A(\pk,\bsigma)$, $((0,\vecx_{0}),(1,\vecx_{1}))=\Invert(\sk,\vecy)$, $(\vecd,\aux)$ is obtained by letting $\aux\gets C(\brho_{\cO_2},\vecs)$ and $\vecd$ is the outcome of measuring the $\cO_1$ registers of the residual state $\brho_\aux$ in the standard basis, and $U$ is uniformly distributed in $\{0,1\}$.

To this end, we define an alternative algorithm $\widehat{\Gen}(1^\secp)$, which is the same as $\Gen(1^\secp)$ with the only difference being that rather than choosing $(\vecA,\vect_\vecA)$ via the $\mathsf{TrapGen}$ algorithm, it chooses $\vecA$ to be close to a low rank matrix.  Specifically, $\widehat{\Gen}(1^\secp)$ does the following:
\begin{enumerate}
    \item Let $\delta=\epsilon/2$ and let $k=n^\delta$.
    \item Sample $(\vecB,\vect_\vecB)\gets\mathsf{TrapGen}_\mathsf{MP}(1^k,1^m,q)$.
    \item Sample $\vecC\gets \{0,1\}^{k\times n}$ and $\vecN\gets \chi^{m\times n}$.
    \item Let $\widehat{\vecA}=\vecB\cdot\vecC+\vecN$.
    \item Sample $\vecs\gets \{0,1\}^n$ and $\vece\gets\chi^m$.
    \item Let $\widehat{\vecu}=\widehat{\vecA}\cdot\vecs+\vece$.
    \item Output $\widehat{\pk}=(\widehat{\vecA},\widehat{\vecu})$ and $\widehat{\sk}=(\widehat{\vecA},\vecB,\vect_\vecB,\vecs)$.
\end{enumerate}
The $\LWE$ implies that $\vecA\approx\widehat{\vecA}$ which in turn implies that 
\begin{equation}\label{eqn:AapproxA'}
  (\vecA,\vecu, \vecs, \vecy, b, \vecx, \brho)\approx  (\widehat{\vecA},\widehat{\vecu}, \vecs, \widehat{\vecy}, \widehat{b}, \widehat{\vecx}, {\widehat{\brho}})  
\end{equation}
where 
\begin{itemize}
    \item $\vecs\gets\{0,1\}^n$ and $\vece\gets\chi^m$.
    \item $\vecu=\vecA\cdot\vecs+\vece$ and $\widehat{\vecu}=\widehat{\vecA}\cdot\vecs+\vece$.
    \item $(\vecy,b,\vecx,\brho)=\A(\pk,\bsigma)$ for 
$\pk=(\vecA,\vecu)$.
\item $(\widehat{\vecy},\widehat{b},\widehat{\vecx},{\widehat{\brho}})=\A(\widehat{\pk},\bsigma)$ for 
$\widehat{\pk}=(\widehat{\vecA},\widehat{\vecu})$.
\end{itemize} 
We next argue that 
\begin{equation}\label{eqn:pk'-equiv}
(\widehat{\vecA}, \vecu, \vecy, J(\vecx_{0})\oplus J(\vecx_{1}), \vecd\cdot(1,J(\vecx_{0})\oplus J(\vecx_{1})),\aux) \approx (\widehat{\vecA}, \vecu,\vecy,J(\vecx_{0})\oplus J(\vecx_{1}), U,\aux)
\end{equation}
where in the above equation, to avoid cluttering of notation, we omit some of the ``hat'' notation, and denote by $\widehat{\pk}=(\widehat{\vecA},\vecu)$ distributed according to $\widehat{\Gen}(1^\secp)$,    $(\vecy,b,\vecx,\brho)\gets \A(\widehat{\pk},\bsigma)$ for $\widehat{\pk}=(\widehat{\vecA},\vecu)$, $\vecx_{b}=\vecx$,  $\vecx_{1-b}=\vecx_{b}-(-1)^{b}\vecs$, and $(\vecd,\aux)$ is obtained by measuring computing $\aux\gets C(\brho_{\cO_2},\vecs)$ and $\vecd$ is the outcome of measuring the $\cO_1$ registers of the residual state $\brho_\aux$ in the standard basis.
\Cref{eqn:need-to-prove} follows immediately from \Cref{eqn:pk'-equiv}, together with the fact that $\vecA\approx \widehat{\vecA}$ and the fact that the distributions in \Cref{eqn:pk'-equiv,eqn:need-to-prove} can be generated efficiently from $\vecA$ and $\widehat{\vecA}$, respectively.

Let
\[\vecd'=\Good(\vecx_{0},\vecx_{1},\vecd)\in\{0,1\}^n,
\]
then by the definition of $\Good$,
\[
\vecd\cdot (1,J(\vecx_{0})\oplus J(\vecx_{1}))= d_{0}\oplus\vecd'\cdot \vecs. 
\]
Therefore, to prove \Cref{eqn:pk'-equiv} it suffices to prove that
\[
(\widehat{\vecA}, \vecu, \vecs, \vecy, b, \vecx, d_{0}, \vecd'\cdot\vecs,\aux) \approx (\widehat{\vecA}, \vecu, \vecs, \vecy, b, \vecx, d_{0}, U, \aux).
\]
To prove the above equation it suffices to prove that with overwhelming probability over $\vecC\gets \mathbb{Z}_q^{k\times n}$
it holds that for every $\vecv\in\mathbb{Z}_q^k$, for every distribution $\D_\vecv$ (that depends on $\vecv$) that outputs $(\vecd',\brho_{\cO_2})$ such that with overwhelming probability $\vecd'$ has min-entropy $\omega(\log \secp)$ given $\aux\gets C(\brho_{\cO_2},\vecs)$,
\begin{equation}\label{eqn:dist-prelim-final}
(\vecs, \aux,\vecd'\cdot \vecs)\equiv (\vecs,\aux, U),
\end{equation}
where $\vecs$ is sampled randomly from $\{0,1\}^n$ conditioned on $\vecC\vecs=\vecv$.
 We note that the distribution of $\vecd'$ may not have min-entropy $\omega(\log \secp)$ (conditioned on $\aux$) since it is generated w.r.t.\ $\widehat{A}$ and not w.r.t.\ $A$, and while $A\approx \widehat{A}$, checking if a distribution has min-entropy cannot be done efficiently. 
 Nevertheless, the distribution $\D_\vecv$ is indistinguishable from having min-entropy  $\omega(\log \secp)$, and thus it suffices to prove \Cref{eqn:dist-prelim-final}.
To this end, for every $\vecC\in\mathbb{Z}_q^{k\times n}$ and $\vecv\in\mathbb{Z}_q^v$ consider the sets \[
\cS=\{\vecx\in\{0,1\}^n~:~ \vecC(\vecx)=\vecv\}~~\mbox{ and }~~\cX=\{0,1\}^n\setminus\{0^n\}
\]
and the hash function 
\[h:\cS\times \cX\rightarrow \{0,1\},
\]
defined by 
\[h(\vecx,\vecd')=\vecx\cdot \vecd' \mod 2.
\]
\Cref{lemma:zviketal} implies that for all but negligible fraction of $\vecC$ and $\vecv$ it holds that $h$ is 2-universal, which together with the leftover hash lemma, implies that \Cref{eqn:dist-prelim-final} holds, as desired.

\qed

\newcommand{\vhatd}{\vecd'}
\newcommand{\fhat}{\hat{f}}
To show our modified version of Lemma~4.6 of~\cite{BCMVV18}, it suffices to show a version of their Lemma~4.9 modified to handle two arbitrary binary strings $\vecd'_1, \vecd'_2$. Once this has been shown, the remaining argument proceeds unchanged. In this section will only present our modified version of Lemma~4.9 and its proof.

First, let us recall some basic properties of the discrete Fourier transform. Define the $q$th root of unity
\[ \omega_q = e^{2\pi i / q}.\]
The ``standard Fourier identity'' is that
\[ \sum_{x \in \mathbb{Z}_q} \omega_q^x = 0.\]
For a function $f: \mathbb{Z}_q^\ell \times \mathbb{Z}_2 \to \mathbb{C}$, the Fourier transform $\hat{f}$ is defined by

\[ \fhat(\vecx, y) = \sum_{\vecv,z} \omega_q^{\vecv \cdot x} (-1)^{y\cdot  z} f(\vecv,z). \]
With this normalization, we have the following version of Plancherel's theorem: \yael{?}
\[ \sqrt{2 q^\ell}  \| f \|_2 = \| \fhat \|_2. \]

Now, in the context of Lemma~4.9 of \cite{BCMVV18}, we are given a random matrix $\vecC \in \mathbb{Z}_{q}^{\ell \times n}$, and arbitrary distinct nonzero binary vectors $\vecd'_1, \vecd'_2 \in \{0,1\}^n$.
Define 
\[ g(\vecv,z_1, z_2) = \Pr_{\vecs \in \{0,1\}^n} [ \vecv = \vecC \vecs, z_1 = \vecd'_1 \cdot \vecs, z_2 = \vecd'_2 \cdot \vecs]. \]
Then, to prove the Lemma, it suffices to show that with high probability over the choice of the matrix $\vecC$, the probability distribution whose density is $g$ is close to the uniform distribution over the space $\mathbb{Z}_q^\ell \times \mathbb{Z}_2 \times \mathbb{Z}_2$. Specifically, denoting by $U$ the uniform distribution  and denoting by $D$ the Total Variation Distance, we wish to show that
\[ D(g, U) \leq q^{\ell/2} \cdot 2^{-n/40}.\]

To do so, we relate the TVD distance to the $L_2$ norm of the difference $g - U$: 
\begin{align*}
D(g, U) &= \frac{1}{2} \|  g - U \|_1  \\
&\leq \frac{1}{2} \sqrt{2q^{\ell} }  \| g - U\|_2 \\
&=\frac{1}{2} \| \hat{g} - \hat{U} \|_2,
\end{align*}
where the second line follows from Cauchy-Schwarz.
Note that for \emph{any} probability density $g$ over $\mathbb{Z}_q^\ell \times \mathbb{Z}_2 \times \mathbb{Z}_2$, we have that $\hat{g}(0^\ell, 0,0) = 1$. This is because
\[ \hat{f}(0^\ell, 0,0) = \sum_{\vecv, z_1, z_2} g(\vecv, z_1, z_2) = 1.\]
Moreover, for the uniform density $U$, we further have $\hat{U}(\vecx, y_1, y_2) = 0$ for all $(\vecx, y_1, y_2) \neq (0^\ell, 0, 0)$. This follows by the standard Fourier identity.

Thus, we get 
\begin{align*}
\frac{1}{2} \| \hat{g} - \hat{U} \|_2 &= \frac{1}{2} \sqrt{ \sum_{\vecx,y_1, y_2} | \hat{g}(\vecx,y_1,y_2) - \hat{U}(\vecx,y_1, y_2) | ^2 }\\ 
&= \frac{1}{2}  \sqrt{\sum_{(\vecx,y_1, y_2) \neq (0^\ell,0)} |\hat{g}(\vecx,y_1, y_2) |^2}.
\end{align*}

To bound this sum, we will now calculate $\hat{g}$, using the identities $(-1)^{y z} = (e^{\pi i})^{yz} = e^{(2\pi i / 2) \cdot  yz}$ to simplify the resulting sums.
\begin{align*}
\hat{g}(\vecx,y_1, y_2) &= \sum_{\vecv,z_1,z_2} \omega_q^{\vecv \cdot \vecx} (-1)^{y_1 z_1 + y_2 z_2} g(\vecv,z_1,z_2) \\
&= \sum_{\vecv,z_1, z_2} e^{2\pi i \cdot (\vecv \cdot \vecx / q + (y_1 z_1 + y_2 z_2) / 2 ) } g(\vecv,y_1, y_2) \\
&= \sum_{\vecv,z_1, z_2} e^{2\pi i \cdot(2\vecv \cdot \vecx  + q(y_1 z_1 + y_2 z_2))/ 2q } g(\vecv,z_1, z_2) \\
&= \sum_{\vecv,z_1, z_2} \omega_{2q} ^{2 \vecv \cdot \vecx + q(y_1 z_1 + y_2 z_2)} g(\vecv,z_1, z_2) \\
&= \E_{\vecs \in \{0,1\}^n} \sum_{\vecv,z_1, z_2} \omega_{2q}^{2 \vecv \cdot \vecx + q (y_1 z_1 + y_2 z_2)} \mathbf{1}[ \vecv = \vecC \vecs, z_1 = \vecd'_1  \cdot \vecs, z_2 = \vecd'_2] \\
&= \E_{\vecs \in \{0,1\}^n} \omega_{2q}^{ 2 \vecx \cdot (\vecC \vecs) + q  ((y_1 \vecd'_1 + y_2 \vecd'_2) \cdot \vecs)}.
\end{align*}

Define $\vecw = 2 \vecC^T \vecx + q(y_1  \vecd'_1 + y_2 \vecd'_2)$
so that $\vecw^T \vecs$ is equal to the exponent in the last line above. Then

\[ \hat{g}(\vecx,y_1, y_2) = \E_{\vecs\in \{0,1\}^n} \omega_{2q}^{\vecw^T \vecs}. \]

Our goal is to show that $\hat{g}$, with the $(0^\ell, 0, 0)$ entry deleted, is small in 2-norm. We are going to do this by bounding the entries individually.

\paragraph{Case 1: $(\vecx,y) = (0^\ell, 1)$.}
In this case, we have
\begin{align*}
    \hat{g}(0^\ell, y_1,y_2) &= \E_{\vecs \in \{0,1\}^n}  \omega_{2q}^{q (y_1 (\vhatd_1)^T + y_2 (\vhatd_2)^T) \vecs} \\
    &= \E_{\vecs \in \{0,1\}^n} (-1)^{ (y_1 (\vhatd_1)^T + y_2 (\vhatd_2)^T) \vecs} \\
    &= 0,
\end{align*}
where in the last line we used the fact that at least one of $y_1, y_2$ is nonzero, and that $\vhatd_1 \neq \vhatd_2$, to say that $y_1 \vhatd_1 + y_2 \vhatd_2$ is a nonzero binary vector.


\paragraph{Case 2: $\vecx \neq 0^\ell$.} 
In this case, we will use the fact that $\vecC$ is a random matrix. Specifically, in Lemma 4.8 of~\cite{BCMVV18} it is shown that with probability $1 - q^{\ell} \cdot 2^{-n/8}$, $\vecC$ is \emph{moderate}. To define this, we start by defining a moderate \emph{scalar}: for $x \in \mathbb{Z}_q$, let its \emph{centered representative} be its unique representative in $(-q/2, q/2]$. Then we say $x$ is \emph{moderate} if its centered representative lies in the range $[-3q/8, -q/8) \cup (q/8, 3q/8]$. This is true for a uniformly random $x \in \mathbb{Z}_q$ with probability $1/2$. A \emph{moderate vector} is one for which at least $1/4$ of the entries are moderate: for a uniformly random vector in $\mathbb{Z}_q^n$, the chance that it is moderate is exponentially close to $1$, by a Chernoff bound. A \emph{moderate matrix} is one for which every nonzero vector in the row span is moderate. 

Now, observe:
\begin{align*}
    \left|\E_{s \in \{0,1\}} \omega_q^{s x}\right| &= \left|\frac{1}{2} ( \omega_q^0 + \omega_q^{x})\right| \\
    &= \left| \frac{1}{2} (\omega_q^{-x/2} + \omega_q^{x/2}) \right| \\
    &= | \cos( \pi x / q) |.
\end{align*}
Thus, if $x$ is moderate, then $|x/q| \in (1/8, 3/8]$, and $|\cos(\pi x/q)| \leq |\cos( \pi /8)|$. So for any moderate vector $\vecr \in \mathbb{Z}_q^n$, it holds that
\[ \left| \E_{\vecs \in \{0,1\}^{n}} \omega_{q}^{\vecr \cdot \vecs} \right| \leq |\cos(\pi/8)|^{n/4}. \]

We will need a slightly refined version of this. Let $r$ be a moderate scalar $\in \mathbb{Z}_q$, and $e_1, e_2 \in \{0,1\}$ be arbitrary. Then
\[ 2r + q(e_1 + e_2) \in [-3q/4, -q/4) \cup [q/4, 3q/4) \cup [q + q/4, q + 3/q4) \cup( q/4, 3q/4] \cup (q + q/4, q + 3q/4] \cup (2q + q/4, 2q + 3/q4].\]
Thus,
\[ |(2r + q(e_1 + e_2))/q| \in [1/4, 3/4] \pmod{1}. \]
Therefore,
\[ \left| \E_{s \in \{0,1\}} \omega_{2q}^{(2r + qe)s} \right| = \left|\cos(\frac{\pi}{2q} (2r + qe))\right| \leq |\cos (\pi/8) |. \]

Thus, if $\vecr$ is moderate and $\vece_1, \vece_2$ are arbitrary binary vectors, then by the same reasoning
\begin{align*}
    \left| \E_{\vecs \in \{0,1\}^n} \omega_{2q}^{(2\vecr  + q(\vece_1 + \vece_2) )\cdot \vecs} \right| \leq |\cos(\pi/8)|^{n/4}.
\end{align*}

Now, to finish the argument, let's return to $\hat{g}$. For any $\vecx \neq 0^\ell$, since $\vecC$ is moderate, we know $\vecx^T \vecC$ is a moderate vector. Thus, we have
\begin{align*}
\hat{g}(\vecx, y_1, y_2) = \E_{\vecs \in \{0,1\}^n} \omega_{2q}^{2 (\vecx^T \vecC) \vecs + q (y_1(\vecd'_1 \cdot \vecs) + y_2(\vecd'_2 \cdot \vecs))}
\end{align*}
Now setting $\vecr = \vecx^T \vecC$ and $\vece_{1,2} = y_{1,2} \vecd'_{1,2}$, we conclude that
\[ |\hat{g}(\vecx, y_1, y_2)| \leq |\cos(\pi/8)|^{n/4}. \]

So in the end we get
\begin{align*}
    D(g, U) &\leq \frac{1}{2} \sqrt{\sum_{(\vecx, y_1, y_2) \neq (0^\ell, 0,0)} | \hat{g}(\vecx, y_1, y_2) |^2} \\
    &\leq \frac{1}{2} \sqrt{\sum_{\vecx \neq 0^\ell} \sum_{(y_1, y_2) \in \{0,1\}^2} |\cos(\pi/8)|^{n/2}} \\
    &= \frac{1}{2} \sqrt{4(q^\ell - 1) |\cos(\pi/8)|^{n/2} }\\
    &\leq q^{\ell/2} \cdot 2^{-n/40}
\end{align*}
This is the desired bound.

\end{document}